\documentclass[12pt]{article}
\usepackage[british]{babel}

\usepackage{microtype}
 
\usepackage{fullpage}
\usepackage{amsmath}
\usepackage{amssymb}
\usepackage{amsthm}
\usepackage{verbatim} 
\usepackage{stmaryrd}
\usepackage{url}
\usepackage{float}
\floatstyle{ruled}
\restylefloat{table}
\restylefloat{figure}
\allowdisplaybreaks
\usepackage{basic}
\usepackage{ambients}
\usepackage{appendix}
\usepackage{paralist}
\usepackage{graphicx}
\usepackage{proof}
\newif\ifdraft\drafttrue
\newcommand{\ntrans}[1]{\mathrel{{\trans{#1}}\makebox[0em][r]{$\not$\hspace{2ex}}}{\!}}
\newcommand{\ttrans}[1]{\stackrel{\, {#1} \,}{\Longrightarrow}}
\newcommand{\ttranst}[1]{\Longrightarrow\trans{#1}\Longrightarrow}
\newcommand{\redtime}{\red_{\sigma}}

\newcommand{\nat}{\mathbb{N}} % Naturals

\newcommand{\rel}{\mathcal R}

\newcommand{\nodep}[4]{#1 {\boldsymbol{[}#2\boldsymbol{]}}^{#3}_{#4}}

\newcommand{\mob}{\mathtt{m}}
\newcommand{\stat}{\mathtt{s}}

\newcommand{\bool}[1]{\llbracket #1 \rrbracket}
\newcommand{\send}[3]{\overline{#1}#2 @ #3}
\newcommand{\rec}[3]{#1 #2@ #3}
\newcommand{\sendobs}[3]{\overline{#1}#2 \triangleright #3}
\newcommand{\recobs}[3]{#1 #2 \triangleright #3}
\newcommand{\conf}[2]{#1 {\Join} #2}

\newcommand{\rsens}[2]{#2?(#1)}
\newcommand{\rsensa}[2]{#2?#1}
\newcommand{\wsensa}[2]{#2!#1}
\newcommand{\wact}[2]{#2!#1}

\newcommand{\mobi}{\mob}

\newtheorem{theorem}{Theorem}
\newtheorem{proposition}{Proposition}
\newtheorem{definition}{Definition}
\newtheorem{example}{Example}
\newtheorem{remark}{Remark}
\newtheorem{lemma}{Lemma}
\newtheorem{corollary}{Corollary}

\usepackage{authblk}

\title{A Semantic Theory of the Internet of Things}

\author[1]{Valentina Castiglioni\thanks{The contribution of this author is limited to an early writing of some of the proofs in the Appendix.}}
\author[1]{Ruggero Lanotte}
\author[2]{Massimo Merro}
\affil[1]{\small Dipartimento di Scienza e Alta Tecnologia, Universit\`a 
dell'Insubria, Como, Italy}
\affil[2]{\small Dipartimento di Informatica, Universit\`a degli Studi di Verona, Italy}
\date{}
\begin{document} 
\maketitle

\begin{abstract}
We propose a process calculus for modelling systems 
in the \emph{Internet of Things} paradigm. Our systems  interact both with the physical environment, via \emph{sensors} and \emph{actuators\/}, and with 
\emph{smart devices}, via short-range and Internet channels. 
The calculus is equipped with a standard notion of \emph{bisimilarity\/} 
which is a  fully abstract characterisation of a well-known contextual 
equivalence. We use our semantic proof-methods  to prove run-time properties
as well as system equalities of non-trivial IoT systems. 
\end{abstract}

\section{Introduction}

In the \emph{Internet of Things\/} (IoT) paradigm,
\emph{smart devices\/}, such as smartphones, 
 automatically collect information from shared resources
(e.g.\ Internet access or physical devices) 
and aggregate them to provide new services to end users~\cite{GBMP13}. 
The ``things'' commonly deployed in  IoT systems are:
\emph{RFID tags\/}, for unique identification,
 \emph{sensors\/}, to detect
 physical changes in the environment, and \emph{actuators\/}, to pass
 information to the environment.
%%%  and \emph{smart devices}, as smartphones
%%and tablets. 
%%Unlike other network paradigms, in the IoT  paradigm,  \emph{smart devices\/}, 

The research on IoT is currently  %%mainly 
focusing on practical applications 
such as the development of enabling technologies~\cite{RK09},  
ad hoc architectures~\cite{PH07},  
semantic web technologies~\cite{DEBM12}, %
and cloud 
computing~\cite{GBMP13}.  %~\cite{GBMP13}. 
However, as  pointed out by Lanese et al.~\cite{LBdF13}, 
there is a lack of research in formal methodologies to model the 
interactions among system components, 
%%(e.g. processes, sensors, and actuators), 
and to verify the correctness of the network deployment before 
its implementation. 
Lanese et al.\ proposed 
  the first process calculus for IoT systems, called \emph{{IoT}-calculus}~\cite{LBdF13}. 
The IoT-calculus captures  the partial topology of communications and the interaction between \emph{sensors\/}, \emph{actuators} and \emph{computing processes} to provide useful services to the \emph{end user\/}.
A behavioural equivalence is then defined to compare IoT systems.

% that might allow to verify the correctness 
%%of IoT applications before their implementation.
 Devising a calculus for modelling a new paradigm requires understanding
 and distilling in a clean algebraic setting the main features of the paradigm. 
The main goal of this paper is to propose a new process calculus
that integrates a number of crucial features not addressed in the IoT-calculus, 
and equipped with a clearly-defined semantic theory 
 for specifying and reasoning  on IoT applications. 
Let us try to figure out what are the main ingredients of IoT, by means of an example.

\begin{figure}[t]
\label{fig:smarthome}
\vspace*{7mm}
\begin{center}
\includegraphics[scale=.80]{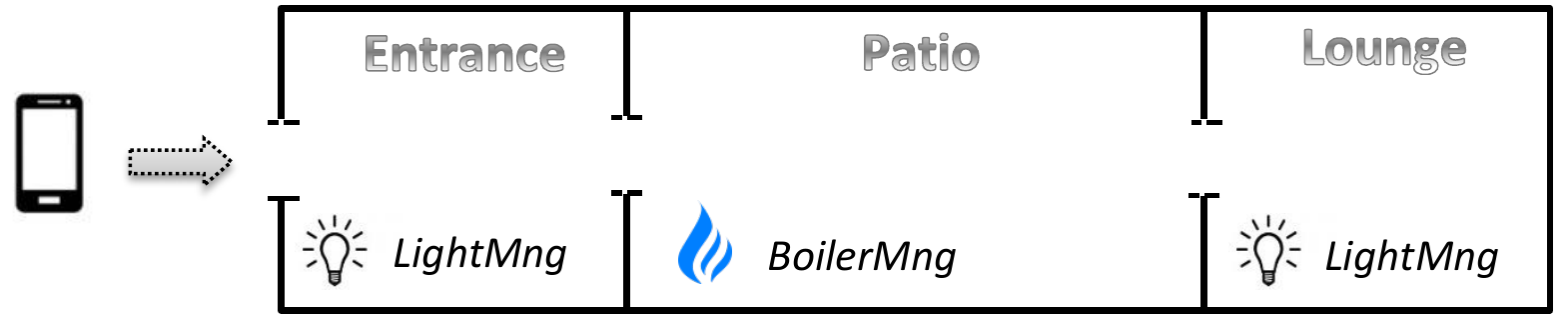}
\end{center}
\vspace*{2mm}
\caption{A simple smart home}
\end{figure}

Suppose a simple \emph{smart home\/}  in which  the user can
use her smartphone to 
 remotely control the heating boiler of her house, 
and automatically turn on  lights when entering  a room. 
In  Fig.~\ref{fig:smarthome}, we draw a small smart home with
 an entrance and a lounge, separated by a patio. 
Entrance and  lounge have their own lights  (actuators) which are 
governed by different light manager processes, $LightMng$. The boiler 
is in the patio and  is governed by a boiler manager
process, $BoilerMng$. This process senses the local temperature
(via a sensor)
and decides whether the boiler should be turned on/off, by
setting a proper actuator to signal  the state of the boiler. 

%%The smartphone can enter the house moving aroung. 
The smartphone executes two processes: $BoilerCtrl$ and 
$LightCtr$. The first one reads user's commands, 
 submitted via the  
phone touchscreen  (a sensor), 
and forward them to  the process $BoilerMng$, 
 via an Internet channel. Whereas, the process $LightCtrl$ 
interacts with the processes $LightMng$, via short-range wireless 
channels (e.g.\ Bluetooth), to automatically turn on lights when the smartphone 
physically enters either the entrance or the lounge. 

The whole system is given by the parallel composition of the 
smartphone (a mobile device) and the smart home (a stationary entity). 

Now, on such kind of systems 
one may wish to prove \emph{run-time properties\/}. 
Think of a \emph{fairness} property saying 
that the boiler will be eventually turned on/off whenever some conditions are satisfied. Or consistency properties, saying the smartphone will never be in two rooms at the same time. 
%%Thus, no two lights can be switched on
%%at the same time (for more than a transitory period). 
%%How will see, in order
%%to ensure this kind of properties we must equip our model with a notion of time. 
Even more, one may be interested in understanding 
whether our system has the same \emph{observable behaviour\/} of another
 system. Think of  a variant of our smart home,
 where lights functionality depends on GPS coordinates of the smartphone (localisation is a  common feature of today smartphones). Intuitively,
 the smartphone might send its GPS position to a centralised light manager, 
${CL}ightMng$ (possibly placed in the patio), via an Internet channel. 
The process ${CL}ightMng$ will then interact with the two
processes $LightMng$, via short-range channels,
  to switch on/off lights, depending on the position of the
 smartphone. Here comes an interesting question: 
Can these two implementations, based on different 
light management mechanisms, 
be actually distinguished by an end
user?

In the paper at hand we %%pursue the work of~\cite{LBdF13} and 
develop a fully abstract semantic theory for
 a  process calculus  of IoT systems, 
called \cname.  
We  provide a formal notion of when  two systems
in \cname\  %% $Sys_1$ and $Sys_2$
 are indistinguishable, 
%%, written $Sys_1 \cong Sys_2$, 
in all possible contexts, from the point of view of the end user.
%% Intuitively, we will write 
%%$Sys_1 \cong Sys_2$ to mean that either system $Sys_1$ or 
%%$Sys_2$ can be replaced by the other in a larger system without changing the observable behaviour of the overall system.
 Formally, we use the
 approach of~\cite{HY95,Sangiorgi-book}, often called \emph{reduction (closed) barbed congruence\/}, which relies 
 on two crucial concepts: a 
\emph{reduction semantics\/}, to describe system computations, and the 
\emph{basic observables\/},  to represent what the environment can directly observe of a system.
%%The only parameter in its definition is the choice of a \emph{primitive observation}
%%or \emph{barb\/}.
%%A challenge in \cname, 
%%is a proper representation of  the observer. 
%%In fact, as our devices include both a physical and a logical component, 
%%the environment  must be able to interact with both of them. 
%% Thus, 
In \cname, there here are at least two possible observables:  
%%%observation predicates for IoT systems.  
 the ability %%of a system 
to transmit along  
channels, \emph{logical observation}, and 
the  capability to diffuse messages via actuators, \emph{physical observation}.
%%In the first case the observer is a smart device  
%%communicating with the system; 
%% in the other, the observer is an end-user observing changes
%%occurring on actuators of the system itself.
We have adopted the second form   as
our contextual equality remains invariant 
when adding logical observation. 
However, the right definition of physical observation is 
far from obvious as it involves some technical challenges 
in the definition of the reduction semantics (see the discussion 
in Sec.~\ref{sec:extensions}).
%%We have adopted the second form of barb as it can encode the other form. However, 
%%taking the physical observation as the primitive one 
%%The choice of the observable involves some technical challenges.

Our calculus is equipped with two  
\emph{labelled transition semantics\/} (LTS) in the SOS style of Plotkin~\cite{Plo04}: an \emph{intensional} semantics and an \emph{extensional} semantics.
The adjective intensional is used to stress the fact that the actions here correspond to activities which can be performed by a system in isolation, without any interaction with the external environment. While, the extensional semantics focuses on those activities which require a contribution of the environment.
Our  extensional LTS builds on the intensional one, by%
introducing specific 
transitions for modelling all interactions which a system can have with its environment. Here, we would like to point out that since 
our basic observation on systems does not involve the recording of the passage of time, this has to be taken into account extensionally in order to gain a proper extensional account of systems.

As a first result we prove that the reduction semantics 
 coincides with the intensional semantics (Harmony Theorem),
 and that they satisfy some desirable time 
properties such as \emph{time determinism\/}, \emph{patience\/}, 
\emph{maximal progress\/} and \emph{well-timedness}~\cite{HR95}.
%%Here, we would like to point out that since 
%% our basic observation on systems does not involve 
%the recording of the passage of time, this has to be taken into account extensionally in order to gain a proper extensional account of systems.
%%\marginpar{MM: ho tolto una frase.}

%%Our \emph{extensional semantics} determines a new LTS, 
%%which in turn gives rise to the standard notion of bisimulation equivalence between systems. 
%%This provides a powerful co-inductive proof technique. 
 %%to show that two systems are contextully equivalent it is sufficient to exhibit a witness bisimulation which contains them. 
However, the main result of the paper is that weak  \emph{bisimilarity\/} in the extensional LTS is \emph{sound} and \emph{complete} with respect to our contextual equivalence, reduction barbed congruence: two systems are related by 
some bisimulation in the extensional LTS if and only if
 they are contextually equivalent.
 This required a non-standard
proof of the congruence theorem (Thm.~\ref{thm:congruence}). 
%%: two systems are related by 
%%some bisimulation in the extensional LTS if and only if
%% they are contextually equivalent.
Finally, in order to show the effectiveness of our
bisimulation proof method,  we prove a number of non-trivial system equalities. 
%%Here, we wish to point out that our %%notion of 
%%bisimilarity %%is first-order and it
%%\marginpar{MM: Rimosso meccanizzazione.}
%%can be easily mechanised by relying on well-known interactive theorem provers such as Isabelle/HOL~\cite{Isabelle} or Coq~\cite{Coq}. 

\paragraph*{Outline} Sec.~\ref{the-calculus} 
contains the calculus together with the reduction 
semantics, the contextual equivalence, and a  discussion on design choices. Sec.~\ref{case_study} gives the details of  our smart home example, 
and provides desirable run-time properties for it. 
Sec.~\ref{lab-sem} defines both  intensional 
and  extensional LTSs. In Sec.~\ref{full-abstraction} we define bisimilarity 
for IoT-systems,
and prove the full abstraction result  together with a number of non-trivial system equalities. 
Sec.~\ref{conc} discusses related  work, and concludes.

%%%%%%%%%%%%%%%%%%%%%%%%%%%%%%%%
%%%%%%                                                                      %%%%%%%
%%%%%%                        A L G E B R A                         %%%%%%%
%%%%%%                                                                      %%%%%%%
%%%%%%%%%%%%%%%%%%%%%%%%%%%%%%%%

\section{The calculus}
\label{the-calculus}

In Tab.~\ref{tab:syntax} we give the syntax of our \emph{Calculus for the Internet of Things\/}, shortly \cname, in a two-level structure: a lower one for \emph{processes\/} and an upper one for \emph{networks\/} of smart devices.
We use letters $n,m$ to denote \emph{nodes/devices\/}, $c,g$ for \emph{channels\/}, $h,k$
%%\marginpar{MM: dove si usa $\lel$?} %%$h,k\in\lel$ 
for (physical) \emph{locations\/}, $s,s'$ for \emph{sensors\/}, $a,a'$ for \emph{actuators\/} and  $x,y,z$ for \emph{variables\/}.
Our \emph{values\/}, ranged over by $v$ and $w$, are constituted by 
basic values, such as booleans and  integers, sensor and actuator values, and 
coordinates of physical locations. 
%%, whereas \emph{values\/}, ranged over by $u$, include also variables. 
%%%\marginpar{MM: Tuples?}

\begin{table}[t]
 \(
      \begin{array}{@{\hspace*{2mm}}rcl@{\hspace*{10mm}}l@{\hspace*{4mm}}rcl@{\hspace*{2mm}}l}
\multicolumn{4}{l}{\textit{Network:}}\\
M,N 
& \Bdf & \zero & \mbox{empty network} \\
& \Bor &  \nodep   n {\conf \I P}{\mu}{h}  & \mbox{node/device}\\
& \Bor & M | N & \mbox{network composition} \\
& \Bor & \res c M  & \mbox{channel restriction}\\[1pt]
\multicolumn{4}{l}{\textit{Process:}}\\
P,Q
 & \Bdf &   \nil                             & \mbox{termination}\\
 & \Bor & \rho.P                            & \mbox{prefixing} \\
 & \Bor & P \newpar Q                & \mbox{parallel composition} \\
 & \Bor &\timeout {\pi.P} {Q}  & \mbox{communication with timeout}\\
 & \Bor & \match b P Q               & \mbox{conditional}\\
 & \Bor & X                                 & \mbox{process variable}\\
 & \Bor & \fix X P                        & \mbox{recursion}
\end{array}
\)
\caption{Syntax}
\label{tab:syntax}
\end{table}

A network $M$ is a pool of \emph{distinct nodes\/} 
running in parallel. 
We write 
$\zero$ to denote the empty network, while  $M | N$
represents  parallel composition. %% of two networks $M$ and $N$. 
In $\res c M$ channel $c$ is private to the nodes of $M$. 
 Each node is a term of the form
\(
\nodep {n} {\conf {\I}{P}}{\mu}{l}
\), 
where $n$ is the  device ID;
 $\I$ is the physical interface of $n$, 
a partial mapping from sensor and actuator names to physical values; 
$P$ is the process modelling the logics of $n$; $l$ is the physical 
location of the device; $\mu \in \{ \stat, \mob \}$ is a tag 
to distinguish between stationary  %%, with tag $\stat$,  
and mobile nodes.

For security reasons, sensors in $\I$ can be read  only by its 
\emph{controller\/} 
 process $P$. Similarly, actuators
in $\I$ can be modified only by $P$. No other  devices can access
the physical interface of $n$. 
 Nodes live in a physical world which can be divided in an enumerable set of physical locations. 
We assume a discrete notion of  \emph{distance} between two 
 locations $h$ and $k$, i.e.\ $\dist h k \in \mathbb{N}$.

In a node $\nodep n{\conf \I P}{\mu}{h}$,  $P$ denotes a timed 
concurrent process which manages both the interaction with the  physical interface $\I$ and channel communication. The communication paradigm 
is \emph{point-to-point} via channels that may have different transmission ranges. 
We assume a global function $\operatorname{rng}()$ that given a channel 
 $c$ returns
an element of $\mathbb{N} \cup \{ -1, \infty\}$. 
%%Two nodes can communicates via a channel $c$ 
%%if they are physically located within the range $\rng c$. 
Thus, a channel $c$ can be used for: i)  
\emph{intra-node communications\/}, if $\rng c = -1$; ii) 
\emph{short-range inter-node communications\/}
 (such as Bluetooth, infrared, etc) if $0 \leq \rng c < \infty $; iii) 
\emph{Internet communications\/}, if $\rng c =\infty$. 

Our processes build on CCS with discrete time~\cite{HR95}. 
%%We add a construct to get the current location of nodes and two operators 
%%to interact with  sensors and actuators. 
We write $\rho.P$, with $ \rho \in \{ \sigma,  @(x),  \rsens x  s ,$ $ \wact v a\}$, to denote intra-node actions. 
The process $\sigma.P$ sleeps for one time unit. 
The process $@(x).P$ gets the current  location of the enclosing node. 
%%In some sense, $@$ can be seen as a special global channel to get GPS coordinates or  Wi-Fi and cell-tower positioning. 
Process $ \rsens x  s.P$ reads a value $v$  from sensor $s$. 
%% and then continues as $P$, where $x$ is instantiated with $v$. 
Process $ \wact v  a.P$ writes the   value $v$ on the actuator $a$. 
%% and then continues as $P$.
%%We assume a restricted form of parallel composition $P \newpar Q$ where no interaction between $P$ and $Q$ is admitted. 
 We write $\timeout {\pi.P} Q$, with $\pi\in\{\OUT{c}{v},\LIN{c}{x}\}$, to denote channel communication with timeout. 
This process 
%%The process $\timeout{\OUT c v . P}Q$ 
can communicate in the current time interval
 and then continues as P; otherwise, 
after one time unit, it evolves into $Q$. 
%%The process $\timeout{\OUT c v . P}Q$ can transmit the value $v$ along channel $c$ in the current time unit and then continues as P; otherwise, 
%%after one time unit, it evolves into $Q$.   
%%The receiving processes $\timeout{\LIN c x. P}Q$ is the obvious counterpart. 
We write $\match b P Q$ for conditional (guard $\bool{b}$ is always 
\emph{decidable\/}). 
In processes of the form $\sigma.Q$ and $\timeout {\pi.P} Q$ the occurrence of $Q$ is said to be \emph{time-guarded}.
The process $\fix X P$ denotes \emph{time-guarded recursion\/}, as all occurrences of the process variable $X$ may only occur time-guarded in $P$.
In processes $\timeout{\LIN c x. P}Q$, $\rsens x s. P$ and ${@(x).P}$ the variable $x$ is said to be bound.
Similarly, in process $\fix X P$ the process variable $X$ is bound.
 In the term $\res c M$ the channel $c$ is bound. 
This gives rise to the standard notions of 
\emph{free/bound (process) variables}, \emph{free/bound channels}, and 
\emph{$\alpha$-conversion}.
%%We identify processes and networks up to $\alpha$-conversion.
A term is said to be \emph{closed} if it does not contain free (process) variables, although it may contain free channels. 
We always work with closed networks: the absence of free variables is preserved at run-time.
We write $T{\subst v x}$ for the substitution of the variable $x$ with the value $v$ in any expression $T$ of our language.
Similarly, $T{\subst P X}$ is the substitution of the process variable $X$ with the process $P$ in $T$.

Actuator names are metavariables for actuators like  
$display@n$ or $alarm@n$. As node names are unique so are actuator names: different nodes have different actuators. 
The sensors embedded in a node can be of two  kinds: \emph{location-dependent\/} 
and  \emph{node-dependent\/}. The first ones sense data at 
 the current location of the node, %% e.g.\ a temperature sensor, 
  whereas the second ones sense data within the 
node, 
independently on the node's location. %%, e.g.\ a touchscreen of a smartphone.  
Thus, node-dependent sensor names are metavariables for sensors like $touchscreen@n$ or $button@n$; whereas a sensor $temp@h$, for external temperature, is a
typical example of location-dependent sensor. Like actuators, 
node-dependent sensor names are unique. This is not the case of 
location-dependent
sensor names  which may appear in different nodes. For 
simplicity, we use the same metavariables for both kind of sensors. 
When  necessary  we will specify the
 type of sensor in use. 

 The syntax given in Tab.~\ref{tab:syntax} is a bit too permissive 
with respect to our intentions.  We could rule out ill-formed networks with a simple 
type system. For the sake of simplicity, we prefer to provide the following definition. 
\begin{definition}%%[Well-formedness]
\label{well_form}
A network $M$ is said to be \emph{well-formed\/} \emph{if}
\begin{itemize}
\item it does not contain two nodes with the same name
\item different nodes have different actuators 
\item different nodes have different node-dependent sensors
\item for each $\nodep n{\conf \I P}{\mu}{h}$ in $M$,  with a prefix $\rsens x s $ (resp.\ $\wact v a$)  in $P$,  $\I(s)$ (resp.\ $\I(a)$) is defined
\item 
for each $\nodep n{\conf \I P}{\mu}{h}$ in $M$  with  $\I(s)$ defined for some location-dependent sensor $s$, it holds that  $\mu=\stat$. 
\end{itemize}
\end{definition}
Last condition implies that location-dependent sensors may be used only in stationary nodes. This restriction will be commented 
in Sec.~\ref{sec:extensions}.
%%The last  condition says that, when used in a node, sensors and actuators
%%must be defined in its physical interface. 
Hereafter, we will always work with \emph{well-formed networks\/}. 
 It is 
easy to show that well-formedness is preserved at runtime. 

To end this section, we report some notational \emph{conventions\/}.
$\prod_{i \in I}M_i$ denotes the parallel composition of all $M_i$, for $i \in I$.
We identify $\prod_{i \in I}M_i = \zero$ and  $\prod_{i \in I}P_i = \nil$, if $I =\emptyset$.
Sometimes we write $\prod_{i}M_i$ when the index set
$I$ is not relevant. 
We write $\timeout{\pi}{\nil}$ instead of $\timeout{\pi.\nil}{\nil}$.
We  write $\pi.P$ as an abbreviation for $\fix{X}\timeout{\pi.P}X$.
For $z\geq 0$, we write $\sigma^{z}.P$ as a shorthand for $\sigma.\sigma. \ldots \sigma.P$, where prefix $\sigma$ appears $z$ times. Finally, 
we write $\res {\tilde c} M$ as an abbreviation for $(\nu {c_1})\ldots (\nu {c_k})M$, 
for $\tilde{c}=c_1, \ldots , c_k$.

\subsection{Reduction semantics}
\label{red_sem}

\begin{table}[t]
\(
 \begin{array}{@{\hspace*{2mm}}ll@{\hspace*{2mm}}l@
{\hspace*{1mm}}l@{\hspace*{0mm}}ll@{\hspace*{0mm}}l@{\hspace*{0mm}}l}
\multicolumn{4}{l}
{\textit{Processes:}} \\[2pt]
%%& P \equiv P &                                                                                                                 %%& \mbox{Struc Proc Refl} 
%%\\
%%& P \equiv Q \textrm{ implies } Q\equiv P &                                                                   %%& \mbox{Struc Proc Sym}
%%\\ 
%%& P \equiv Q \textrm{ and } Q\equiv R \textrm{ implies } P\equiv R &                          %%& \mbox{Struc Proc Trans}
%%\\
& P \newpar \nil \equiv P &                                                                                              %%& \mbox{Struc Nil}
\\
& P \newpar Q \equiv Q \newpar P &                                                                               %%& \mbox{Struc Proc Par Sym}
\\
& (P \newpar Q) \newpar R \equiv P \newpar (Q \newpar R) &                                        %%& \mbox{Struc Proc Assoc}
\\
%%%& P\equiv Q \textrm{ implies } P \newpar R \equiv Q \newpar  R & \textrm{ for all } R   & \mbox{Struc Proc Cxt}\\
& [b]P;Q \equiv P  \textrm{ if } \bool{b}=\true  &                                                           %%& \mbox{Struc Then}
\\
& [b]P;Q \equiv Q  \textrm{ if } \bool{b}=\false    &                                                       %%& \mbox{Struc Else}
\\
& \fix{X}P \equiv P\subst{\fix{X}P}{X} &                                                                        %%& \mbox{Struc Fix}
\\[3pt] 
\multicolumn{4}{l}{\textit{Networks:}} \\[2pt]
%%%& \nodep{n}{\conf \I \nil}{\mu}{h}\equiv \zero &                                                                                    & \mbox~{NodeZero\\
%%%% MASSIMO: la legge sopra e' falsa!
& P\equiv Q \textrm{ implies } \nodep{n}{\conf \I P}{\mu}{h}\equiv\nodep{n}{\conf \I Q}{\mu}{h} & %%& \mbox{Struc Node}
\\
%%& M \equiv M &                                                                                                                                          
%% & \mbox{Struc Net Ref} 
%%\\
%%& M \equiv N \textrm{ implies } N\equiv M &                                                                                             %% & \mbox{Struc Net Sym}
%%\\ 
%%& M \equiv N \textrm{ and } N\equiv O \textrm{ implies } M\equiv O &                                                    %%& \mbox{Struc Net Trans}
%%\\
& M | \zero \equiv M &                                                                                                                               %% & \mbox{Struc Zero}
\\
& M | N \equiv N | M &                                                                                                                            %%    & \mbox{Struc Net Par Sym}
\\
& (M | N) | O \equiv M | (N | O) &                                                                                                         %%      & \mbox{Struc Net Assoc}
\\
&
 \res{c}\zero \equiv \zero
  &&
  %%\mbox{Struct Zero Res}
\\
& \res{c}\res{d}M \equiv \res{d}\res{c}M
   %% &&
   %% \mbox{Struct Res Res}
    \\
  & \res{c}(M|N) \equiv M | \res{c}N \Q \textrm{   if } c 
\textrm{ not in }M
    &&
    %%\mbox{Struct Res Par}\\
%%\\
%%& M\equiv N \textrm{ implies } M | O \equiv N | O & \textrm{ for all } O                                                   & \mbox{Struc Net Cxt}
\end{array}
\)
\caption{Structural congruence}
\label{struc}
\end{table}

%%We give a first intuition on the behaviour of our devices by defining reduction relations on networks. 
%%% NON E' UN'INTUZIONE MA PROPRIO LA SEMANTICA OPERAZIONALE DI RIFERIMENTO!

%%\begin{adjustwidth}{-5mm}{-5mm}
\begin{table}[t]
\small
\( 
\begin{array}{l@{\hspace*{1cm}}l}
\Txiom{pos}
{-}
{\nodep{n}{\conf \I {@(x).P}}{\mu}{h} \redtau \nodep{n}{\conf \I {P \subst h x }}{\mu}{h}}
& 
\Txiom{sensread}
{\I(s)=v}
{\nodep n {\conf \I {\rsens x s.P}}{\mu}{h}  \redtau \nodep n {\conf \I {P\subst v x }}{\mu}{h}}
\\[15pt]
\Txiom{actunchg}
{\I(a)=v }
{\nodep n {\conf \I {\wact v a.P}}{\mu}{h}  \redtau \nodep n {\conf {\I} {P}}{\mu}{h}}
&
\Txiom{actchg}
{\I(a)\neq v \Q \I':=\I[a \mapsto v] }
{\nodep n {\conf \I {\wact v a.P}}{\mu}{h}  \red_{a} \nodep n {\conf {\I'} {P}}{\mu}{h}}
\\[15pt]
\multicolumn{2}{c}
{
\Txiom{loccom}{\rng c = -1}
{\nodep{n}{\conf \I 
{\timeout{\OUT c v .P}{R}
\newpar 
{\timeout{\LIN c x .{Q}}{S}} }}
{\mu}{h}
\;
\redtau
\;
\nodep{n}
{\conf \I 
{P \newpar {Q \subst v x}}}
{\mu}{h}
}
}
\\[15pt]
\multicolumn{2}{c}
{
\Txiom{timestat}
{
\nodep{n}{\conf {\I}
{\prod_{i}\timeout{\pi_i.P_i}Q_i \: \newpar \:
\prod_{j}\sigma.R_j}}
{\stat}{h} \: 
\not \redtau
}
{\nodep{n}{\conf {\I}{
\prod_{i}\timeout{\pi_i.P_i}Q_i \: \newpar \: 
\prod_{j}\sigma.R_j}}{\stat}{h} 
\redtime 
\nodep{n}{\conf \I{
\prod_{i} Q_i\, \newpar \, 
 \prod_{j}R_j}}{\stat}{h}}
}
\\[15pt]
\multicolumn{2}{c}
{
\Txiom{timemob}
{
\nodep{n}{\conf {\I}
{\prod_{i}\timeout{\pi_i.P_i}Q_i \: \newpar \:
\prod_{j}\sigma.R_j}}
{\mob}{h} \: 
\not \redtau \Q\; \dist h k \leq \delta
}
{\nodep{n}{\conf {\I}{
\prod_{i}\timeout{\pi_i.P_i}Q_i \: \newpar \: 
\prod_{j}\sigma.R_j}}{\mob}{h} 
\redtime 
\nodep{n}{\conf \I{
\prod_{i} Q_i\, \newpar \, 
 \prod_{j}R_j}}{\mob}{k}}
}
\\[15pt]
%%\multicolumn{2}{l}
%{
%%\Txiom{glbcom}
%%{\dist h k \leq \rng c}
%{\nodep{n}{\conf \I {\timeout{\OUT c v .P}{P'}\newpar R}}{\mu}{h} \; | \;  \nodep{m}{\conf \I {\timeout{\LIN c x .{Q}}{Q'} \newpar S}}{\mu}{k}  \redtau
%%\nodep{n}{\conf \I P\newpar R}{\mu}{h} \; | \; \nodep{m}{\conf \I {Q \subst v x \newpar S}}{\mu}{k}}
%%}\\[15pt]
\multicolumn{2}{c}
{
\Txiom{glbcom}
{\dist h k \leq \rng c}
{\nodep{n}{\conf \I {\timeout{\OUT c v .P}{R}}}{\mu_1}{h} \; | \;  \nodep{m}{\conf \I {\timeout{\LIN c x .{Q}}{S}}}{\mu_2}{k}  \redtau
\nodep{n}{\conf \I P}{\mu_1}{h} \; | \; \nodep{m}{\conf \I {Q \subst v x}}{\mu_2}{k}}
}\\[15pt]
\multicolumn{2}{l}{
\Txiom{parp}
{\prod_{i}\nodep {n_i} {\conf {\I_i} {P_i}}{\mu_i}{h_i} \red_{\omega} 
\prod_{i}\nodep {n_i} {\conf {\I'_i} {P'_i}}{\mu'_i}{h'_i} \Q 
\omega \in \{ \tau, a \}}
{\prod_{i}\nodep {n_i} {\conf {\I_i} {P_i | Q_i}}{\mu_i}{h_i} \red_{\omega} 
\prod_{i}\nodep {n_i} {\conf {\I'_i} {P'_i| Q_i}}{\mu'_i}{h'_i}}
\Q
\Txiom{parn}
{M \red_{\omega} M' \Q \omega \in \{ \tau, a\}}
{M | N \red_{\omega} M' | N}
}
%%\Txiom{parp}
%%{\nodep n {\conf \I P}{\mu}{h} \red_{\omega} \nodep n {\conf {\I'} {P'}}{\mu'}{h'} \Q 
%%\omega \in \{ \tau, a \}}
%%{\nodep n {\conf \I {P \newpar Q}}{\mu}{h} \red_{\omega} \nodep n {\conf {\I'} {P' \newpar Q}}{\mu'}{h'}}
\\[15pt]
\multicolumn{2}{l}{
\Txiom{timepar}
{M\redtime M' \Q N\redtime N' \Q  M | N \not\!\redtau }
{M | N \redtime M' | N'}
\Q\Q\Q\Q\q
\Txiom{timezero}
{-}
{\zero \redtime \zero}
}
\\[15pt]
\multicolumn{2}{l}{
{
\Txiom{res}{M \red_{\omega} N \Q \omega \in \{\tau , a , \sigma \}}
{\res c M \red_{\omega} \res c N }
}
\Q\Q\Q\Q
\Txiom{struct}
{M \equiv N \q N\red_{\omega} N' \q  \mbox{\small $\omega \in \{\tau , a , \sigma \}$}
\q N'\equiv M'}
{M \red_{\omega} M'}
}
\end{array}
\)
\caption{Reduction semantics}
\label{reduction} 
\end{table}
%%\end{adjustwidth}

The dynamics of the calculus is specified in terms of  \emph{reduction 
relations\/} over networks  described in Tab.~\ref{reduction}. 
As usual in process calculi, a reduction semantics~\cite{Mil91}
 relies on an auxiliary standard  relation,  $\equiv$,
called \emph{structural congruence\/}, which  brings the participants of a potential interaction into contiguous positions.
Formally, structural congruence is defined as the  
 congruence induced by the axioms of Tab.~\ref{struc}.

As \cname\ is a timed calculus, with a discrete notion of time, it will be necessary to distinguish between 
instantaneous reductions, $M\redi N$, and timed reductions, $M\redtime N$.
Relation $\redi$ denotes activities which take place within one time interval, whereas $\redtime$ 
represents the passage of one time unit. Our instantaneous 
reductions are of two kinds: those which involve the change of the values 
associated  to  some actuator $a$, written $\red_{a}$, and the others, written 
$\redtau$. Intuitively, reductions of the form $M \red_{a} N$ denote \emph{watchpoints}
which cannot be ignored by the physical environment 
(in Ex.~\ref{ex:reda}, and more extensively at the 
end of Sec.~\ref{sec:extensions}, we explain why it is important
to distinguish between $\redtau$ and $\red_a$).
Thus, we define the instantaneous reduction relation 
$\redi \deff  \redtau \cup \red_{a}$, for any actuator $a$.
We also define $\red \deff \redtau \cup \redtime$. 
%%while 
%%the whole reduction relation is defined as 
%%$\red \deff \redi \cup \redtime$. 
%%Intuitively, a finite number of instantaneous reductions may be executed 
%%between two timed reductions.  

The first seven rules in Tab.~\ref{reduction} model intra-node activities.  
Rule \rulename{sensread} represents the reading of the current data detected 
at some sensor $s$. Rule \rulename{pos} serves to compute the current position of a node.
Rules \rulename{actunchg} and \rulename{actchg} implement
the writing of some data $v$ on an actuator $a$, distinguishing 
whether the value of the actuator changes or not. 
 Rule \rulename{loccom} models intra-node communications on a local channel
$c$ ($\rng c = -1$).
Rule \rulename{timestat} models the passage of time within a stationary node. 
Notice that all untimed intra-node actions 
%% preempt 
%%the passage of time. 
%%This is a generalisation of a well-known 
%%concept in timed calculi, called \emph{maximal progress}~\cite{HR95}. 
%%In fact in \cname\ those actions
 are considered urgent actions as they must
 occur before the next
timed action. 
As an example,  position detection is a time-dependent
 operation 
which cannot be delayed. Similar argument applies to sensor reading,  
actuator writing and channel communication. 
Rule \rulename{timemob} models the passage of time within a mobile node.
This rule also serves to model \emph{node mobility\/}. 
Mobile nodes can nondeterministically move from one physical location $h$
 to a (possibly different) location $k$, at the end of a time interval. 
Node mobility respects the following time discipline:  
 in one time unit   a node located at $h$ can move to any location $k$
 such that $\dist h k \leq \delta$, for
some fixed  $\delta \in \mathbb{N}$ (obviously, it is possible to have $h=k$ and $\dist h k = 0$). 
For the sake of simplicity, we fix the same constant $\delta$ 
for all nodes of our systems. 
The premises of Rules~\rulename{timestat} and 
\rulename{timemob} ensure that if a node  can perform a timed 
reduction $\redtime$ then 
 the same node cannot perform an instantaneous reduction $\redtau$. Actually, 
due to the syntactic restrictions in the premises of both rules, that node
cannot perform an instantaneous reduction $\red_a$ either. This is 
formalised in Prop.~\ref{prop:maxprog}. 
%%Notice also that if Rule \rulename{stop}  
%%was omitted from the semantics 
%%then a node with some left mobility credits would be obliged to 
%%change location to 
%%allow time passing. 

Rule \rulename{glbcom} models inter-node communication 
along a global channel $c$ ($\rng c \geq 0$). Intuitively, two different nodes can communicate via a
 common channel $c$ if and only if they are within  
the transmission range of $c$. 
%%Rule \rulename{RedPar} is the standard reduction rule for parallel composition of networks.
%%Rule \rulename{RedTZero} has obvious meaning.
Rule \rulename{timepar} is for inter-node time synchronisation; the passage of time is allowed only if all instantaneous reductions 
have already fired.  \emph{Well-timedness}  (Prop.~\ref{prop:welltime}) ensures the absence of infinite 
instantaneous 
traces  which would prevent the passage of time. 
The remaining rules are 
standard. 

We write $\redi^k$ as an shorthand for $k$ consecutive reductions $\redi$;
$\redi^\ast$ is the reflexive and transitive closure of $\redi$. Similar conventions apply to the reduction relation $\red$.

Below we report a few standard time properties which hold in our
 calculus: \emph{time determinism\/}, \emph{maximal progress\/}, 
\emph{patience\/} and \emph{well-timedness\/}. 
\begin{proposition}[Local time determinism]
\label{prop:timed}
If $M\redtime M'$ and $M\redtime M''$, then $M'\equiv \prod_{i \in I}
\nodep{n_i}{\conf {\I_i} {P_i}}{\mu_i}{h_i}$ and  $M''\equiv \prod_{i \in I}
\nodep{n_i}{\conf {\I_i} {P_i}}{\mu_i}{k_i}$, with 
$\dist {h_i} {k_i} \leq 2\delta $, for all $i \in I$. 
\end{proposition}
In its standard formulation, \emph{time determinism} says that a system reaches at most one new state by executing a reduction $\redtime$. However,  by an application of Rule~\rulename{timemob}, our mobile nodes may change location when executing a reduction $\redtime$, thus we have a local variant 
 of time determinism. 

According to~\cite{HR95}, the  \emph{maximal progress}
property says that processes communicate as soon as a possibility of communication arises. 
In our calculus,  we generalise this property saying that 
instantaneous reductions cannot be delayed.
\begin{proposition}[Maximal progress]
\label{prop:maxprog}
If $M\redi M'$, then there is no $M''$ such that $M\redtime M''$.
\end{proposition}
On the other hand, if no instantaneous reductions are possible  then time is free to pass. 
\begin{proposition}[Patience]
\label{prop:patience}
If $M \redi M'$  for no  $M'$, then 
there is $N$ such that $M \redtime N$.
\end{proposition}
Finally, time-guardedness in recursive processes
 allows us to prove that our networks are always well-timed. 
\begin{proposition}[Well-timedness]
\label{prop:welltime}
For any $M$  there is a $ z\in\mathbb{N}$ such that  if
$M \redi^u N$ then $u\leq z$. 
\end{proposition}

\subsection{Behavioural equivalence}
\label{sec:barbed}

In this section we provide a standard notion of contextual equivalence for our systems. 
Our touchstone equivalence is \emph{reduction barbed congruence}~\cite{HY95,MS92}, a standard contextually defined process equivalence. 
Intuitively, two systems are reduction barbed congruent if they have the same \emph{basic observables} in all \emph{contexts\/} and under all possible \emph{computations}. 
%%We recall that our nodes, and thus our systems, have both a logical component, the processes running at nodes, and a physical component, given by the physical interfaces at each node. 

%%As expected, the previous operator applies only if the sensor which we want to update is defined in some interface of some node of $M$, otherwise it leaves $M$ unchanged.

As already pointed out in the Introduction, the notion of reduction barbed congruence relies on two crucial concepts: a \emph{reduction semantics\/} to describe system computations, and the \emph{basic observable} which denotes what the environment can directly observe of a system\footnote{See~\cite{Sangiorgi-book} for a comparison between this approach and the original barbed congruence~\cite{MS92}.}.
So, the question is: What are the “right” observables, or barbs, 
 in our calculus? 
%%However, there is not a standard classification for actions into observable and not observable.
%%We have some choice on what to take as a basic observation, or simply \emph{barb}, in our systems. 
Due to the hybrid nature of our systems we could choose to observe either channel communications -logical observation- as in standard process calculi, or 
the capability to diffuse messages via actuators -physical observation- or both things. Actually, it turns out that in \cname\  logical observations
can be expressed in terms of physical ones (see Sec.~\ref{sec:extensions}  for more details). So, we adopt  as basic observables
the capability to publish messages 
on actuators. 
\begin{definition}[Barbs]
\label{def_barb} 
We write 
$\barb{M}{ \wact v {a@h}  }$ if 
$M\equiv 
\res {\tilde g}{\big( \nodep{n}{\conf \I P}{\mu}{h} \, | \, M' \big) }$ with $\I(a)=v$. 
We write $\redbarb{M}{ \wact v {a@h}  }$ if $M\red^\ast \barb{M'}{ \wact v {a@h}  }$.
\end{definition}

The reader may wonder why our 
barb reports the location and not the node of the actuator. 
We recall that actuator names are unique, so they somehow codify the name of 
their node. The location is then necessary because  
the environment is potentially  aware of its position when observing an actuator: if on Monday at 6.00AM your 
smartphone rings to wake you up, then  you may react differently depending whether you are at home or on holidays in the Bahamas! 
%% (by logical means, 
%%via $@(x).P$, or physical means). 
%
\begin{definition}
A binary relation $\rel$ over networks is
 \emph{barb preserving\/} if $M\;\rel\;N$ and  $\barb{M}{\wact v {a@h}}$
 implies $\redbarb{N}{\wact v {a@h}}$. 
\end{definition}
\begin{definition}
A binary relation $\rel$ over networks  is 
 \emph{reduction closed\/} if whenever $M\;\rel\;N$ the following 
conditions are satisfied: 
\begin{itemize}
\item     $M\red M'$ implies $N
\red^\ast N'$ with $M'\;\rel\;N'$
\item     $M\red_{a} M'$ implies $N
\red^\ast \red_{a} \red^\ast N'$ with $M'\;\rel\;N'$. 
%%
%%\item   $M\;\rel\;N$ and $M\redtime M'$ implies
%% $N (\redtau)^\ast \redtime (\redtau)^\ast N'$ and $M'\;\rel\;N'$;
\end{itemize}
\end{definition}
We require  reduction closure of both $\red$ and $\red_a$, for any $a$
(for understanding this choice, please see Ex.~\ref{ex:reda}). 

In order to model sensor updates  made by the physical environment on a sensor $s$ in a given location $h$, 
we define an operator $[s@h \mapsto v]$ on networks.  
%%of a system $M$ 
%%we use an  operator $M [sensor \mapsto value]$ to denote the updating of a
%% sensor. 
\begin{definition}
\label{def:sens-change}
Given  a location $h$, a sensor $s$, and a value $v$
in the domain of $s$, we define:\\[1pt]
\(
\begin{array}{r@{\hspace*{-0.1mm}}c@{\hspace*{-0.1mm}}l}
\nodep n {\conf \I P}{\mu}{h}[s@h \mapsto v] & {\deff} & \nodep n
{\conf {\I[s \mapsto v]} P}{\mu}{h} \mbox{, if $\I(s)$ defined}\\
\nodep n {\conf \I P}{\mu}{k}[s@h \mapsto v] & {\deff} & \nodep n
{\conf \I P}{\mu}{k} \mbox{, if $\I(s)$ undef.\ or $h\neq k$}\\
 (M|N)[s@h \mapsto v]& {\deff} & M[s@h \mapsto v] \; | N[s@h \mapsto v] \\
\big( \res c M \big)[s@h \mapsto v] & {\deff} & \res c {\big( M [s@h \mapsto v] 
\big)}\\
\zero[s@h \mapsto v] & {\deff} & \zero \enspace . 
\end{array}
\)
\end{definition}
As for barbs, the reader may wonder why when updating a sensor 
we use its location, also for node-dependent sensors. This is because 
when changing a node-dependent sensor (e.g.\ touching a touchscreen of a smartphone) 
the environment is in general  aware of its position. 
\begin{definition}
\label{def-contextual}
A binary relation $\rel$ over networks is
 \emph{contextual\/} if $M\;\rel\;N$  implies that 
\begin{itemize}
\item for all networks
$O$, $M|O \; \RRr \; N|O$
\item for all channels $c$, $\res c M \; \RRr \; \res c N$
\item  for all  sensors
$s$, locations $h$, and  values $v$ in the domain of $s$,  
$M[s@h \mapsto v]  \RRr  N[s@h\mapsto v]$.  
%%\item for all node-based sensors $s$, nodes $n$, and  values $v{\in}Dom(s)$, 
%%%$\mbox{\small $M[s@n \mapsto v] \; \RRr \; N[s@n\mapsto v]$}$. 
\end{itemize}
\end{definition}
The first two clauses requires closure under \emph{logical 
contexts} (parallel systems), while the last clause regards 
\emph{physical contexts\/}, which can \emph{nondeterministically update} sensor values. 
%%Notice that substitutions  $ M[s \mapsto v]$ and  $ N[s\mapsto v]$ 
%%associate the same value $v$ to the sensor $s$ in both $M$ and $N$. 
%%Well-formedness (Def.~\ref{well_form}) guarantees that if $s$ is dependent from a physical location $h$, then $s$ can be embedded only in nodes located at $h$, both in $M$ and in $N$. In this manner, sensor $s$ coherently assumes the same value in the two networks.

Finally, everything is in place to define our touchstone behavioural equality. 
\begin{definition}
\label{barbed_cong}
\emph{Reduction barbed congruence\/}, written  $\cong$, is the largest symmetric relation
over networks which is reduction closed, barb preserving  and contextual.
\end{definition}

\begin{remark}
\label{rem:nondet}
Obviously, if $M \cong N$ 
then $M$ and $N$ will be still equivalent in any setting where    
 sensor updates %, implemented via rule~\rulename{SensEnv},
are governed 
by specific physical laws. This is because the  physical contexts that 
can affect sensor values, according to some physical law, are definitely 
fewer than those which can change sensors nondeterministically. 
\end{remark}
We recall that the reduction relation $\red$ ignores the passage of time, and therefore the reader might suspect that our reduction
barbed congruence is impervious to the precise timing of activities. We will show 
that this is not the case.
\begin{example}
Let $M$ and $N$ be two networks such that 
$M = \nodep{n}{\conf \emptyset {\sigma.\timeout{\OUT{c}{}}\nil}}{\stat}{h}$ and 
$N =\nodep{n}{\conf \emptyset {\timeout{\OUT{c}{}}.\nil}}{\stat}{h}$, with
$\rng c = \infty$. 
It is easy to see that $M \redtime N$.  
As the reduction relation $\red$  does not distinguish instantaneous reductions from  timed ones, one may think that networks $M$ and $N$ are reduction barbed congruent, and that a prompt transmission along channel $c$ is equivalent 
to the same transmission delayed of one time unit. 
However,  let us consider the test
$T = \nodep{test}{\conf {\J} {\sigma.{\wact 1 a}.\timeout{\LIN{c}{}.\wact 0 a}{\nil}}}{\stat}{l}$, 
with $\J(a)=0$, for some (fresh) actuator $a$. 
%%Notice that the presence of a barb $\downarrow_{\wact 0 {a@test}}$  means that no timed reduction has occurred.
%%Analogously, the presence of a  barb at $\downarrow_{\wact 1 {a@test}}$ 
%%together with the absence of a barb at $f$ 
%%denotes that exactly one timed reduction has occurred (followed by a reduction $\red_a$).
Our claim is that test $T$ can distinguish the two networks, and thus  $M\not\cong N$.
In fact, if  $M | T \red \red_a O =  \nodep{n}{\conf \emptyset {\timeout{\OUT{c}{}}\nil}}{\stat}{h} | \nodep{test}{\conf {\J'} {\timeout{\LIN{c}{}.\wact 0 a}{\nil}}}{\stat}{l}$, with $\J'(a)=1$,
then there is no $O'$ such that $N|T \redmany \red_a \redmany O'$ with 
$O \cong O'$. This is because $O$ can perform a reduction sequence $\red\red_a$ that 
cannot be matched by any $O'$. 
\end{example}
Behind this example there is the general principle that reduction barbed congruence is sensitive
to the passage of time. 
\begin{proposition}
\label{prop:time-observation} If $M \cong N$ and $M \redtime M'$ then there is $N'$ such that $N \redtau^{\ast}\redtime\redtau^{\ast}N'$ and $M'\cong N'$.
\end{proposition}
\begin{proof}
Suppose  $M\redtime M'$.
Consider the test 
$
T = \nodep{n}{\conf{\J}{\sigma. \wact 1 a . \wact 0 a . \nil} } {\stat}{k}
$ 
such that both systems $M | T$ and $N|T$ are well-formed, and $\J(a)=0$.
By construction, the  presence of a barb $\Downarrow_{\wact 1 {a@k}}$  in a derivative
of one of those systems implies that exactly one timed reduction $\redtime$ has been inferred in the derivation.

Since $M\redtime M'$ it follows  that 
$M|T \redtime\red_a  M'|T'$, 
with $T'=\nodep{n}{\conf{{\J[a \mapsto 1]}}{\wact 0 a}.\nil}{\stat}{k}$
and $\barb{M'|T'}{\wact 1 {a@k}}$.
As $M\cong N$ and  $\cong$ is contextual, the reduction sequence above 
must be mimicked  by $N|T$, that is  $N|T\redmany\red_a \redmany \hat{N}$, with $ M'|T' \cong \hat{N}$.
As a consequence, $\redbarb{\hat{N} }{\wact 1 {a@k}}$. This implies that  
exactly one timed reduction has been inferred in the 
reduction sequence $N|T\redmany\red_a \redmany \hat{N}$. 
As $M|T$ and $N|T$ are well-formed networks, the actuator $a$ can appear 
neither in $M$ nor in $N$. So, the above reduction sequence can 
be decomposed as follows:
\[
N | T \redmany N' |T \red_a N' | T' \redmany N'' | T' = \hat{N}
\]
with $N \redtau^{\ast}\redtime \redtau^{\ast}N''$. 
From $M'|T'\cong N''|T'$ it is easy to derive  $M' \cong N''$ (for details see Lem.~\ref{lem:cut_cxt} in Sec.~\ref{full-abstraction}).
\end{proof}

Now, we provide some insights into the design decision 
of having two different  reduction relations $\redtau$ and 
 $\red_a$. 
\begin{example}
\label{ex:reda}
Let $M$ and $N$ be two networks such that 
$M =  \nodep   n {\conf \I {{\wact 1 a}| {\wact 0 a}.{\wact 1 a}}}{\mu}{h}$ 
and
$N =  \nodep n {\conf {\I} {\wact 1 a}.{\wact 0 a}.{\wact 1 a}}{\mu}{h}$, 
with  $\I(a)=0$  and undefined otherwise. 
Then, within one time unit,  $M$ may display on the actuator $a$ either the sequence of values $01$ or the 
sequence $0101$,  while $N$ can only display the sequence $0101$. 
As a consequence, from the point of view of the physical environment, 
the observable behaviours of $M$ and $N$ are 
clearly different. In the following we show how 
$\cong$ can observe that difference. 
We recall  that the relation $\cong$ is reduction closed. 
Now, if $M \redtau \red_a M'= \nodep   n {\conf \J {{\wact 1 a}}}{\mu}{h}$, with 
$\J(a)=1$, 
the only possible reply of $N$ respecting reduction closure  is $N \red^\ast \red_a N'= \nodep n {\conf {\J} {\wact 0 a}.{\wact 1 a}}{\mu}{h}$. However, it is 
evident that $M' \not \cong N'$ because  $N'$ can turn the actuator $a$ to $0$ while $M'$ cannot. 
Thus, $M \not \cong N$. 
%%Actually, reduction closure with respect to $\red_a$ 
%%gives to  $\cong$ 
%%enough distinguishing power 
%%to guarantee strong preservation of barbs. 
\end{example}
Notice  that 
 if the  relation $\red_a$ was merged together with $\redtau$
then  in the previous example we would have $M \cong N$. In fact,  
if we would merge the two reduction relations then 
the capability 
to observe messages on actuators, given by the barb,  
 would not be enough to observe changes on actuators within one time unit. 
On the other hand, the decision of not including $\red_a$ as part of  $\red$ 
gives to  $\cong$ 
enough distinguishing power 
to observe strong preservation of barbs. 
\begin{proposition}
\label{prop:strong-barbs}
If $M \cong N$ and $\barb{M}{\wact v {a@h}}$ then $\barb{N}{\wact v {a@h}}$. 
\end{proposition}
\begin{proof}
We recall that %%we defined $\redbarb{M}{ \wact v {a@h}  }$ if $M\red^\ast \barb{M'}{ \wact v {a@h}  }$, with
 $\red \deff \redtau \cup \redtime$.  
Let us suppose that $\barb{M}{\wact v {a@h}}$.
As $\cong$ is barb preserving it follows that $\redbarb{N}{\wact v {a@h}}$, namely,  $N\redmany \barb{N'}{\wsensa v {a@h}}$, for some $N'$. We note that
both  reduction relations  $\redtau$ and $\redtime$ do  not modify actuator values. 
As a consequence, this holds also for $\red$. Thus, $\barb{N}{\wact v {a@h}}$.
\end{proof}

%The proof of this result relies on Prop.~\ref{prop:strong-obs}(1) and 
%Theorem.~\ref{thm:full-abstraction}.
%%\marginpar{MM: fila il discorso?}
%%A weak preservation of barb $\downarrow_{\wact v a}$ would have made these two systems reduction barbed congruent.
%%\end{example}

\subsection{Design choices}
\label{sec:extensions}
In this section we provide some insights into the design choices adopted in
\cname. The main goal of \cname\ is to provide a simple calculus to 
deal with the programming paradigm of IoT systems. Thus, for instance, 
\cname\ is a value-passing rather than a  name-passing calculus, as the $\pi$-calculus~\cite{Sangiorgi-book}. However, the theory of \cname\ can be easily adapted to deal with the 
transmission of channel names at the price of adding the standard burden of 
scope-extrusion of names. Furthermore, as both  actuators and sensors can  only be 
managed inside their nodes, it would make little sense to transmit their names along 
channels. 

\cname\ is a timed process calculus with a discrete notion of time. The time model we adopt in \cname\ is known as the fictitious clock approach (see e.g.~\cite{HR95}): a global clock is supposed
to be updated whenever all nodes agree on this, by globally synchronising on a special timing action $\sigma$. Thus, time synchronisation relies on some clock
synchronisation protocol for mobile wireless systems~\cite{SBK05}.  
 However, our notion of time interval is 
 different from that adopted in synchronous languages~\cite{Esterel, Amadio,Boussinot} where the environment injects events at the start of instants and collect 
them at the end. In the synchronous approach, events happening during a time 
interval are not ordered while in our calculus we want to maintain the causality 
among actions, typical 
 of process calculi. 

In  \emph{cyber-physical systems\/}~\cite{van2000introduction},  
 sensor changes 
are modelled  either using continuous models
 (differential equations) or through discrete models 
(difference equations)\footnote{Difference
 equations relate to differential equations as 
discrete math relate to continuous math.}. %%%As in~\cite{LBdF13}, 
However, in this paper we aim at providing a  
behavioural semantics for IoT applications from the
point of the view of the end user. And  
the end user  cannot directly observe changes on the 
sensors of an IoT application: she can only observe 
the effects of those changes via actuators and communication channels.
 Thus, in \cname\ we do not represent sensor changes 
via specific models, but  we rather abstract on them by 
supporting 
 \emph{nondeterministic sensor updates\/} (see Def.~\ref{def:sens-change}
 and~\ref{def-contextual}). Actually, as seen in Rem.~\ref{rem:nondet}, 
 behavioural equalities derived in 
our setting remains valid when adopting any specific model for
sensor updates.

%%Our notion of time interval is 
%% different from that adopted in synchronous languages~\cite{Boussinot} where the environment injects events at the start of instants and collect 
%%them at the end. In the synchronous approach, events happening during a time 
%%interval are not ordered, while in our calculus we want to maintain the causality 
%%among actions, typical 
%% of process calculi. Nevertheless, we decided to allow node mobility only at the end of time intervals because logical operations, 
%%such as channel communications, can be considered significantly faster than 
%%physical movements.

Another design decision in our language regards the possibility 
to change  the value associated to sensors and actuators more than once  
within the same time interval. At first sight this choice  may appear weird  as 
   certain actuators are physical devices that may require some time to turn
   on. On the other hand, other actuators, such  as  lights or displays,
 may have very quick 
reactions. A similar argument applies to sensors.  In this respect 
our calculus does not enforce a synchronisation of physical events 
 as for logical signals in synchronous languages.  In fact, 
actuator changes are under  nodes' control: it is the process running 
inside a node that decides when changing the value exposed on an actuator of that node. Thus, if the actuator of a node models a
 slow device then it is under the responsibility of the process running at that 
node to change the actuator with a proper delay. Similarly, sensors should  be read only when this makes sense. For instance, a temperature
sensor should be read only when the temperature is supposed to be stable.

Let us now discuss on node mobility. 
The reader may  wonder why \cname\ does not provide a process to drive
node mobility, as in Mobile Ambients~\cite{Ambients}. Notice that, unlike Mobile Ambients, 
our nodes do not represent mobile computations within an Internet domain. 
Instead, they represent smart devices  which 
do not decide where to move to: an external agent  moves them. 
We also decided to allow node mobility only at the end of time intervals. 
This is because both intra-node and inter-node logical operations, 
such as channel communications, can be considered significantly faster than 
physical movements of devices.
%%Indeed, with present technologies the communication times are
%%negligible with respect to times for mechanical movements. 
For instance, consider a transmitter that moves at 20 m/s and
that transmits a 2000-byte frame over a channel having a 2 megabit/s bandwidth. The actual transmission would take about
0.008 s; during that time, the transmitter moves only about 16 cm away. In other words, we can assume that the nodes
are stationary when transmitting and receiving, and may change their location only while they are idle.
 However, to avoid uncontrolled movements
of nodes 
we decided to fix for all of them the same bound $\delta$, representing the  
maximum distance a node can travel within one time unit. There 
would not be problems in allowing different $\delta$ for different nodes. 
 Finally, for the sake of simplicity, 
in the last constraint of Def.~\ref{well_form} we impose that location-dependent
sensors can only occur in stationary nodes. This allows us to have a 
local, rather than a global, representation of those sensors. Notice 
that  mobile location-dependent sensors would have the same technical 
challenges of  \emph{mobile wireless sensor networks}~\cite{MWSN}.

Another issue is about a proper representation of network topology. 
A tree-structure topology, as in Mobile Ambients, 
would be desirable to impose that a device cannot be 
in two mutually exclusive places at the same time. This desirable property 
cannot be expressed
in~\cite{LBdF13}, where links between nodes can be added and 
removed nondeterministically. 
However, a tree-structured topology
would imply an higher-order bisimulation (for details see~\cite{MeZa05});  %,MeHe06}).
while in the current paper we look for a simple (first-order) 
bisimulation proof-technique which could be easily mechanised. 

Finally, we would like to explain our choice about barbs. As already
said in the previous section there are other possible 
definitions of barb.  For instance, one could 
choose to observe the capability to transmit along a channel $c$, by defining 
$\barb{M}{\out c @ h}$ if 
$M\equiv \res {\tilde{g}}
{\big( \nodep{n}{\conf \I {\timeout{\OUT c v .P}P'\newpar Q}}{\mu}{k} \, | \, M' \big)}$, with $c \not \in {\tilde {g}}$ and $\dist h k \leq\rng c $. 
%%\marginpar{in S  invece di mu usare stat}
%% Massimo: Perche'??
However, if you consider the system $S= \res{c} (M | \nodep{m}{\conf \J {\timeout{\LIN{c}{x}.{{\wact 1 a}}}{\nil}}}{\mu}{h})$, with $\J(a)=0$, for some 
appropriate $m$,  then it is easy to show that $\barb{M}{\out c @ h}$ if and only if 
$S\red \red_a S' \downarrow_{a@h!1}$. Thus, the barb on channels can always  be 
reformulated in terms of our barb. The vice versa is not 
possible. 
The reader may also wonder whether it is possible to 
turn the reduction $\red_a$ into $\redtau$ by introducing, at the same time, 
some special barb which would be capable to observe actuators changes. For instance, 
something like $\barb{M}{ \wact {v.w} {a@h}  }$ if 
$M\equiv 
\res {\tilde g}{\big( \nodep{n}{\conf \I {{\wact w a}.P|Q}}{\mu}{h} \, | \, M' \big) }$, with $\I(a)=v$ and $v \neq w$. 
It should be easy to see that this extra barb would 
not help in distinguishing the terms proposed in Ex.~\ref{ex:reda}.
Actually, here there is something deeper that needs to be spelled out. In process calculi,  the term $\beta$ of a barb $\downarrow_{\beta}$ is a 
concise encoding of a context $C_{\beta}$ expressible in the calculus and capable to observe
the barb $\downarrow_{\beta}$. However, our barb $\downarrow_{\wact v {a@h}}$ does not 
have such a corresponding \emph{physical context} in our language. For instance, 
in \cname\
we do not represent the ``eyes of a person'' looking at the values appearing to some 
display. 
Technically speaking, we don't have terms of the form $a?(x).P$ that could be 
used by the physical environment to read values on the actuator $a$. 
This is because such terms would  not be part of an IoT system. 
 The lack of this physical, together with the persistent nature of actuators' state,  explains  why our barb $\downarrow_{\wact v {a@h}}$ 
must work together with 
the reduction relation $\red_a$ to provide the desired distinguishing 
power of $\cong$. 
%%}

%%%%%%%%%%%%%%%%%%%%%%%%%%%%%%%%%%%%%%%%%%%%%%%%%%%%%%%%%%%%%%%%%%%%
\section{Case study: a smart home}
\label{case_study}
\begin{table}[t]
%\begin{center}
%%\includegraphics[scale=0.75]{smart_home_crop.pdf}
%\end{center}
%\vspace*{-3mm}
\[
\begin{array}{@{\hspace*{2mm}}lr@{\hspace*{2mm}}c@{\hspace*{2mm}}l}
& 
Sys  & \deff  & Phone \q \big| \q Home \\
&  Phone & \deff & \nodep{n_P}{\conf {\I_P}{BoilerCtrl \newpar LightCtrl}}{\mob}{out}  \\
&  Home & \deff & LightMng1 \q \big| \q LightMng2  %%%\q \big| \q  LightCorr 
\q \big|  \q BoilerMng\\
& LightMng1 & \deff & \nodep{n_{1}}{\conf {\I_{1}}{L_1}}{\stat}{loc1} \\
& LightMng2 & \deff & \nodep{n_{2}}{\conf {\I_{2}}{L_2}}{\stat}{loc4}\\
%%%&  LightCorr &\deff & \nodep{n_0}{\conf {\I_0}{L_0}}{\stat}{loc3}\\
& BoilerMng & \deff &\nodep{n_{B}}{\conf {\I_{B}}{Auto}}{\stat}{loc2}
\\[5pt]
& BoilerCtrl & \deff  &\fix{X} \rsens{z}{mode}.
\timeout{\OUT {b}{z}.\sigma.X}{X} \\
&  LightCtrl  & \deff & 
\prod_{j=1}^{2}\fix{X}\timeout{\OUT{c_j}{}.\sigma.X}{X} \\
& 
 L_j & \deff & \fix{X}\timeout{\LIN{c_j}{}.\wact{\mathsf{on}}{light_j}.\sigma.X}\wact{\mathsf{off}}{light_j}.X \q \textrm{ for } j\in \{ 1, 2 \}\\
& 
Auto & \deff & \fix{X}\timeout{\LIN{b}{x}.[x=\mathsf{man}]\,\wact{\mathsf{on}}{boiler}.\sigma.Manual;TempCtrl} TempCtrl
\\
& Manual  & \deff & \fix{Y}\LIN{b}{y}.[y=\mathsf{auto}]X;\sigma.Y \\
& 
TempCtrl  & \deff & \rsens t {temp}.[t < \Theta]\,\wact{\mathsf{on}}{boiler}.\sigma.X;\wact{\mathsf{off}}{boiler}.\sigma.X
\end{array}
\]
%\end{math}
%\end{center}
\caption{A smart home in \cname}
\label{case_study_home}
\label{fig_home}
\end{table}

In this section we model the simple smart home discussed in the Introduction, 
and represented in Fig.~\ref{fig:smarthome}.
%%The house consists of an entrance (Room 1) and a lounge (Room2), %, spanning to $5$ locations, 
%%\enlargethispage{\baselineskip}
%%separated by a patio ). 
Our house spans over $4$
contiguous physical locations $loci$, for $i = [1..4]$,  such that  
 $\dist {loci} {locj}= | i - j|$. The entrance (also called Room1) is 
in $loc1$, the patio spans from $loc2$ to $loc3$ and the lounge (also 
called Room2) is in $loc4$.
The house can only be accessed via its entrance, i.e.\ Room1.

Our system $Sys$ consists of the 
smartphone, $Phone$, and the smart home, $Home$. 
%%The latter consists of 
%%three light managers, for the two rooms and the corridor,  and one boiler manager. 
The smartphone  is represented as a mobile node, with $\delta=1$, 
initially placed 
outside the house: $out\neq locj$, for $j \in [1..4]$. 
As the phone can only access the house from its entrance, and 
 $\delta=1$, 
 we have $\dist l {loci} \geq i$, for any $l \not \in \{ loc1, loc2, loc3 , loc4\}$ and $i \in [1..4]$.
Its interface $\I_{P}$ contains only one sensor, called $mode$, representing the touchscreen to control the boiler. 
%%Hence, the support of the interface $\I_P$ is given by $Supp(\I_P)=\{touch\}$.
%%% Mai parlato di supporto fino adesso e non serve qui.
This is a node-dependent sensor. 
%% which is regularly checked by the
% process $BoilerCtrl$. 
The process $BoilerCtrl$ reads $mode$ and 
forwards its value to the boiler manager, $BoilerMng$,  via the
 Internet channel $b$ ($\rng b = \infty$).  
The domain of the sensor $mode$ is $\{\mathsf{man}, \mathsf{auto}\}$, where $\mathsf{man}$ stands for manual and $\mathsf{auto}$ for automatic;  
initially, $\I_P(mode)=\mathsf{auto}$. 

In  $Phone$  there is a second process, called 
 $LightCtrl$, which allows the smartphone to switch on lights \emph{only when} 
getting in touch with the light managers installed in the rooms. 
%%and the corridor. 
%%:
%%\[
%%LightCtrl = \fix{X}\OUT{r}{}.\sigma.X \newpar  \fix{X} \OUT{c}{}.\sigma .X \enspace . 
%%\]
Here channels $c_1$ and $c_2$ serve to control the lights of Room 1 and 2, 
respectively; %%whereas channel $c_0$ controls the light of the Corridor; 
these are short-range channels: $\rng {c_1}=\rng{c_2}=0$. 
%% and
%% $\rng {c_0}=1$.  
The light managers %%of the two rooms and the corridor 
are 
 $LightMng1$,  $LightMng2$, %% and $LightCorr$, 
 respectively.
 These are stationary nodes
%%located at  $loc1$,  $loc5$ and $loc3$, respectively, and 
running 
 the  processes  $L_1$ and $L_2$ %% and $L_0$
 to manage the corresponding 
 lights via the actuators $light_j$, for $j\in\{1,2\}$. 
The domain of these actuators is $\{ \mathsf{on} , \mathsf{off}\}$; 
initially, 
$\I_j(light_j)=\mathsf{off}$, for $j \in \{1,2\}$.
%\[
%L_0 = \fix{X}\timeout{\LIN{d}{}.\wact{\on}{light_0}.\sigma.X}\wact{\off}{light_0}.X
%\]
%%\enlargethispage{\baselineskip}
%%%\begin{remark}
%%\label{rem_casestudy}
%%As $\rng {c_1}{=}\rng{c_2}{=}0$, $\delta{=}1$, and  $\dist {loc1} {loc4} {=} 3$,
%% the  phone cannot act on the %
%%lights of the two rooms at the same time, manifesting a kind of 
%%``ubiquity''. 
%%This kind of  undesired behaviour is
%%admissible  
%%in the calculus of~\cite{LBdF13}. 
%%, as 
%% connections between nodes can be nondeterministically established 
%%without any control. 
%%Basically, in the process calculus proposed 
%%of~\cite{LBdF13} there is no way 
%% to prescribe that a device 
%%can be either in a place or in another,  but never in both. 
%%\end{remark}

Let us describe  the behaviour of the boiler manager $BoilerMng$ 
 in node
$n_{B}$. Here, the physical interface $\I_{B}$ contains a sensor $temp$ and
an actuator $boiler$; $temp$ is a location-dependent temperature sensor, whose domain is $\nat$, and $boiler$ is an actuator to display 
boiler functionality, whose domain is $\{ \mathsf{on},\mathsf{off} \}$. 
Processes $Auto$ and $Manual$ model the two  boiler modalities. %%Initially the boiler is set up in automatic mode. 
In $Auto$ mode sensor $temp$ is periodically checked: if the temperature is under a threshold $\Theta$ then the boiler will be switched on, otherwise it 
 will be switched off.
Conversely,  in manual mode, the boiler is always switched on.
Initially, $\I_{B}(temp)=\Theta$ and $\I_{B}(boiler)=\mathsf{off}$.

Our system $Sys$ enjoys a number of desirable  \emph{run-time properties\/}. 
For instance, if the boiler is in manual mode or its temperature is under
the threshold $\Theta$ then the boiler will get switched on, 
within one time unit. Conversely, if the boiler is in automatic mode 
and its temperature is higher than or equal to 
the threshold $\Theta$, then the boiler will get switched off
within one time unit. These two \emph{fairness} properties 
can be easily proved because our calculus is well-timed.
In general, similar properties cannot be expressed  in untimed calculi 
like that in~\cite{LBdF13}. 
 Finally, our last property states
the  phone cannot act on the 
lights of the two rooms at the same time, manifesting a kind of 
``ubiquity''. 
%%Notice that 
Again, this undesired behaviour is
admissible  
in the calculus of~\cite{LBdF13}. 
%% what informally said in the previous remark: 
%%a node cannot be at the same time in two different locations. In 
%%particular it  cannot %% interact at the same time with the light
%%managers of the two rooms. In particular, 
%%switch on the lights of the two rooms at the same time. 
For the sake of simplicity, in the following proposition
we omit location names both 
in barbs and in sensor updates,  writing $\downarrow_{\wact v a}$ instead of 
$\downarrow_{\wact v {a@h}}$, and $[s \mapsto v]$ instead of $[s@h \mapsto v]$.
The system $Sys'$ denotes an arbitrary (stable) derivative of $Sys$. 
\begin{proposition}
\label{caseproperty}
 Let  $Sys \; (\redi^{\ast}\redtime)^{\ast} \; Sys'$, for some $Sys'$. 
%%\vspace*{-1.5mm}
\begin{itemize}
\item  
If ${Sys'}[mode{\mapsto}\mathsf{man}]
\redi^{\ast}Sys'' \redtime$ %%for some $Sys''$, 
then $\barb{Sys''}{\wact {\mathsf{on}} {boiler}}$ 
\item If 
$Sys'[temp \mapsto t] \redi^{\ast}Sys'' \redtime$, 
with $t < \Theta$,   %%for some $Sys''$,
%%for  some $Sys''$ and some  
then $\barb{Sys''}{\wact {\mathsf{on}} {boiler}}$  
\item If 
$Sys'[ temp \mapsto t] {\redi^{\ast}} {Sys''}
\redtime$, with $t \geq \Theta$, %%for some  $Sys''$, 
  then 
$\barb{Sys''}{\wact {\mathsf{off}} {boiler}}$ 
\item 
If 
${Sys'}\redi^{\ast}\barb{Sys''}{\wact{\mathsf{on}}{light_1}}$ then ${Sys''}\downarrow_{\wact{\mathsf{off}}{light_2}}$, and vice versa. 
%%\item  $\barb{Sys'}{\wact{\mathsf{on}}{light_2}}$ iff $\barb{Sys'}{\wact{\mathsf{off}}{light_1}}$.
%%\end{itemize}
\end{itemize}
\end{proposition}

%%In the full paper~\cite{fullIoT} we provide 
Finally, %%in order to show the potentialities of our calculus, 
we propose a
variant $\overline{Sys}$ of our system, where  lights functionality depends on the GPS coordinates of the smartphone. 
%%The idea is to add in the smart home a 
%%centralised light manager. 
%% which can be remotely accessed by the smartphone to 
%%manage all lights  of the house. 
Intuitively, the smartphone sends its actual position
%%, computed via an IPS (indoor positioning system), 
to a centralised light manager via an Internet channel $g$, 
$\rng {g} = \infty$. 
The centralised manager will then interact with the local 
light managers
to switch on/off  lights of  rooms, depending on the position of the smartphone.
\begin{table}[t]
\(
\begin{array}{rcl}
\overline{Sys} & \deff & \overline{Phone} \q \big| \q \overline{Home}\\
\overline{Home} & \deff & Home \q \big| \q \overline{CLightMng} \\
\overline{Phone} & \deff &\nodep{n_P}{\conf {\I_P}{BoilerCtrl \newpar
\overline{LightCtrl}}}{\mob}{out}\\
\overline{LightCtrl} & \deff & \fix{X}@(x).\timeout{\OUT{g}{x}.\sigma.X}{X} \\
\overline{CLightMng} & \deff & \nodep{n_{LM}}{\conf {\emptyset}{\overline{CLM}}}{\stat}{loc3}\\
\overline{CLM} & \deff &
\fix{X}{{\lfloor \LIN{g}{y}.[y =loc1]\timeout{\OUT{c_1}{}. \sigma.X}{X};}}\\
&& \hspace*{9.85mm} [y =loc4]\timeout{\OUT{c_2}{}.\sigma.X}{X};\sigma.X \rfloor X
%%\\
%%&& \hspace*{7.5mm} [y \in \{loc2,loc3,loc4\} ]\,
%%\timeout{\OUT{c_0}{}.\sigma.X}{X};\sigma.X \rfloor{X}\\
\end{array}
\)
\caption{Smart home: a position based light management}
\label{tab:home2}
\end{table}
In Table~\ref{tab:home2}, new components 
%%which 
%have been changed with 
%respect to the system of Table~\ref{case_study_home} 
have been overlined. 
Short-range channels have now different ranges and they serve to communicate with the centralised light manager $\overline{CLightMng}$. Thus, 
%%$\rng {c_0} = 0$ and
 $\rng{c_1}=2$ and $\rng{c_2}=1$.

Prop.~\ref{caseproperty} holds for  the new system $\overline{Sys}$ as well. Actually, the two systems are closely related. 
\begin{proposition} 
\label{prop:SYS-barb}
For $\delta=1$, 
$\res {\tilde{c}} Sys \cong \res{\tilde{c}}\res{g} \overline{Sys}$.
\end{proposition}
The bisimulation proof technique developed in the remainder of the paper
will be very useful to prove such kind of non-trivial system equalities.

We end this section with a comment. While reading this case study the reader 
should  have realised 
 that our reduction 
semantics does not model %%contemplate a rule for 
sensor updates. 
This is because sensor changes depend on the physical environment, while a 
reduction semantics models the evolution of a system in isolation. 
Interactions with the external %%(physical and/or logical) 
environment will be
treated in our \emph{extensional semantics\/} (see Sec.~\ref{lab-sem})

%%%%%%%%%%%%%%%%%%%%%%%%%%%%%%%%%%%%%%%%%%%%%%%%%%%%%%%%%%%%%%%%%%%
\section{Labelled transition semantics}
\label{lab-sem}
\begin{table}[t]
\[
\small
\begin{array}{l@{\hspace*{1cm}}l}
\Txiom{SndP}
{-}
{\timeout{\OUT c v .P}Q \trans{\out c v} P}
&
\Txiom{RcvP}
{-}
{\timeout{\LIN c x .P}Q \trans{\inp c v} P{\subst v x} }
\\[15pt]
\Txiom{Sensor}
{-}
{ \rsens x  s .P \trans{\rsensa v s} P {\subst v x} }
&
\Txiom{Actuator}
{-}
{\wact v a .P \trans{\wact v  a} P}
\\[15pt]
\Txiom{PosP}
{-}
{@(x).P\trans{@h} P{\subst h x}}
&
\Txiom{Com}
{P \trans{\out c v} P' \Q Q \trans{\inp c v} Q' \Q  \rng c = -1}
{P \newpar Q \trans{\tau} P' \newpar Q'}
\\[15pt]
\Txiom{ParP}
{P \trans{\lambda} P'\Q \lambda \neq \sigma}
{P \newpar Q \trans{\lambda} P' \newpar Q}
&
\Txiom{Fix}
{P{\subst {\fix{X}P} X} \trans{\lambda} Q}
{\fix{X}P\trans{\lambda}Q}
\\[15pt]
\Txiom{TimeNil}
{-}
{\nil \trans{\sigma}\nil}
&
\Txiom{Delay}
{-}
{\sigma.P \trans{\sigma} P}
\\[15pt]
\Txiom{Timeout}
{-}
{\timeout{\pi.P}{Q} \trans{\sigma} Q}
&
\Txiom{TimeParP}
{P \trans{\sigma} P' \Q Q \trans{\sigma}Q' \Q P | Q \not\!\!\!\trans{\tau}}
{P \newpar Q \trans{\sigma} P' \newpar Q'}
\end{array}
\]
\caption{Intensional semantics for processes}
\label{tab:lts_processes} 
\end{table}

\begin{table}[t]
\(
\small
\begin{array}{l@{\hspace*{1.5cm}}l}
\Txiom{Pos}
{P\trans{@h}P'  }
{\nodep{n}{\conf \I P}{\mu}{h}\trans{\tau}\nodep{n}{\conf \I P'}{\mu}{h}   }
&
\Txiom{SensRead}
{\I(s)=v \Q P \trans{\rsensa v s} P' }{
\nodep n {\conf \I P}{\mu}{h}  \trans{\tau} \nodep n {\conf \I {P'}}{\mu}{h} }
\\[20pt]
\Txiom{ActUnChg}
{\I(a)=v \Q P \trans{\wact v a } P'}
{\nodep n {\conf \I P}{\mu}{h}  \trans{\tau} \nodep n {\conf {\I} {P'}}{\mu}{h} }
&
\Txiom{LocCom}{P \trans{\tau} P'}
{\nodep n {\conf \I P}{\mu}{h}
\trans{\tau}\nodep n {\conf \I {P'}}{\mu}{h}
}
\\[20pt]
\multicolumn{2}{c}{
\Txiom{ActChg}
{ \I(a)\neq v \Q P \trans{\wact v a } P' \Q \I':=\I[a \mapsto v] }
{\nodep n {\conf \I P}{\mu}{h}  \trans{a} \nodep n {\conf {\I'} {P'}}{\mu}{h}}
}
\\[20pt]
%%\Txiom{Move}
%%{ i<j   \leq \delta  \Q j =i + \dist h k   }
%%{\nodep{n}{\conf \I P}{\mobi}{h} \trans{\tau} \nodep{n}{\conf \I {P}}{\mobj}{k}}
%%&
%%\Txiom{Stop}{ i < \delta}
%%{\nodep{n}{\conf \I P}{\mobi}{h} \trans{\tau} \nodep{n}{\conf \I {P}}{\mobdelta}{h}}
%%\\[15pt]
\Txiom{TimeStat}
{P \trans{\sigma} P'  \q \nodep n {\conf \I P}{\stat}{h}\!\!  \ntrans{\tau }
%%\\
%%\nodep n {\conf \I P}{\stat}{h}  \ntrans{a}
}
{\nodep n {\conf \I P}{\stat}{h}  \trans{\sigma} \nodep n {\conf \I {P'}}{\stat}{h}}

& 

\Txiom{TimeMob}
{P \trans{\sigma} P'  \q \nodep n {\conf \I P}{\mob}{h}\!\!  \ntrans{\tau}
\q {\scriptstyle \dist h k {\leq} \delta}
%%\\ \nodep n {\conf \I P}{\mobdelta}{h}  \ntrans{a}
}
{\nodep n {\conf \I P}{\mob}{h}  \trans{\sigma} \nodep n {\conf \I {P'}}{\mob}{k}}

\\[20pt]
\Txiom{Snd}
{P \trans{\out c v} P' \Q \rng c \geq 0}
{\nodep{n}{\conf \I P}{\mu}{h}\trans{\send{c}{v}{h}}\nodep{n}{\conf \I P'}{\mu}{h}}
&
\Txiom{Rcv}
{P \trans{\inp c v} P'  \Q  \rng c \geq 0 }
{\nodep{n}{\conf \I P}{\mu}{h}\trans{\rec{c}{v}{h}}\nodep{n}{\conf \I P'}{\mu}{h}}
\\[20pt]

\multicolumn{2}{c}{\Txiom{GlbCom}
{M \trans{\send{c}{v}{k}} M' \Q N \trans{\rec{c}{v}{h}}N'  \Q     \dist h k\leq \rng c}
{M | N \trans{\tau} M' | N'}
}
\\[20pt]
\multicolumn{2}{l}{
\Txiom{ParN}
{M \trans{\nu} M' \Q \nu\neq\sigma}
{M | N \trans{\nu} M' | N}
\Q
\Txiom{TimePar}
{M \trans{\sigma} M' \Q N \trans{\sigma} N' \Q
 M | N \ntrans{\tau} }
{M | N \trans{\sigma} M' | N'}
}
%%%\Q  M | N \ntrans{a}
\\[20pt]
\Txiom{TimeZero}
{-}
{\zero \trans{\sigma}\zero}
&
{\Txiom{Res}{M \trans{\nu} N \Q \nu \not\in \{ \send{c}{v}{h}, \rec{c}{v}{h}
\} }
{\res c M \trans{\nu} \res c {N}}}
\end{array}
\)
\caption{Intensional semantics for networks}
\label{tab:lts_networks}
\end{table}
In this section we provide two 
labelled semantic models, in the  SOS style of Plotkin~\cite{Plo04}: the \emph{intensional semantics} and the \emph{extensional semantics}.
The adjective intensional is used to stress the fact that the actions of that semantics correspond to those activities which can be performed 
by a system in isolation, without any interaction with the external environment.
Whereas, the extensional semantics focuses on those  activities which 
require a contribution of the environment. 
\subsection{Intensional semantics}
%\label{lab-sem_1}
%%Here we provide the \emph{intensional semantics} of our systems, modelling how networks of devices evolve depending on the processes running at each node.
%%Intuitively, the intensional semantics 
%%represents system activities which do not require any interaction with the 
%%  environment.
Since our syntax distinguishes between networks and processes, we have two different kinds of transitions: 
\begin{itemize}
\item $P \trans{\lambda} Q$, with
 $\lambda \in \{\sigma, \tau,  {\out c v}, {\inp c v}, @h, {\rsensa v s}, {\wact v a}\}$, for \emph{process transitions\/}
\item  $M \trans{\nu} N$, with 
$\nu \in \{\sigma, \tau, a, \send{c}{v}{h},  \rec{c}{v}{h} \} $, for \emph{network transitions\/}.  
\end{itemize}

In Tab.~\ref{tab:lts_processes} we report standard transition rules for processes. As in CCS, we assume $\match b P Q=P$ if $\bool{b}=\true$, and 
$\match b P Q =Q$ if $\bool{b}=\false$. Rule \rulename{Com} model 
intra-node communications along channel $c$; that's why $\rng c = -1$.
 The symmetric counterparts of Rules \rulename{ParP} and \rulename{Com} are
 omitted.

In Tab.~\ref{tab:lts_networks} we report the transition rules for networks.
Rule \rulename{Pos} extracts the position of a node.   
Rule \rulename{SensRead} models the reading of a value from a sensor
of the enclosing node. Rules~\rulename{ActUnChg} and \rulename{ActChg}  describes the writing of a value $v$ on an actuator $a$ of the node, distinguishing whether
the value of the actuator is changed or not. 
%%%, which are both treated as internal actions of the node.
Rule \rulename{LocCom} models intra-node communications. 
%%Rules \rulename{Move} and \rulename{Stop}, together with rule 
Rule~\rulename{TimeStat} models the passage of time for a stationary node.
Rule~\rulename{TimeMob} models both time passing and node mobility at the end of a time interval.
%%A mobile node can travel for a distance of at most $\delta$ within each time unit.
Rules \rulename{Snd} and \rulename{Rcv}  represent transmission and reception along a global channel.  Rule \rulename{GlbCom} models inter-node 
communications. The remaining rules are straightforward. 
The symmetric counterparts of Rule~\rulename{ParN} and Rule~\rulename{GlobCom} are omitted. 
%%: two nodes can communicate via a channel $c$ if they are
%%located within the communication range of $c$ (global channel have infinite
%%transmission rage).
%% MASSIMO: La frase sopra l'abbiamo gia' scritta due volte. 

As expected, the  reduction semantics and the labelled intensional
semantics coincide. 
%%The following result is standard in literature and allows us to prove that the reduction semantics and the intensional semantics coincide.
%%the equivalence between the two semantics model we have proposed.
%%\vspace*{-3mm}
\begin{theorem}[Harmony theorem]
\label{thm:harmony}
Let $\omega \in \{ \tau, a ,\sigma\}$: 
\begin{itemize}
\item 
$M\trans{\omega}M'$ implies $M\red_{\omega} M'$ 
\item 
$M\red_{\omega} M'$ implies $M\trans{\omega}{\equiv}M'$. 
%%\item \label{act_red}  $M\trans{a}M'$ implies $M\red_{a} M'$
%%\item \label{red_act}  $M\red_{a@n} M'$ implies $M\trans{a@n}_{\equiv}M'$
%%\item \label{sigma_red} $M\trans{\sigma}M'$ implies $M\redtime M'$
%%\item \label{red_sigma}  $M\redtime M'$ implies $M\trans{\sigma}_{\equiv}M'$.
\end{itemize}
\end{theorem}

\subsection{Extensional semantics}

\begin{table}[t!]
\(
\small
\begin{array}{l@{\hspace*{1cm}}l}
\Txiom{SndObs}
{M \trans{\send{c}{v}{h}} M' \q\; \dist h k  {\leq} \rng c }
{M \trans{\sendobs{c}{v}{k}} M'}%
&
\Txiom{RcvObs}
{M \trans{\rec{c}{v}{h}} M' \q\; \dist k h  {\leq} \rng c}
{M \trans{\recobs{c}{v}{k}} M'}%
\\[20pt]
\Txiom{SensEnv}{\textrm{$v$ in the domain of $s$}}{ M \trans{\rsensa v {s@h}} 
M[{s@h} \mapsto v]
}
&
%%\Txiom{NodeSensEnv}{v \in Dom(s)}{ M \trans{\rsensa v {s@n}} 
%%M[{s@n} \mapsto v] 
%%}
%%\\[12pt]
%%\multicolumn{2}{c}{
\Txiom{ActEnv}
{\barb{M} {\wact v {a@h}}}
{M  \trans{\wact v {a@h} } M }
%%}
\end{array}
\)
\caption{Extensional semantics: additional rules}
\label{tab:extensional}
\end{table}
Here we redesign our LTS to focus on the interactions of 
our systems with the external environment. 
As the  environment has a \emph{logical part\/} (the parallel nodes) and a \emph{physical part\/} (the physical world) our extensional semantics
 distinguishes two different kinds of transitions:
\begin{itemize}
\item 
$M \trans{\alpha} N$, \emph{logical transitions\/}, 
for $\alpha \in \{ \tau, \sigma, a, \sendobs{c}{v}{k}, \recobs{c}{v}{k} \}$, 
to denote the interaction with the \emph{logical environment}; 
here, actuator changes, $\tau$- and $\sigma$-actions
 are inherited from the intensional semantics, so we don't provide inference
rules for them;
\item  $M \trans{\alpha} N$, \emph{physical transitions\/}, for 
$\alpha \in \{ \rsensa v {s@h},  \wact v {a@h} \} $, to denote the interaction 
with the \emph{physical world\/}.
\end{itemize}
In Tab.~\ref{tab:extensional} the extensional actions deriving from rules 
 \rulename{SndObs} and \rulename{RcvObs}   mention 
the location $k$ of the logical environment which can 
\emph{observe} the communication occurring at channel $c$. 
Rules  \rulename{SensEnv}  and \rulename{ActEnv} model the interaction
 of a system $M$ with the physical environment. 
In particular, the environment can \emph{nondeterministically update} the current value of a (location-dependent or node-dependent) sensor $s$  with a value $v$, and can read the value $v$ appearing on an actuator $a$ at $h$.  As already discussed 
in Sec.~\ref{sec:barbed} the environment is potentially aware of its position 
when doing these actions.

Note that our LTSs are \emph{image finite\/}. They are also 
\emph{finitely branching\/}, and hence \emph{mechanisable\/}, under the obvious assumption of finiteness of all domains of admissible values, and the set of physical locations.

\section{Full abstraction}
\label{full-abstraction}
Based on our extensional semantics,  we are ready to 
define a notion of bisimilarity which will be showed to be
 both sound and complete with respect to our contextual equivalence. 
%In order to provide our notion of bisimulation 
We adopt a standard notation for weak transitions. We denote with $\ttrans{}$ the reflexive and transitive closure of $\tau$-actions, namely $(\trans{\tau})^*$, whereas  $\ttrans{\alpha}$ means $\ttranst{\alpha}$, and finally $\ttrans{\hat{\alpha}}$ denotes $\ttrans{}$ if $\alpha=\tau$ and $\ttrans{\alpha}$ otherwise.
\begin{definition}
[Bisimulation]
\label{bisimulation}
A binary symmetric relation $\rel$ over networks is a \emph{bisimulation\/} if $M \RRr N$ and 
$M \trans{\alpha} M'$ imply there exists $N'$ such that $N\Trans{\hat{\alpha}}N'$ 
and 
$M' \RRr N' $. 
We say that  $M$ and $N$ are \emph{bisimilar\/}, written $M \approx N$, if $M \RRr N$ for some  bisimulation $\RR$.
\end{definition}

Sometimes it is useful to count the number of $\tau$-actions performed by a
process. The \emph{expansion} relation~\cite{A-KHe92}, written 
$\isexpan$, is an
asymmetric variant of $\approx$ such that $P \isexpan Q$ holds if $P \approx Q$ and $Q$ has at least as many
$\tau$-moves as $P$. 

As a workbench we can use  our notion of bisimilarity to prove a number of algebraic laws
on well-formed networks.  
%%We recall that we always assume well-formed networks. 
\begin{theorem}[Some algebraic laws]\
\label{thm:algebraic-laws}
\begin{enumerate}
\item
\label{law1}
 $ \nodep{n}{\conf \I {{\wact v a}.P} | R}{\mu}{h} 
\gtrsim \nodep{n}{\conf \I {P | R}}{\mu}{h}$, if $\I(a)=v$ and $a$ does not
occur in $R$
\item  
\label{law2} $ \nodep{n}{\conf \I {@(x).P | R}}{\mu}{h} \gtrsim \nodep{n}{\conf \I {\subst h x}P | R}{\mu}{h}$ %% if $\mu \in \{\stat,\mobdelta\}$
 \item 
\label{law3} 
{\small \mbox{$ \nodep{n}{\conf \I {\timeout{\OUT c v.P}{S}} | \timeout{\LIN c x. Q}{T} |  R}{\mu}{h} 
\gtrsim \nodep{n}{\conf \I {P | Q{\subst v x}| R}}{\mu}{h}$}}, if $c$ is not in $R$ and $\rng c = -1$
\item 
\label{law4}
 $ \res c {( \nodep{n}{\conf {\I} {\timeout{\OUT c v.{P}}{S}} |  {R}}{\mu}{h} \, 
| \, 
\nodep{m}{\conf  {\J} {\timeout{\LIN c x. {Q}}{T}} |  {U}}{\mu'}{k})}$\\
$\gtrsim 
 \res c {( \nodep{n}{\conf {\I} {{{P}} |   {R}}}{\mu}{h} \, | \, 
\nodep{m}{\conf  {\J} {{Q{\subst v x}} |  {U}}}{\mu'}{k} ) } $  if $\rng c{=}\infty$ and  $c$ does not occur in $R$ and $U$. 

\item
 \label{law5} $\nodep{n}{\conf \I P}{\mu}{h} \approx\nodep{n}{\conf \I \nil}{\mu}{h}$ if  subterms 
 $\timeout{\pi.P_1}{P_2}$ or ${\wact v a}.P_1$
do not occur in $P$
\item 
\label{law6}
$\nodep{n}{\conf \I \nil}{\mu}{h} \approx \zero$  if $\I(a)$ is undefined for any actuator $a$
%%\item 
%%\label{law7}
%%$ \nodep{n}{\conf \I {P}}{\stat}{h} 
%%\approx \nodep{m}{\conf \J {P}}{\stat}{h}$ 
%if $\I$ and $\J$ are defined only for the same set of location-dependent sensors
\item
\label{law8}
  $\nodep{n}{\conf \emptyset P}{\mobi}{h} \approx \nodep{m}{\conf \emptyset P}{\stat}{k}$ if $P$ does not contain terms of the form $@(x).Q$, and 
for any channel $c$ in $P$ either $\rng c{=} \infty$ or $\rng c = -1$. 
\end{enumerate}
\end{theorem}
 Laws~\ref{law1}-\ref{law4} are a sort of tau-laws. 
Laws~\ref{law5} and \ref{law6} models garbage collection 
of processes and nodes, respectively. Law~\ref{law8} gives
a sufficient condition for node anonymity as well as for 
 non-observable node mobility. 

Now, it is time to show how our labelled bisimilarity can be used to deal with
 more complicated systems. In particular, if you
consider the systems of Prop.~\ref{prop:SYS-barb}, it holds the following:
\begin{proposition} 
\label{prop:SYS-bis}
If $\delta=1$  then 
$\res {\tilde{c}} Sys \approx \res{\tilde{c}}\res{g} \overline{Sys}$.
\end{proposition}
Due to the size of the systems involved, the proof of the proposition 
above is quite challenging. In this respect, the first four 
laws of Thm.~\ref{thm:algebraic-laws} are fundamentals
to apply 
 non-trivial up to expansion proof-techniques~\cite{Sangiorgi-book}. 
%%%These are the kind of results that  would benefit of theorem provers such as Isabelle~\cite{Isabelle} or Coq~\cite{Coq}.  

In the remainder of the section we provide the full abstraction result, i.e.\
we prove that our labelled bisimilarity  is a sound and complete characterisation of 
reduction barbed congruence. 
 %%the contextual equivalence of Sec.~\ref{sec:barbed}; this is the subject of Sec.~\ref{sec:soundness}. 
%%Moreover it is also complete as proved in Sec.~\ref{sec:completeness}.
%%\subsection{Soundness}
%%\label{sec:soundness}
%%The following theorem shows that bisimilarity is a contextual relation.
%%\begin{theorem}
%%[Congruence]
%%Let $M$ and $N$ be two networks such that $M \approx N$. 
%%Then, for all networks $O$,  $M | O \approx  N | O$ and for all sensors $s$ and values $v$. $M[s\mapsto v]\approx N[s\mapsto v]$.
%%\end{theorem}

In order to prove soundness, we provide 
the following easy technical 
result relating barbs with  
extensional actions. 
\begin{proposition}
\label{prop:barb}
%%Let $M$ be a network.%%
%%\begin{itemize}
%%\item $\barb{M}{\out c @h}$ if and only if $M\trans{\sendobs{c}{v}{h}}M'$
%%\item $\barb{M}{\inp c @h}$ if and only if $M\trans{\recobs{c}{v}{h}}M'$%
%%\item
$\barb{M}{\wact v {a@h}}$ if and only if 
$M \trans{\wact v {a@h} }M$.
%%\end{itemize}
\end{proposition}%%
\begin{proof}
It follows from the definition of rule \rulename{ActEnv}.
\end{proof}

A crucial result is that our  bisimilarity is a congruence.
\begin{theorem}
\label{thm:congruence}
The relation $\approx$ is contextual. 
\end{theorem}

\begin{proof}[Proof (Sketch)]
The most difficult case is when proving that  $M\approx  N$ entails $M[s@h \mapsto v]  \approx  N[s@h\mapsto v]$, for all  sensors
$s$, locations $h$, and  values $v$ in the domain of $s$. In fact, a standard approach 
to this proof consisting in  trying to show that the relation
\[
\left \{ \big( M[s@h \mapsto v]  \, , \,   N[s@h\mapsto v]\big) : M \approx N \right \}
\]
is a bisimulation, is not affordable. 

Thus, our proof is by well-founded induction.  Details can be found in the Appendix. 
\end{proof}

Now, everything is in place to prove that our labelled bisimilarity 
is sound with respect to reduction barbed congruence. Basically, 
we have to prove that the labelled bisimilarity is reduction-closed, barb preserving and contextual. 
\begin{theorem}[Soundness]
\label{thm:sound}
Let $M$ and $N$ be two networks such that $M\approx N$, then $M\cong N$.
\end{theorem}

\begin{proof}
We recall that  $\cong$ is defined (Def.~\ref{barbed_cong}) as the largest symmetric reduction which is reduction closed, barb preserving and contextual.

First, we prove that bisimilarity is reduction closed.

Suppose that $M\red M'$.
Then we have two cases: either $M\redtau M'$ or $M\redtime M'$.
In the first case Th.~\ref{thm:harmony}  implies that $M\trans{\tau}\equiv M'$.
Since by hypothesis $M\approx N$, then there exists $N'$ such that $N\Trans{}N'$ and $N'\approx M'$.
Now, by Th.~\ref{thm:harmony}  we have that each of the $\tau$-actions in the sequence $N\Trans{}N'$ can be rewritten in terms of $\redtau$.
Thus the entire sequence $N\Trans{}N'$ ca be rewritten as the sequence of instantaneous reductions $N\redtau^* N'$, which is a particular case of $N\redmany N'$, with $N'\approx M'$.
Let us conder now the second case: $M\redtime M'$.
By Th.~\ref{thm:harmony} it follows that $M\trans{\sigma} \equiv M'$.
As $M\approx N$ there exists $N'$ such that $N\ttrans{\sigma}N'$ and $N'\approx M'$.
By several applications of Th.~\ref{thm:harmony} we get $N\redtau^*\redtime\redtau^* N'$. Thus, $N\redmany N'$, with $N'\approx M'$.

 The case $M\red_{a} N$ is similar. 
%%Th.~\ref{thm:harmony} implies that $M\trans{a}_{\equiv}M'$.
%%Since by hypothesis $M\approx N$, then there exists $N'$ such that $N\ttrans{a}N'$ and $N'\approx M'$.
%%Now, by Th.~\ref{thm:harmony}  we have that each of the $\tau$-moves in the sequence $N\ttrans{a}N'$ can be rewritten in terms of $\redtau$ and moreover Th.~\ref{thm:harmony}(\ref{a_red}) guarantees that the $a$-move can be rewritten in terms of a reduction derivation $\red_{a}$.
%%Hence we have obtained that $N\redtau^*\red_{a}\redtau^* N'$ with $N'\approx M'$ is inferable.
%%
%%Thus bisimilarity is reduction closed.

From reduction closure and  Prop.~\ref{prop:barb} it follows immediately that
$\approx$ is barb preserving. 

Thm.~\ref{thm:congruence} proves that our labelled bisimilarity is contextual.

As $\cong$ is defined as the largest relation which is reduction closed, barb-preserving and contextual, it follows that $\approx \; \subseteq \;\cong$.
\end{proof}

Here, before proving completeness, we would 
like to point out some peculiarities of our bisimilarity. 
As the reader may have noticed, our bisimulation is completely 
standard, in a weak fashion. However, the real distinguishing power of 
physical transitions follows the pattern of strong bisimulation. 

\begin{proposition}[Physical environment and strong observation] \ 
\label{prop:strong-obs}
\begin{itemize}
\item If $M \approx N$ and $M \trans{\wact v {a@h}} M'$ then there is 
$N'$ such that $N \trans{\wact v {a@h}} N'$ and $M' \approx N'$
\item 
If $M \approx N$ and $M \trans{\rsensa v {s@h}} M'$ then there is 
$N'$ such that $N \trans{\rsensa v {s@h}} N'$ and $M' \approx N'$. 
 %%$ \nodep{n}{\conf \I {P}}{\mu}{h} 
%%\approx \nodep{n}{\conf \J {Q}}{\mu}{h}$ entails 
%%  $\I(a)= \J(a)$ for any $a$
\end{itemize}
\end{proposition}
\begin{proof}
Let us prove the first item. By Thm.~\ref{thm:sound} we derive $M \cong N$. 
By Prop. \ref{prop:strong-barbs} we know that $\barb{M}{\wact v {a@h}}$ implies 
$\barb{N}{\wact v {a@h}}$. The result follows by inspection of the definition 
of the rule \rulename{ActEnv} to derive the transition $\trans{\wact v {a@h}}$.

Let us prove the second item. 
By  an application of Thm.~\ref{thm:congruence} we have 
$M[{s@h}\mapsto v]\approx N[{s@h}\mapsto v]$, 
for all sensors $s$ and values $v$ in the domain of $s$. 
The result follows by inspection of the rule  \rulename{SensEnv} to derive
the transition $\trans{\rsensa v {s@h}}$.
\end{proof}
This first sub-result  is perfectly in line with Proposition~\ref{prop:strong-barbs}.  Intuitively, 
the whole Prop.~\ref{prop:strong-obs} says  that changes in the 
physical environment may have immediate consequences on  IoT systems:
 waiting for a $\tau$-action  might make a difference.

  We prove now completeness. The proof relies on showing that for each extensional action 
 $\alpha$ it is possible to exhibit a test $T_{\alpha}$
  which determines whether or not a system $M$ can perform the action $\alpha$.
We need a technical lemma to cut down observing contexts. 
%%%\begin{lemma}
%%%\label{aux_compl1}
%%%Let $O ,M$ and $N$ be three  networks such that  $O=\nodep{n}{\conf{\mathcal I}{\nil}}{\stat}{k}$. If $  M | O \cong   N | O$, then $  M \cong   N $.
%%%\end{lemma}
\begin{lemma}
\label{lem:cut_cxt}
Let $M$ and $N$ be two networks. 
Let $O=\nodep{n}{\conf{\mathcal I}{\wact v a}.\nil}{\stat}{k}$, 
for an arbitrary node name $n$, an arbitrary actuator $a$, and arbitrary values $v$ and $w$, in the domain of $a$,  such that 
$\I$ is only defined for $a$ and $\I(a)=w\neq v$. 
If both $M |O$ and $N | O$ are 
well-formed and  
 $  M | O \cong   N | O$ then $  M \cong   N $.
\end{lemma}

\begin{theorem}[Completeness]
\label{thm:complete}
Let $M$ and $N$ such that $M\cong N$, then $M\approx N$. 
\end{theorem}

\begin{proof} 
We show that relation $\rel = \{(M,N)\mid M\cong N\}$ is a bisimulation up to 
$\equiv$.
Let us consider two networks $M$ and $N$ such that $(M,N)\in\rel$. We proceed
by case analysis on the possible extensional actions of $M$. 

First, we consider  \emph{logical transitions\/}. 
\begin{itemize}

\item Let us suppose that $M\trans{a}M'$.
By Th.~\ref{thm:harmony}  we derive $M\red_a M'$.
Let us define the test $T_{a}$: 
\[
T_a \deff \nodep{n}{\conf{\J}{ \wact 1 b . \nil} } {\stat}{k}
\]
where $n$ is a fresh node name and $b$ is a fresh  actuator  such that $\J(b)=0$.  By Prop.~\ref{prop:maxprog}, no $\sigma$-move can fire if  a
reduction $\red_b$ is possible. Thus,  the presence of a barb  $\Downarrow_{\wact 0 {b@k}}$ means that no $\sigma$-actions have occurred yet.
Since $M\red_a M'$, we can apply rule \rulename{parn} to infer $M|T_a \red_a M'|T_a$, with $\barb{M'|T_a}{\wact 0 {b@k}}$.
As $M\cong N$ and the relation $\cong$ is both contextual and reduction closed, it follows that  $N|T_a \redmany \red_a \redmany \hat{N}$, for some $\hat{N}$, with $M'|T_a \cong \hat{N}$.
As a consequence, $\redbarb{\hat{N}}{\wact 0 {b@k}}$.  This implies that $\hat{N}\equiv N'|T_a$ for some $N'$, such that $N|T_a\redmany \red_a \redmany N'|T_a$, with  $N\redmany \red_a \redmany N' $, and $M'|T_a\cong N'|T_a$. As the presence 
of a barb $\Downarrow_{\wact 0 {b@k}}$ ensures that no $\sigma$-actions have occurred, it follows that  $N\redtau^{\ast} \red_a \redtau^{\ast} N' $. By 
several applications of Thm.~\ref{thm:harmony} it follows that 
$N \Trans{a}\equiv N'$ (this relies on the straightforward result that $\equiv$ is a strong bisimulation). 
By  $M'|T_a \cong N'|T_a$ and Lem.~\ref{lem:cut_cxt} we derive $M'\cong N'$. This implies that 
$(M', N') \in \; \equiv \rel \equiv $.

\item Let us suppose that $M\trans{\tau}M'$. This case is 
similar to the previous one with $T_{\tau}=T_a$.

\item Let us suppose that $M\trans{\sigma}M'$.
By Th.~\ref{thm:harmony}  we derive $M\redtime M'$.
As $M\cong N$,  by Prop.~\ref{prop:time-observation}  there exists $N'$ such that  $N \redtau^{\ast}\redtime\redtau^{\ast}N'$ and $M'\cong N'$. By several applications of Th.~\ref{thm:harmony}  we obtain  $N \ttrans{\sigma} \equiv N'$. 
As $M'\cong N'$, it follows that  $(M',N') \in \; \equiv \rel \equiv $.

\item Let us suppose that $M\trans{\sendobs{c}{v}{k}}M'$.
This transition can only be derived by an application 
of rule \rulename{SndObs} if $M\trans{\send{c}{v}{h}}M' $, for some $h$, such that  $\dist h k \leq \rng c$.
%%By Lem.~\ref{lem:struc_in_out}(\ref{struc_in})  we derive that there exist $n,P,P',Q,\mu,M_1, {\tilde g}$ such that $c\not \in  {\tilde g} $ and
%%\[
%%M\equiv \res {\tilde g} \big(\nodep{n}{\conf \I{\timeout{\OUT c v .P}P'\newpar Q}}{\mu}{h} | M_1\big)
%%\]
%%and 
%%\[
%%M'\equiv\res {\tilde g} \big( \nodep{n}{\conf \I{P\newpar Q}}{\mu}{h} | M_1
%%\big).
%%\]
Let us build up a context that is capable to observe the action $\sendobs{c}{v}{k}$.  We define testing term $T_{\sendobs{c}{v}{k}}$. For simplicity, in the following we abbreviate it with $T$: 
\[
T \deff \nodep{m}{\conf{\J}{ \timeout{\LIN c x.[x=v]\wact 1 b. \wact 0 b . \nil;\nil}\nil} } {\stat}{k}
\]
where $m$ is a fresh node name and $b$ is a fresh  actuator name  such that $\J(b)=0$.
The intuition behind this testing process is the following:  $T$ has barb 
$\Downarrow_{\wact 1 {b@k}}$ only if the communication along $c$ has already occurred and no time actions have been fired (Prop.~\ref{prop:maxprog}).

Since $\cong$ is contextual, $M\cong N$ implies $M | T \cong N | T$.
From $M\trans{\send{c}{v}{h}}M'$ we can easily infer $
M | T \trans{\tau} \trans b  M' |T'$, with 
$T'= \nodep{m}{\conf{\J[b \mapsto 1]}{  \wact 0 b . \nil} } {\stat}{k}$.
Notice that $\barb{M'|T'}{\wact 1 {b@k}}$.
By  Th.~\ref{thm:harmony}, we derive $M|T \redtau \red_b M'|T'$.
As $M | T \cong N | T$ it follows that 
$
N | T \redmany \red_b \redmany \hat{N}$, with 
$\hat{N} \Downarrow_{\wact 1 {b@k}}$. This implies that 
 $\hat{N} \equiv  N'|T'$, for some $N'$. Furthermore, 
no timed actions have occurred in the reduction sequence, and 
hence: $
N | T \redtau^{\ast}\red_b \redtau^{\ast}  N' |T '$.  By several applications 
of Thm.~\ref{thm:harmony} we obtain $N | T \Trans{} \trans{b} \Trans{} \equiv N' |T '$.
This implies that $N\Trans{\send{c}{v}{h'}} \equiv N'$, for some $h'$ such that  $\dist {h'} k \leq \rng c$. By an application of rule \rulename{SndObs} we get $N\Trans{\sendobs{c}{v}{k}} \equiv N'$.
From $M'|T\cong N'|T$ and Lem.\ref{lem:cut_cxt} we derive $M'\cong N'$. This allows us to show that
$(M',N') \in \; \equiv \rel \equiv$.

\item The case of $M\trans{\recobs{c}{v}{k}}M'$, is similar  to the previous one. The observing term is 
\[
T_{\recobs{c}{v}{k}} \deff \nodep{m}{\conf{\J}{ \timeout{\OUT c v.\wact 1 b. \wact 0 b . \nil}\nil} } {\stat}{k}
\]
where $m$ is a fresh node name and $b$ is a fresh  actuator name such that $\J(b)=0$. 
\end{itemize}

Let us now consider  \emph{physical transitions}. Here, as already explained 
in Sec.\ref{sec:extensions}, we will not provide an observing context
as our language for IoT systems does not allow us to write physical observers. 
\begin{itemize}
\item Let  $M\trans{\wsensa v {a@h}}M'$.
Since this transition can  be only derived by an application of rule \rulename{ActEnv}, it follows that $M'=M$ and $\barb{M}{\wsensa v {a@h}}$. By Prop.~\ref{prop:strong-barbs} we obtain $\barb{N}{\wsensa v {a@h}}$. By applying again rule \rulename{ActRead} to $N$, we obtain $N\trans{\wsensa v {a@h}}N'=N$ with $(M', N') \in \rel$. 

\item Let  $M\trans{\rsensa v {s@h}}M'$.  %M[{s@h}\mapsto v]$. 
Since this transition can  be only derived by an application of 
 rule \rulename{SenEnv}, it follows that $M' = M[{s@h}\mapsto v]$. 
By an application of the same rule \rulename{SensEnv} we obtain $N\trans{\rsensa v {s@h}}N'=N[{s@h}\mapsto v]$. As $\cong$ is contextual 
we have $M[{s@h} \mapsto v]\cong N[{s@h}\mapsto v]$. This implies
that $(M', N') \in  \rel$.
 \end{itemize}
\end{proof}

By Thm.~\ref{thm:sound} and Thm.~\ref{thm:complete} we derive our 
full abstraction result: 
reduction barbed congruence coincides with our labelled bisimilarity.
\begin{corollary}
[Full abstraction]
\label{thm:full-abstraction}
$M\approx N$ if and only if $M\cong N$.
\end{corollary}

\begin{remark}
A consequence of Thm.~\ref{thm:full-abstraction} and Rem.~\ref{rem:nondet} is that 
 our 
 bisimulation proof-technique  remains sound in a setting with
 more restricted contexts, where 
nondeterministic sensor updates %, implemented via rule~\rulename{SensEnv},
are replaced 
by some specific model for sensors. 
\end{remark}

%%%%%%%%%%%%%%%%%%%%%%%%%%%%%%%%%%%%%%%%%%%%%%%%%%%%%%%%%
%%%%%%%%%%%%%%%%%%%%%%%%%%%%%%%%%%%%%%%%%%%%%%%%%%%%%%%%%
%%%%%%%%%%%%%%%%%%%%%%%%%%%%%%%%%%%%%%%%%%%%%%%%%%%%%%%%%

%%%%%%%%%%%%%%%%%%%%%%%%%%%%%%%%
%%%%%%                                                                      %%%%%%%
%%%%%%                 C O N C L U S I O N S                   %%%%%%%
%%%%%%                                                                      %%%%%%%
%%%%%%%%%%%%%%%%%%%%%%%%%%%%%%%%

\section{Conclusions, related and future work}
\label{conc}
%%\marginpar{MM: accurate related work!}
%%{\bf \begin{itemize}
%%\item Cosa succederebbe con comunicazioni che prendono tempo? Collisioni
%%
%%\item confronto con linguaggi reattivi (piu' che altro nei related)
%%\item riguardare commenti dei referee e in particolare Lanese.
%%\end{itemize}
%%}
We have proposed a process calculus, called \cname,  to investigate the 
semantic theory of systems based on the Internet of Things paradigm.
The calculus is equipped with a simple \emph{reduction semantics\/}, to model the 
dynamics of systems in isolation, and  an \emph{extensional semantics} 
to emphasize the interaction of IoT systems
with the environment. The latter semantics has been 
used to define a \emph{labelled bisimilarity} which has been proved to be 
\emph{fully abstract} with respect to a natural notion of contextual equivalence. Our bisimilarity has been used to prove non-trivial 
system equalities.

To our knowledge,  paper~\cite{LBdF13} is the first 
 process calculus for IoT systems to 
capture the interaction between sensors, actuators and computing processes. 
Smart objects are represented as point-to-point communicating nodes of 
heterogeneous networks.  
The network topology is represented as a graph whose links can 
 be nondeterministically established or destroyed. 
%However, as already pointed out in Sec. 3, this non-determinism could entail some ill-behaviour by the systems
The paper  contains a labelled transition system with two different kinds of 
transitions. The first one takes into account interactions with the physical environment, 
similarly to our physical transitions, but includes also topology changes.
The second kind of transition models nodes activities, mainly communications, similarly to our logical transitions. 
Then the paper proposes two notions of bisimilarity: 
one using only the first kind of transitions and equating systems from the point of view of the end user, and a second one using all transitions and equating systems  from the point of view of the other devices.

 We report here the main differences
between \cname\ and the IoT-calculus.  
In \cname\, we support timed behaviours, with desirable 
time and fairness properties.  
Both sensors and actuators in \cname\ are under the control
of a single entity, i.e.\ the controller process of the node where they are deployed. This was a 
security issue.  
The nondeterministic link entailment of the IoT-calculus makes the semantics of communication simpler 
than ours; on the other hand it does not allow  to enforce that a smart device should be either in a place or in another, but never in both.
\cname\ has a finer control of inter-node communications
as they depend  on nodes' distance and transmission range of channels.
Node mobility in \cname\  is 
timed 
constrained: in one time unit at most a fixed distance $\delta$ may be covered. 
Finally, Lanese et al.'s \emph{end-user bisimilarity}
 shares the same motivations of our bisimilarity. 
In the IoT-calculus, 
end users provide values to sensors and check actuators.
They  can also directly observe node mobility, 
but they cannot observe channel communication. 
Our bisimilarity may observe node mobility in an indirect manner: 
the movements of a  mobile node can be observed  if the node  either uses an actuator or transmits along a short-range channel or communicates its physical position. 
Lanese et al.'s End-user bisimilarity is not
 preserved by parallel composition. Compositionality 
is recovered by strengthening its discriminating ability.

Our calculus takes inspiration from algebraic models for wireless
 systems~\cite{MeSa06,NaHa06,Merro07,Godskesen07,GhaFokMov11,MBS11,MerSib13,KouPhi11,LM11,
SRS06,FGHMPT12,CHM15,BHJRVPP15}. All 
these models adopt  broadcast communication on  partial topologies, while we
 consider point-to-point communication, as in~\cite{LBdF13}.
Our way  of modelling network topology is taken from~\cite{MeSa06,Merro07}.
%%\cite{NaHa06} deals with 
%%security analysis of communication protocols;
Paper~\cite{GodNanz09} provides formal models for node mobility depending on 
the passage of time. 
Prop.~\ref{prop:time-observation} was inspired by~\cite{CHM15}. 
A fully abstract observational theory for untimed ad hoc networks 
can be found in~\cite{Merro07,Godskesen07}. 
%%Papers~\cite{NaHa06,GhaFokMov11} do security analysis and model checking, respectively, of communication protocols for 
%%mobile wireless networks. 
Paper~\cite{SRS06} provides a symbolic semantics
for ad hoc networks. 
%%Calculi in~\cite{SRS06,FGHMPT12,BHJRVPP15} also have a special focus on the corretness of routing protocols. 

Vigo et al.~\cite{VNN13} proposed a  calculus for wireless-based 
cyber-physical (CPS) systems 
endowed with a theory that allows modelling and reasoning about cryptographic primitives, together with explicit notions of communication failure and unwanted communication. One of the main goal of the paper is a faithful representation 
denial-of-service. However, 
as pointed out in~\cite{WuZh15},  the calculus does not provide a notion of network 
topology, local broadcast and behavioural equivalence. It also lacks 
a clear distinction between physical components (sensor and actuators) and
 logical ones (processes). Compared to~\cite{VNN13}, paper~\cite{WuZh15} 
introduces a static network topology and enrich the theory with an harmony theorem.

As already said,  \cname\ has some similarities with the \emph{synchronous
 languages} of the Esterel family~\cite{Esterel,Boussinot,Castellani,Amadio}.
In this setting, computations proceed in phases called “instants”, which are
 quite similar to our time intervals.  
For instance, 
our timed reduction semantics has many points in common with that
of a recent synchronous reactive language, $CRL$~\cite{Castellani}\footnote{The $CRL$ language does not support mobility.}.
%%Our time
%% properties at the end of Sec.~\ref{red_sem} have  similar 
%%counterparts in~\cite{Castellani}. 
The authors define
 two bisimulation equivalences.
%%
%% formalising respectively a \emph{fine-grained} observation of programs (the observer is viewed as a program) and a \emph{coarse-grained} observation (the observer is viewed as part of the context). The latter equivalence is more abstract than the former, as it only compares the I/O behaviours of programs at each instant, while the former also compares their intermediate results.
The first bisimulation formalises a \emph{fine-grained}
 observation of programs: the observer
is viewed as a program, which is able to interact with the observed program
at any point of its execution. 
The second reflects a \emph{coarse-grained} observation of
programs: here the observer is viewed as part of the environment, 
which interacts with the observed program only at the start and the end of
 instants. 
The fine-grained bisimilarity is more in the style of a bisimulation 
for a process calculus. 
 Finally, the paper in~\cite{Amadio} presents a version of the $\pi$-calculus where communication is synchronous in the sense of reactive synchronous languages, thus integrating the reactive synchronous paradigm into a classical process calculus. A notion of labelled bisimilarity is introduced, which is then characterised as a contextual bisimilarity (these equivalences are close to the coarse-grained bisimilarity of~\cite{Castellani}).

Finally, \cname\ is somehow reminiscent of the SCEL language~\cite{Rocco2015}. 
A framework to model behaviour, knowledge, and data aggregation of Autonomic Systems.

To end the paper, we want to stress  the great potentialities of 
smart objects,  that are not only drivers for changes in terms of content and applications. Given their ability to potentially change in function as they can be digitally enhanced and ``upgraded'', they may acquire \emph{disruptive potentiality} that could lead to 
\emph{serious repercussions\/}~\cite{sicurezza}.
So, the next step of our research will be
 to investigate \emph{security aspects} of IoT
systems from a semantic point of view. There are many
 issues in this respect, such as: 
(i) ensuring continuity and availability in the provision of IoT-based services; 
(ii) preventing unauthorised accesses to users' sensitive information (e.g. login/passwords); (iii) forbidding untrusted operations (e.g. fake requests of deactivation of alarms); (iv)
ensuring traceability and profiling of unlawful processing; 
 (v) ensuring data protection and  individuals’ privacy  (e.g. personal data and health monitoring sensors); (vi) preventing the download of malicious code
via firmware and software upgrade.

\paragraph{Acknowledgements}
We thank Ilaria Castellani and Matthew Hennessy for their precious comments on 
an early draft.

\bibliographystyle{plain}
\bibliography{IoT_bib}

\appendix
\section{Technical Proofs}

\subsection{Proofs of Sec.~\ref{the-calculus}}
\setcounter{proposition}{0}
\setcounter{theorem}{0}

We start with the proofs of Sec.~\ref{red_sem}.
\begin{lemma}[Node time determinism]
\label{lem:timed}
If $\nodep n P \mu h \redtime \nodep {n'} {P'} {\mu'} {h'}$ and 
$\nodep n P \mu h \redtime \nodep {n''} {P''} {\mu''} {h''}$ then $n=n'=n''$, 
$P' \equiv P''$, $\mu = \mu'=\mu''$ and $\dist {h'} {h''} \leq 2 \delta$. 
\end{lemma}
\begin{proof}
By structural induction on $P$. 
\end{proof}

\paragraph{\textbf{Proof of Prop.~\ref{prop:timed}.}}
The proof is by rule induction on why $M \redtime M'$.
\begin{itemize}

\item Let $M\redtime M'$ by an application of rule \rulename{timezero}. This case
is straightforward. 

\item Let  $M\redtime M'$ by an application of rule \rulename{timemob}: 
\[
\Txiombis
{\nodep n {{\conf \I {{\prod_{i \in I}\timeout{\pi_i.P_i}{Q_i}} | \prod_{j \in J} \sigma.P_j}}} {\mob} {k} \not\redtau \Q \dist k {k'} \leq \delta}
{\nodep n {{\conf \I {{\prod_{i \in I}\timeout{\pi_i.P_i}{Q_i}} | \prod_{j \in J} \sigma.P_j}}} {\mob} {k}
\redtime 
\nodep n {\conf \I {\prod_{i \in I}Q_i} | \prod_{j \in J}P_j} {\mob} {k'}}
\]
with 
$M = \nodep n {{\conf \I {{\prod_{i \in I}\timeout{\pi_i.P_i}{Q_i}} | \prod_{j \in J} \sigma.P_j}}} {\mob} {k}$
and 
$M' = \nodep n {\conf \I {\prod_{i \in I}Q_i} | \prod_{j \in J}P_j} {\mob} {k'}$. 
Suppose there exists $M''$ such that $M \redtime M''$. Notice that, due to 
its structure, network $M$ may 
perform a timed reduction only by an application of rule \rulename{timemob}.
Thus, by  rule \rulename{timemob} we would have: 
\[
\Txiombis
{\nodep n {{\conf \I {{\prod_{i \in I}\timeout{\pi_i.P_i}{Q_i}} | \prod_{j \in J} \sigma.P_j}}} {\mob} {k} \not\redtau \Q \dist k {k''} \leq \delta}
{\nodep n {{\conf \I {{\prod_{i \in I}\timeout{\pi_i.P_i}{Q_i}} | \prod_{j \in J} \sigma.P_j}}} {\mob} {k}
\redtime 
\nodep n {\conf \I {\prod_{i \in I}{Q''}_i} | \prod_{j \in J}{P''}_j} {\mob} {k''}}
\]
with 
$M'' =  \nodep n {\conf \I {\prod_{i \in I}{Q''_i} | \prod_{j\in J}{P''_j}}}{\mob}{k''}$. By Lem.~\ref{lem:timed} it follows that 
$\prod_{i \in I}{Q_i} | \prod_{j\in J}{P_j}
\equiv \prod_{i \in I}{Q''_i} | \prod_{j\in J}{P''_j}$. 
Moreover, by triangular inequality it holds that 
$\dist{k'}{k''}\leq\dist{k}{k'}+\dist{k}{k''}\leq 2\delta$.

\item Let  $M\redtime M'$ by an application of 
  rule \rulename{timestat}. This case is similar to the previous one. 

\item Let $M\redtime M'$   by an application of \rulename{timepar}: 
\[
\Txiombis{M_1\redtime M_1' \Q M_2\redtime M_2'\Q M_1 | M_2 \not\!\redtau}{M_1 | M_2 \redtime M_1' | M_2'}
\]
with $M=M_1 | M_2$ and $M'=M_1' | M_2'$.
By inductive hypothesis we know that $M_1'$ and $M_2'$ are  unique, up to structural 
congruence, and up to 
node locations. 
So is $M_1' | M_2'$. 

 \item Let $M\redtime M'$ by an application of either  rule \rulename{res} 
or rule  \rulename{timestruct}. These cases are similar to the previous one. 
\end{itemize}
\hfill\qed

In order to prove maximal progress we need two simple lemmas. 
\begin{lemma}
\label{lem:maxprog1}
If $\prod_{i\in I}\nodep {n_i} {\conf {\I_i} {P_i}}{\mu_i}{h_i}
\redtau M$ then $\prod_{i\in I}\nodep {n_i} {\conf {\I_i} {P_i|Q_i}}{\mu_i}{h_i}
\redtau N$, for some $N$. 
\end{lemma}
\begin{lemma}
\label{lem:maxprog2}
If $\nodep n P \mu h \not\redtime$ then for any process $Q$ we have 
$\nodep n {P|Q} \mu h \not\redtime$. 
\end{lemma}

\paragraph{\textbf{Proof of Prop.~\ref{prop:maxprog}.}}
The proof is by rule induction on why  $M\redi M'$.
%% We sketch the proof. 
%%We remark that  $M\redi  M'$ implies that $M\not \equiv \zero$ and hence rule \rulename{timezero} is not enabled. 
\begin{itemize}
\item 
Let $M\redi M'$ by an application of  rule \rulename{pos}. 
Then $M = \nodep{n}{\conf \I {@(x).P}}{\mu}{h}$ and the 
 only rules that would allow $M$ to perform a timed reduction $\redtime$
are \rulename{timestat} and \rulename{timemob}. 
However, the premises of both rules are not satisfied by $M$.  Thus 
$M \not \redtime$. 
%% Similar reasoning can be done for rules \rulename{sensread}, \rulename{actunchg} and \rulename{actchg}.
%%% MASSIMO: E chi se ne frega???? Che c'entra con la prova???

\item 
Let $M\redi M'$ by an application of  rule \rulename{sensread}. 
Then $M= \nodep n {\conf \I {\rsens x s.P}}{\mu}{h}$ and this case is 
similar to the first one. 

\item Let $M\redi M'$ by an application of  rule \rulename{actunchg}. 
Then $M= \nodep n {\conf \I {\wact v a.P}}{\mu}{h} $ and this case is similar
to the first one. 

\item Let $M\redi M'$ by an application of  rule \rulename{actchg}. 
Then $M= \nodep n {\conf \I {\wact v a.P}}{\mu}{h} $ and this case is similar
to the first one. 

\item 
Let $M\redi M'$ by  an application of rule \rulename{loccom}, 
 because $M \redtau M'$.  
Then $M = \nodep{n}{\conf \I 
{\timeout{\OUT c v .P}{R}
\newpar 
{\timeout{\LIN c x .{Q}}{S}} }}
{\mu}{h}$  and the 
 only rules that would allow $M$ to perform a timed reduction $\redtime$
are \rulename{timestat} and \rulename{timemob}. As $M \redtau M'$ none of these
rules can fire. Thus $M \not \redtime$. 
\item   
Let $M\redi M'$ by  an application of rule \rulename{globcom}, because $M \redtau M'$. Then $M = \nodep{n}{\conf \I {\timeout{\OUT c v .P}{R}}}{\mu}{h} \; | \;  \nodep{m}{\conf \I {\timeout{\LIN c x .{Q}}{S}}}{\mu}{k}$ and 
the 
 only rule that would allow $M$ to perform a timed reduction $\redtime$
is \rulename{timepar}. However, since $M \redtau M'$ this rule cannot
fire and $M\not \redtime$.  

\item 
Let $M\redi M'$ by an application of rule \rulename{parp}. This means 
that 
 \begin{center} 
$M = \prod_{i\in I}\nodep {n_i} {\conf {\I_i} {P_i | Q_i}}{\mu_i}{h_i}
\redi 
\prod_{i \in I}\nodep {n_i} {\conf {\I'_i} {P'_i| Q_i}}{\mu'_i}{h'_i}=M'$
\end{center}
because 
$\prod_{i \in I}\nodep {n_i} {\conf {\I_i} {P_i}}{\mu_i}{h_i} \redi 
\prod_{i \in I}\nodep {n_i} {\conf {\I'_i} {P'_i}}{\mu'_i}{h'_i}.$
By inductive hypothesis we have that $\prod_{i\in I}\nodep {n_i} {\conf {\I_i} {P_i}}{\mu_i}{h_i} \not\redtime$. We recall that rule \rulename{timepar} is the only 
one yielding timed reductions on parallel networks. Thus, if  $\prod_{i\in I}\nodep {n_i} {\conf {\I_i} {P_i}}{\mu_i}{h_i} \not\redtime$ it means that 
rule \rulename{timepar} could not be applied. 
There are only two possibilities. 
\begin{itemize}
\item Either $\prod_{i\in I}\nodep {n_i} {\conf {\I_i} {P_i}}{\mu_i}{h_i}
\redtau N$, for some $N$.  Then by Lem.~\ref{lem:maxprog1} we obtain $M \redtau N'$, for some $N'$. As rule \rulename{timepar} is the only rule 
yielding timed reductions from parallel networks, it follows  that 
$M \not\redtime$. 
\item Or $\nodep {n_j} {\conf {\I_j} {P_j}}{\mu_j}{h_j} \not\redtime$, for
some $j \in I$.  Then by Lem.~\ref{lem:maxprog2} we have $\nodep {n_j} {\conf {\I_j} {P_j | Q_j}}{\mu_j}{h_j} \not\redtime$. As rule \rulename{timepar} is the only rule 
yielding timed reductions from parallel networks, it follows that 
$M \not\redtime$. 
\end{itemize}
\item 
Let $M\redi M'$ by an application of rule \rulename{parn}. 
Then $M = M_1 | M_2$ for some $M_1$ and $M_2$, with $M_1 \red_{\omega} M'_1$, 
$\omega \in \{ \tau , a \}$ for some actuator name $a$, and 
$M = M_1 | M_2 \red_{\omega} M'_1 | M_2 = M'$. By rule induction 
the sub-network $M_1$ cannot perform a timed reduction $\redtime$. 
As \rulename{timepar} is the only rule for deriving timed reduction 
of parallel networks it follows that also $M$ cannot perform a 
timed reduction $\redtime$. 

\item 
Let $M\redi M'$ by an application of rules \rulename{res} 
 and \rulename{struct}.  These cases are similar to the previous one.  
\end{itemize}
\hfill\qed

%%\begin{lemma}
%%[Well-formedness preservation]
%%\label{lem:well_form}
%%If $M$ is a  network and $M\redi M'$, or $M\redtime M'$, then $M'$ is a well-formed network.
%%\end{lemma}
%%
%%\begin{proof}
%%The thesis follows by induction on the length of the derivation of $M\redi M'$, and $M\redtime M'$.
%%Here we just propose a proof sketch.
%%
%%First of all we notice that none of the reduction rules in Table~\ref{reduction} changes the node names nor creates new nodes.
%%Hence we are sure that, whichever derivation we are considering, if $M$ does not contain two nodes with the same name so does $M'$. 
%%Same arguments hold for actuators and node-dependent sensors names.  
%%Moreover, no rules changes the domain of the interfaces. 
%%Finally, no rule changes nodes mobility tag from stationary to mobile.
%%Thus, if $M$ contains any stationary node with some location-dependent sensors, so does $M'$.
%%\end{proof}

\paragraph{\textbf{Proof of Prop.~\ref{prop:patience}.}}
The proof is by contradiction.
We suppose there is no $N$ such that 
$M \redtime N$ and we prove that there is $M'$ such that $M \redi M'$. 
We proceed by induction on the structure of $M$. 
\begin{itemize}
\item 
 Let $M = \zero$. This case is not admissible because by an application 
or rule \rulename{timezero} we derive $M \redtime M$.
\item Let $M = \nodep n P \mu h$.  As $M \not\redtime$  and \rulename{timestat} and 
\rulename{timemob} are the 
the only rules that could be used to derive a timed reduction from $M$, 
it follows that there are two possibilities. 
\begin{itemize}
\item Either $M \redtau M'$, for some $M'$, and we are done.
\item Or $P$ has not the proper structure for applying rule \rulename{timestat}
or rule \rulename{timemob}. This means that $P \equiv P_1 | P_2$, with $P_1= \rho. P'_1$  and 
$\rho \in \{ @(x), s?(x), a!v \}$. In this case, by an application 
of one among the rules \rulename{pos}, \rulename{sensread}, 
\rulename{actunchg}, and \rulename{actchg}, followed by an application of rule 
\rulename{parp}, we can infer $M \redi M'$, for some $M'$. 
\end{itemize}
\item Let $M = M_1 | M_2$, for some $M_1$ and $M_2$. As $M \not\redtime$ 
and  \rulename{timepar}
is the only rule which could be used to derive a timed reduction from $M$, 
it follows that there are two possibilities. 
\begin{itemize}
\item Either $M \redtau M'$, for some $M'$; hence $M \redi M'$ and we are done. 
\item Or at least one among $M_1$ and $M_2$ cannot perform a timed reduction. 
Suppose $M_1 \not\redtime$; by inductive hypothesis there is $M'_1$ such 
that $M_1 \redi M'_1$. By an application of rule \rulename{parn} we derive
$M \redi M'_1 | M_2$. 
\end{itemize}
\item Let $M = \res c {M_1}$. This case requires an easy application of the 
inductive hypothesis. 
\end{itemize} 
\hfill\qed

Next step is the proof of the well-timedness property. 
We need a couple of technical lemmas. 
\begin{definition}
Let us define a function $\pfxi{}$ that given a process $P$ returns
an upper bound to the number of the \emph{untimed prefixes} that can give rise
to an instantaneous reduction when the process is plugged in a node. 
\[
\begin{array}{l@{\hspace*{10mm}}l@{\hspace*{10mm}}l}
\pfxi{\nil}\deff 0  &  \pfxi{\sigma.P} \deff 0 &\pfxi{\rho.P} \deff 1 + \pfxi{P} \Q(\textrm{if $\rho\neq \sigma$)}\\
\pfxi{X}\deff\infty  &  \pfxi{\fix{X}P}\deff\pfxi{P} &
\pfxi{\timeout{\pi.P}Q}\deff 1+\pfxi{P}\\
\multicolumn{3}{c}{
\pfxi{\match b P Q}\deff \operatorname{max}(\pfxi{P},\pfxi{Q})
\Q\Q \Q\Q \pfxi{P\newpar Q}\deff \pfxi{P} + 
\pfxi{Q} }
\end{array}
\]
\end{definition}

\begin{lemma}
\label{lem:wt1}
For any closed process $P$,  $\pfxi{P}$ is finite. 
\end{lemma}
\begin{proof} The proof is by structural induction on $P$. The only delicate
case is when $P=\fix{X}{P_1}$, as $P_1$ may contain the process variable $X$ and $\pfxi{X}=\infty$. However,  in our calculus we only admit time-guarded recursion. Thus, 
 $X$ may occur in $P_1$ only if guarded by at least one $\sigma$ prefix, and 
$\pfxi{\sigma.Q}=0$, for any $Q$. It follows that $\pfxi{\fix{X}{P_1}} \in \mathbb{N}$, for any $P_1$. 
\end{proof}

\begin{definition}
\label{def:redi}
Let us define a function $\ri{}$ that given a network $M$ returns
an upper bound to the number of  consecutive instantaneous reductions
that $M$ may perform:
\[
\begin{array}{rcl@{\hspace*{2cm}}rcl}
 \ri{\zero} & \deff & 0 &	
 \ri{\nodep{n}{\conf \I P}{\mu}{h}} & \deff &  \pfxi{P}\\[3pt]
\ri{\res c M} & \deff &  \ri{M} & 
%\multicolumn{3}{c}{
%%\ri{M[{s@h} \mapsto v]} \deff \ri{M} \Q\Q 
\ri{M | N} & \deff & \ri{M} + \ri{N} \enspace .
\end{array}
\]
\end{definition}

\begin{lemma}For any network $M$, $\ri{M}$ is finite. 
\label{lem:wt2}
\end{lemma}
\begin{proof}
The proof is by structural induction on $M$. The proof is straightforward. 
The only interesting case is when $M$ is a node; here the result follows from 
Lem.~\ref{lem:wt1}.
\end{proof}

\paragraph{\textbf{Proof of Prop.~\ref{prop:welltime}.}}
The proof is by induction on the structure of $M$ and follows directly
from Lem.~\ref{lem:wt2}.\hfill\qed

%%%%%%%%%%%%%%%%%%%%%%%%%%%%%%%%%%%%%%%%%%%%%%%%%%%%%%%%%%%%%%%%%%%%%

\subsection{Proofs of Sec.~\ref{case_study}}
Here we prove Prop.~\ref{caseproperty} which formalises the properties 
of the system defined in Table~\ref{case_study_home}.
\begin{lemma} 
\label{lem:case-study}
If $Sys \; (\redi^{\ast}\redtime)^{\ast} \; Sys'$ then  
$Sys' \equiv Phone' | Home'$ where:
\begin{itemize}
\item $Phone' =  \nodep{n_P}{\conf {\I_P}{BoilerCtrl 
\newpar LightCtrl}}{\mob}{l'}$, for some $l'$, with $\I_P(mode)=\mathsf{auto}$
\item $Home' = LR_1 | LR_2 |  BoilerMng$, for some $LR_1$ and $LR_2$,
\item %%either 
$BoilerMng =\nodep{n_{B}}{\conf {\I_{B}}{Auto}}{\stat}{loc2}$, with 
$\I_B(temp)=\Theta$.  
%%or $BoilerMng' =  \nodep{n_{B}}{\conf {\I_{B}}Manual}{\stat}{loc2}$.
\end{itemize}
\end{lemma}
\begin{proof}
The proof is by mathematical induction on the integer $j$ such that 
$Sys \; (\redi^{\ast}\redtime)^{j} \; Sys'$. 

The case $j=0$ is trivial. 

Let us move on the inductive case. Let 
$Sys \; (\redi^{\ast}\redtime)^{j} \; Sys_1$, for $j>0$.  By inductive
hypothesis we have:
$Sys_1 \equiv Phone_1 | Home_1$ where:
\begin{itemize}
\item $Phone_1 =  \nodep{n_P}{\conf {\I_P}{BoilerCtrl \newpar
 LightCtrl}}{\mob}{l_1}$, for some $l_1$, with $\I_P(mode)=\mathsf{auto}$
\item $Home_1 = LR_1 | LR_2  | BoilerMng$, for some $LR_1$ and  $LR_2$,
\item %%either 
$BoilerMng =\nodep{n_{B}}{\conf {\I_{B}}{Auto}}{\stat}{loc2}$,  with 
$\I_B(temp)=\Theta$.  
%%or $BoilerMng_1 =  \nodep{n_{B}}{\conf {\I_{B}}Manual}{\stat}{loc2} \enspace .$
\end{itemize}
%%We recall that by 
%%Prop.~\ref{prop:welltime} our systems are always well-timed.
%%Moreover, by Prop.~\ref{prop:patience} they also satisfy patience, and 
%%a timed reduction will fire if no instantaneous reduction can. 
We  recall that sensor changes 
are not modelled in the reduction semantics as they require the 
intervention of the physical environment. So, the value of these 
sensor will remain unchanged during the reduction sequence. 
Thus, we want to show that whenever $Sys_1 \: \redi^{\ast}\: \redtime Sys'$, 
then $Sys'$ has still the same structure as $Sys_1$. 
Let us consider a portion of $Sys_1$ composed by the phone 
and the boiler manager. The, we have the following sequence 
of instantaneous reductions. We recall that $\I_P(mode) = \mathsf{auto}$ and 
$\I_B(temp)=\Theta$. 
\[
\begin{array}{rl}
\multicolumn{2}{l}{Phone_1 \q \big| \q BoilerMng}\\[3pt]
\q = &  
 \nodep{n_P}{\conf {\I_P}{BoilerCtrl \newpar LightCtrl}}{\mob}{l_1} \q \big| \q \nodep{n_{B}}{\conf {\I_{B}  }{Auto}}{\stat}{loc2}  \\[3pt]
\Q \redi 
& \nodep{n_P}{\conf {\I_P}{\timeout{\OUT {b}{\mathsf{auto}} .\sigma. BoilerCtrl}{\ldots} \newpar LightCtrl}}{\mob}{l_1} \q \big| \q \nodep{n_{B}}{\conf {\I_{B} }{Auto}}{\stat}{loc2} \\[3pt]  
\redi & \nodep{n_P}{\conf {\I_P}{\sigma. BoilerCtrl \newpar LightCtrl}}
{\mob}{l_1} \q \big| \q \nodep{n_{B}}{\conf {\I_{B} }{TempCtrl }}{\stat}{loc2}   \\[3pt]
\redi
& \nodep{n_P}{\conf {\I_P}{\sigma. BoilerCtrl \newpar LightCtrl}}{\mob}{l_1} 
\q \big| \q \nodep{n_{B}}{\conf {\I_{B}}{\wact{\mathsf{off}}{boiler}.\sigma.Auto}}{\stat}{loc2}  \\[3pt]
 \redi
& \nodep{n_P}{\conf {\I_P}{\sigma. BoilerCtrl \newpar LightCtrl}}{\mob}{l_1} 
\q \big| \q \nodep{n_{B}}{\conf {\I_{B}[boiler \mapsto \mathsf{off}] }{\sigma.Auto }}{\stat}{loc2} \\[3pt]
=
& \nodep{n_P}{\conf {\I_P}{\sigma. BoilerCtrl \newpar LightCtrl}}{\mob}{l_1} 
\q \big| \q \nodep{n_{B}}{\conf {\I_{B}}{\sigma.Auto}}{\stat}{loc2}
\end{array}
\]
 Now, both
the phone and the boiler manager can only perform a timed reduction.
However, the whole system may have further instantaneous reductions depending whether the 
phone is in position to interact with the light managers of the house. 
In any case, thanks to (i) \emph{well-timedness} (Prop.~\ref{prop:welltime}), 
(ii)
\emph{patience} (Prop.~\ref{prop:patience}), (iii) 
rule \rulename{parn}, (iv) rule \rulename{struct} we will have a reduction
sequence  
of the form: 
\[
\begin{array}{l}
Phone_1 | LR_1 | LR_2   | BoilerMng\\
\Q\Q\Q\redi^{\ast} \: \redtime \; Phone' |  LR_1' | LR_2' | BoilerMng
\end{array}
\]
where $\I_P(mode)=\mathsf{auto}$ and $\I_B(temp) = \Theta$ (the reduction 
semantics cannot change sensor values) and $Phone'$ is exactly as $Phone_1$ except for the fact that is located at a 
possibly new location $l'$, with $\dist {l_1} {l'}  \leq 1$. 
\end{proof}

\paragraph{\textbf{Proof of Prop.~\ref{caseproperty}.}}
From Lem.~\ref{lem:case-study} we know that $Sys'$ preserves  
the structure of $Sys$ and also the value of its 
sensors. Let us prove the four cases of the proposition,  one by one. 
\begin{enumerate}
\item Let us consider the evolution of $Sys'[mode \mapsto \mathsf{man}]$. 
By inspection of the definitions in Table~\ref{case_study_home} it is 
easy to derive that 
\[
\begin{array}{rl}
\multicolumn{2}{l}{Sys'[mode \mapsto \mathsf{man}]}\\
\equiv & Phone'[mode \mapsto \mathsf{man}] | LR_1 | LR_2 |   \nodep{n_{B}}{\conf {\I_{B}}{Auto}}{\stat}{loc2}\\
\redi^{\ast}  \redtime  & Phone'[mode \mapsto \mathsf{man}] |  LR'_1 | LR'_2 |   \nodep{n_{B}}{\conf {\I'_{B}}{Manual}}{\stat}{loc2}
\end{array}
\]
with $\I'_B(boiler)= \mathsf{on}$. 
\item Let us consider the evolution of $Sys'[temp \mapsto t]$, with 
$t < \Theta$. We spell out this case in more detail. 
We recall that the sensor $mode$ of the phone is set to $\mathsf{auto}$. 
We also recall that by applying rules \rulename{parn} and \rulename{struct}, if a parallel component can  
execute an instantaneous reduction then the whole network can 
execute the same reduction. 
Thus, in following we concentrate on the reductions deriving from 
the  phone and 
the boiler manager when the the environment changes the value of sensor
 $temp$ to a value $t < \Theta$.
%%Let $Q$ be a process derived from $LightCtrl$, $M$ be a network derived from $LightR1 | LightCorr | LightR2$, and $l $ be a location. We have the following 
%%derivation: 
\[
\begin{array}{rl}
\multicolumn{2}{l}{Phone' \q \big| \q BoilerMng[temp \mapsto t]}\\[3pt]
\q = &  
 \nodep{n_P}{\conf {\I_P}{BoilerCtrl \newpar LightCtrl}}{\mob}{l'} \q \big| \q \nodep{n_{B}}{\conf {\I_{B}  [temp \mapsto t] }{Auto}}{\stat}{loc2}  \\[3pt]
\Q \redi 
& \nodep{n_P}{\conf {\I_P}{\timeout{\OUT {b}{\mathsf{auto}} .\sigma. BoilerCtrl}{\ldots} \newpar LightCtrl}}{\mob}{l'} \q \big| \q \nodep{n_{B}}{\conf {\I'_{B} }{Auto}}{\stat}{loc2} \\[3pt]  
\redi & \nodep{n_P}{\conf {\I_P}{\sigma. BoilerCtrl \newpar LightCtrl}}{\mob}{l'} \q \big| \q \nodep{n_{B}}{\conf {\I'_{B} }{TempCtrl }}{\stat}{loc2}   \\[3pt]
\redi
& \nodep{n_P}{\conf {\I_P}{\sigma. BoilerCtrl \newpar LightCtrl}}{\mob}{l'} 
\q \big| \q \nodep{n_{B}}{\conf {\I'_{B}}{\wact{\mathsf{on}}{boiler}.\sigma.Auto}}{\stat}{loc2}  \\[3pt]
 \redi
& \nodep{n_P}{\conf {\I_P}{\sigma. BoilerCtrl \newpar LightCtrl}}{\mob}{l'} 
\q \big| \q \nodep{n_{B}}{\conf {\I'_{B}[boiler \mapsto \mathsf{on}] }{\sigma.Auto }}{\stat}{loc2} \\[3pt]
=
& \nodep{n_P}{\conf {\I_P}{\sigma. BoilerCtrl \newpar LightCtrl}}{\mob}{l'} 
\q \big| \q \nodep{n_{B}}{\conf {\I''_{B}}{\sigma.Auto}}{\stat}{loc2}
\end{array}
\]
where $\I''_B(temp)=t < \Theta$ and $\I''_B(boiler)=\mathsf{on}$. Now, both
the phone and the boiler manager can only perform a timed reduction.
However, the whole system may have further instantaneous reductions depending whether the 
phone is in position to interact with the light managers of the house. 
In any case, thanks to \emph{well-timedness} (Prop.~\ref{prop:welltime}) and 
\emph{patience} (Prop.~\ref{prop:patience}) we will have a reduction
sequence  
of the form: 
\[
Sys'[temp \mapsto t] \redi^{\ast} \: \redtime \; Phone'' |  LR'_1 | LR'_2 |   \nodep{n_{B}}{\conf {\I''_{B}}{Auto}}{\stat}{loc2}
\]
where $\I''_B(temp)=t$, $\I''_B(boiler)= \mathsf{on}$, and the mobile phone may have moved to a new location $l''$, with $\dist {l'} {l''} =1$. 
\item Let us consider the evolution of $Sys'[temp \mapsto t]$, with 
$t \geq \Theta$. Here, similarly to the previous case, we can derive:
\[
\begin{array}{l}
Sys'[temp \mapsto t] \equiv Phone' | LR_1 | LR_2 |  \nodep{n_{B}}{\conf 
{\I_{B}[temp \mapsto t]}{Auto}}{\stat}{loc2}\\[3pt]
\Q\Q\Q\redi^{\ast} \: \redtime \; Phone'' |  LR'_1 | LR'_2 |  \nodep{n_{B}}{\conf {\I'_{B}}{Auto}}{\stat}{loc2}
\end{array}
\]
with $\I'_B(temp)=t$ and $\I'_B(boiler)= \mathsf{off}$. 

\item We prove only the implication from left to right. The other is similar. 
We know that $Sys' \redi^{\ast} Sys''\downarrow_{\wact {\mathsf{on}} {light_1}}$. 
By Lem.~\ref{lem:case-study} we know the structure of $Sys'$.
We recall that initially,  in $Sys$, the actuator $light_1$ is set to $\mathsf{off}$. Notice also that this actuator is exclusively managed by the $LightMng1$
component, via the process $L_1$, running at the stationary node $n_1$, 
located al $loc1$. More precisely, the 
actuator $light_1$ can be modified by $L_1$ only after a synchronisation
at the short-range channel $c_1$. We recall that  $\rng {c_1} = 0$. 
We also recall that mobile nodes can  change their location only 
by executing a timed reduction via rule \rulename{timemob}. We  fixed 
$\delta=1$, 
which is the maximum distance that a mobile node can afford within a time unit. Thus, 
if  $Sys' \redi^{\ast} Sys''\downarrow_{\wact {\mathsf{on}} {light_1}}$
there are two possibilities:
\begin{itemize}
\item either the mobile phone  is currently located at $loc1$; 
\item or the mobile phone  was located at $loc1$ in the previous time interval,
and in the current time interval it is in a location $l'$,
 with $\dist {l'} {loc_1}=1$, 
as $\delta=1$. Notice that the phone cannot be farther than that otherwise
the timeout in $L_1$ would have already switched off  $light_1$.
\end{itemize}

In the first case, the light manager $LightMng2$, located at $loc4$, has necessarily set the actuator
$light_2$ to $\mathsf{off}$. This is because $\rng {c_2} = 0$, 
$\dist {loc1}{loc4} =3$, and the only manner to switch on $light_2$ is
to place the mobile phone at $loc4$. However, as the mobile phone is 
currently at $loc1$, and $\delta=1$, this could have happened only 3 time instants ago. By that time, the timemout in $LightMng2$ (more precisely in $L_2$) has already 
switched off the light. 

The second case, when the mobile phone is currently located at some location 
$l'$, with $\dist {l'} {loc_1}=1$, is similar. This is because $\dist {loc1}{loc4} =3$, and by triangular inequality $\dist {l'} {loc4} \geq 2$. Thus, 
the phone is far enough to ensure that timeout of $L_2$ already fired to 
switch off $light_2$. 
\end{enumerate}
\hfill\qed

\paragraph*{\textbf{Proof of Prop.~\ref{prop:SYS-barb}.}}
By Prop.~\ref{prop:SYS-bis} and Thm.~\ref{thm:sound}. 
\hfill\qed

%%%%%%%%%%%%%%%%%%%%%%%%%%%%%%%%%%%%%%%%%%%%%%%%%%%%%%%%%%%%%%%%%%%%%%%%%%%

\subsection{Proofs of Sec.~\ref{lab-sem}}
This section is devote to the proof of the \emph{Harmony Theorem\/}, i.e.\ Thm.~\ref{thm:harmony}. We start with a  technical lemma that provides
the structure of a process depending on its possible actions.

\begin{lemma}
\label{struc_p}
Let $P$ be a process.
\begin{enumerate}

\item  
\label{struc_p_time}
If $P\trans{\sigma}P'$ then 
$P\equiv\prod_{i\in I} \timeout{\pi_i.P_i}Q_i\newpar \prod_{j\in J} \sigma.P_j$
and
$P'\equiv \prod_{i\in I} Q_i\newpar \prod_{j\in J}  P_j$, for appropriate index
sets, prefixes and processes. 

%%there exist $I,J$ and , for all $i\in I$ and $j\in J$, $\pi_i,P_i,Q_i,P_j$ s.t. 
\item  
\label{struc_p_read}
If $P\trans{\rsensa v s}P'$ then there are $P_1$ and $Q$ such that $P\equiv \rsens x s .P_1\newpar Q$ and $P'\equiv P_1\subst v x \newpar Q$.

\item  
\label{struc_p_write}
If $P\trans{\wsensa v {a}}P'$ then there are $P_1$ and $Q$ such that $P\equiv \wact v {a} .P_1\newpar Q$ and $P'\equiv P_1\newpar Q$.

\item 
\label{struc_p_pos}
If $P\trans{@h}P'$ then there are $P_1$ and $Q$ such that $P\equiv @(x) .P_1\newpar Q$ and $P'\equiv P_1\subst h x \newpar Q$.

\item 
\label{struc_p_in}
If $P\trans{\out c v}P'$ then there are $P_1$, $Q_1$ and $Q$ such that $P{\equiv} \timeout{\OUT c v .P_1}Q_1 \newpar Q$ and $P'{\equiv} P_1\newpar Q$.

\item 
\label{struc_p_out} 
If $P\trans{\inp c v}P'$, then there exist $P_1,Q_1,Q$ s.t. $P\equiv \timeout{\LIN c x .P_1}Q_1\newpar Q$ and $P'\equiv P_1\subst v x \newpar Q$.

\item 
\label{struc_p_tau} 
If $P\trans{\tau}P'$ then there are $P_1$, $P_2$, $Q_1$, $Q_2$, $R$, and $c$ with $\rng c=-1$, such that 
$P\equiv \timeout{\LIN c x .P_1}Q_1\newpar \timeout{\OUT c v .P_2}Q_2 \newpar R$ and $P'\equiv P_1\subst v x \newpar P_2 \newpar R$.

\end{enumerate}
\end{lemma}

\begin{proof}
%%We prove all seven cases,  one by one. 
% in Table~\ref{tab:lts_processes}.

Let us start with item (\ref{struc_p_time}). We proceed by rule induction 
 on why $P\trans{\sigma}P'$

\begin{itemize}
\item Let $P\trans{\sigma}P'$   by an application of  rule \rulename{TimeNil};  then the thesis is immediate for $I=J=\emptyset$.

\item Let $P\trans{\sigma}P'$  by an application of  rule \rulename{Delay}:
\[
\Txiombis{-}{\sigma.P_1\trans{\sigma}P_1}
\]
with $P=\sigma.P_1$ and $P'=P_1$.
Thus,  for $I=\emptyset$ and $J=\{1\}$ we have $P = \sigma. P_1 \equiv \nil\newpar\sigma.P_1$ and $P' = P_1 \equiv \nil | P_1$.

\item Let $P\trans{\sigma}P'$  by an application of  rule \rulename{Timeout}:
\[
\Txiombis{-}{\timeout{\pi_1.P_1}{Q_1} \trans{\sigma} Q_1}
\]
with $P=\timeout{\pi_1.P_1}{Q_1}$ and $P'=Q_1$.
Thus,  for $I=\{1\}$ and $J=\emptyset$  we have $P = \timeout{\pi_1.P_1}{Q_1} 
\equiv \timeout{\pi_1.P_1}{Q_1} | \nil$ and $P' =Q_1 \equiv Q_1 | \nil $.

\item Let $P\trans{\sigma}P'$  by an application of  rule \rulename{TimeParP}:
\[
\Txiombis{R_1 \trans{\sigma} R_1' \q R_2 \trans{\sigma}R_2'}{R_1 \newpar R_2 \trans{\sigma} R_1' \newpar R_2'}
\]
with $P=R_1\newpar R_2$ and $P'=R_1'\newpar R_2'$.
By inductive hypothesis, there exist $I,J, I'$ and $J'$ such that 
$R_1\equiv  \prod_{ i\in I} \timeout{\pi_i.P_i}Q_i\newpar \prod_{ j\in J} 
\sigma.P_j$, 
$R_1'\equiv  \prod_{ i\in I }  Q_i\newpar  \prod_{ j\in J} P_j$, 
$R_2\equiv  \prod_{ {i'}\in I' } \timeout{\pi_{i'}.P_{i'}}Q_{i'}\newpar \prod_{ {j'}\in J' } \sigma.P_{j'}$
and 
$R_2'\equiv  \prod_{ {i'}\in I' } Q_{i'}\newpar  \prod_{ {j'}\in J' } P_{j'}$.
%Now, by structural congruence 
%\[
%R_1\newpar R_2 \equiv 
%( \prod_{ i\in I } \timeout{\pi_i.P_i}Q_i\newpar  \prod_{ j\in J }  \sigma.P_j) \newpar ( \prod_{ i\in I' } \timeout{\pi'_i.P'_i}Q'_i\newpar 
% \prod_{ j\in J' } \sigma.P'_j)
%\]
%and 
%\[
%R_1' \newpar R_2' \equiv ( \prod_{ i\in I} Q_i\newpar \prod_{ j\in J} P_j) \newpar ( \prod_{ i\in I' } Q'_i\newpar 
%\prod_{ j\i-n J' } P'_j).
%\]
Fo concluding this case we choose as index sets $\bar{I}=I\cup I'$ and $\bar{J}=J\cup J'$.

%\item $P\trans{\sigma}P'$ is derived by applying transition rule \rulename{Then}.
%Hence, there exists a proof for closed transition rule 
%\[
%\Txiombis{\bool{b}=\true \Q Q \trans{\sigma} Q'}{\match b Q R \trans{\sigma} Q'}
%\]
%with $P=\match b Q R$ and $P'=Q'$.
%Structural congruence gives $\match b Q R\equiv Q$.
%Therefore, the thesis follows by inductive hypothesis on $Q\trans{\sigma}Q'$.

%The cases of $P\trans{\sigma}P'$ derived with transition rules \rulename{Else} and 
%\rulename{Fix} are analogous and obtained by exploiting structural congruence.

\item Let $P\trans{\sigma}P'$  by an application of  rule \rulename{Fix}:
\[
\Txiombis{P_1{\subst {\fix{X}P_1} X} \trans{\sigma} P_2}
{\fix{X}P_1\trans{\sigma}P_2}
\]
with $P=\fix{X}P_1$ and $P'=P_2$.
By inductive hypothesis, there exist $I$ and $J$ such that 
$P_1{\subst {\fix{X}P_1} X} \equiv  \prod_{ i\in I} \timeout{\pi_i.P_i}Q_i\newpar \prod_{ j\in J} 
\sigma.P_j$ and $P_2\equiv \prod_{ i\in I}  Q_i\newpar \prod_{ j\in J} P_j$. 
By structural congruence we have  $\fix{X}P_1 \equiv P_1{\subst {\fix{X}P_1} X}$ and  therefore  $P= \fix{X}P_1 \equiv  \prod_{ i\in I} \timeout{\pi_i.P_i}Q_i\newpar \prod_{ j\in J} 
\sigma.P_j$.
\end{itemize}
Let us prove now the item (\ref{struc_p_read}) of the proposition.
We proceed by rule induction on why  $P\trans{\rsensa v s}P'$. 
\begin{itemize}
\item Let $P\trans{\rsensa v s}P'$ by an application of  rule rule \rulename{Sensor}:
\[
\Txiombis{-}{ \rsens x  s .P_1 \trans{\rsensa v s} P_1 {\subst v x} }
\]
with $P=\rsens x  s .P_1$ and $P'=P_1 {\subst v x}$.
This case is easy. 

\item Let $P\trans{\rsensa v s}P'$  by an application of  rule \rulename{Fix}:
\[
\Txiombis{P_1{\subst {\fix{X}P_1} X} \trans{\rsensa v s} P_2}
{\fix{X}P_1\trans{\rsensa v s}P_2}
\]
with $P=\fix{X}P_1$ and $P'=P_2$. By inductive hypothesis there exist $P_3$ and $Q_1$ such that $P_1{\subst {\fix{X}P_1} X} \equiv \rsens x s .P_3\newpar Q_1$ and $P_2\equiv P_3\subst v x \newpar Q_1$. 
By structural congruence $P= \fix{X}P_1 \equiv P_1{\subst {\fix{X}P_1} X}\equiv \rsens x s .P_3\newpar Q_1$.

\item Let $P\trans{\rsensa v s}P'$ by an application of  rule \rulename{ParP}:
\[
\Txiombis{P_1 \trans{\rsensa v s} P_1'}{P_1 \newpar R \trans{\rsensa v s} P_1' \newpar R}
\]
with $P=P_1\newpar R$ and $P'=P_1'\newpar R$.
By inductive hypothesis there exist $P_2$ and $Q_1$ such that $P_1\equiv \rsens x s .P_2\newpar Q_1$ and $P_1'\equiv P_2\subst v x \newpar Q_1$, thus, the thesis holds for  $Q=Q_1\newpar R$.
\end{itemize}
The  cases (\ref{struc_p_write}), (\ref{struc_p_pos}), (\ref{struc_p_in}) and (\ref{struc_p_out}) are analogous to previous items.

We prove (\ref{struc_p_tau}). Let us suppose that $P\trans{\tau}P'$. We do a case analysis.
\begin{itemize}
\item Let $P\trans{\tau}P'$ by an application of  rule \rulename{Com}:
\[
\Txiombis
{P \trans{\out c v} P' \Q Q \trans{\inp c v} Q' \Q  \rng c = -1}
{P \newpar Q \trans{\tau} P' \newpar Q'}
\]
Then our result follows by application of the items (\ref{struc_p_in}) and (\ref{struc_p_out}) of the proposition. We need to work up to structural congruence.
\item Let $P\trans{\tau}P'$ by an application of  rule \rulename{ParP} or \rulename{Fix}. This case  is analogous to that the corresponding ones in (\ref{struc_p_read}). 
\end{itemize}
\end{proof}

\begin{lemma}
\label{lem:struc-in-out} \ 
\begin{enumerate}
\item 
\label{struc_in}
If $M\trans{\send{c}{v}{h}}M'$  then there exist $n,P,P',Q,\mu,N, \tilde g$ with $c \not \in  {\tilde g}$ such that 
$
M\equiv \res {\tilde g} \big( \nodep{n}{\conf \I{\timeout{\OUT c v .P}P'\newpar Q}}{\mu}{h} \, | \,  N \big)
$
and 
$
M'\equiv  \res {\tilde g}  \big( \nodep{n}{\conf \I{P\newpar Q}}{\mu}{h} \, | \, N
\big)$. 
\item 
\label{struc_out}
If $M\trans{\rec{c}{v}{h}}M'$ then there exist $n,P,P',Q,\mu,N,{\tilde g}$ with $c \not \in  {\tilde g}$ such that 
$
M\equiv \res {\tilde g} \big( \nodep{n}{\conf \I{\timeout{\LIN c x {.P}}P'\newpar Q}}{\mu}{h} \, | N \, \big)
$
and 
$
M'\equiv \res {\tilde g} \big( \nodep{n}{\conf \I{P\subst v x \newpar Q}}{\mu}{h} \, | \, N \big)$. 
\end{enumerate}
\end{lemma}

\begin{proof}
We proceed by rule induction to prove (\ref{struc_in}).
\begin{itemize}
\item Let $M\trans{\send{c}{v}{k}}M'$  by an application of  rule   rule \rulename{SndN}:
\[
\Txiombis{P \trans{\out c v} P'  \Q  \rng c \geq 0}{\nodep{n}{\conf \I P}{\mu}{k}\trans{\send{c}{v}{k}}\nodep{n}{\conf \I P'}{\mu}{k}}
\]
with $M=\nodep{n}{\conf \I P}{\mu}{k}$ and $M'=\nodep{n}{\conf \I P'}{\mu}{k}$.
Lem.~\ref{struc_p}(\ref{struc_p_in}) ensures that since $P\trans{\out c v}P'$ then there exist $P_1,Q_1,Q$ such that $P\equiv \timeout{\OUT c v .P_1}Q_1 \newpar Q$ and $P'\equiv P_1\newpar Q$.
This implies 
$M\equiv  \nodep{n}{\conf \I {\timeout{\OUT c v .P_1}Q_1 \newpar Q}}{\mu}{k}$  and
$M'\equiv  \nodep{n}{\conf \I {P_1\newpar Q}}{\mu}{k}$.

\item Let $M\trans{\send{c}{v}{k}}M'$  by an application of  rule \rulename{ParN}:  
\[
\Txiombis{M_1 \trans{\send{c}{v}{k}} M_1'}{M_1 | M_2 \trans{\send{c}{v}{k}} M_1' | M_2}
\]
with $M= M_1 | M_2$ and $M'=M_1' | M_2$.
By inductive hypothesis, since  $M_1\trans{\send{c}{v}{k}}M_1'$,  there exist $n,P_1,P_1',Q_1,\mu,k,N_1,{\tilde g} $  such that $c \not \in  {\tilde g}$ and 
$M_1\equiv \res {\tilde g} \nodep{n}{\conf \I{\timeout{\OUT c v P_1}P_1'\newpar Q_1}}{\mu}{k} | N_1$
 and
$M_1'\equiv \res {\tilde g} \nodep{n}{\conf \I{P_1\newpar Q_1}}{\mu}{k} | N_1$.
Hence  
$M\equiv \res {\tilde g}  \nodep{n}{\conf \I{\timeout{\OUT c v P_1}P_1'\newpar Q_1}}{\mu}{k} | N_1 | M_2$ 
and 
$M'\equiv \res {\tilde g} \nodep{n}{\conf \I{P_1\newpar Q_1}}{\mu}{k} | N_1 | M_2$
implying the thesis  for  $N= N_1 | M_2$.

\item Let $M\trans{\send{c}{v}{k}}M'$  by an application of  rule \rulename{Res}:
\[
\Txiombis{M_1 \trans{\send{c}{v}{k}} M_1'}{\res {c'} {M_1}  \trans{\send{c}{v}{k}} \res {c'} {M_1'} }
\]
with $c\neq c'$, $M= \res{c'} {M_1}$ and $M'=\res{c'}{M_1'}$.
By inductive hypothesis,   there exist $n,P_1,P_1',Q_1,\mu,k,N_1,{\tilde g} $ such that   $c \not \in  {\tilde g}$ and
$M_1\equiv \res {\tilde g} \nodep{n}{\conf \I{\timeout{\OUT c v P_1}P_1'\newpar Q_1}}{\mu}{k} | N_1$
and 
$M_1'\equiv \res {\tilde g} \nodep{n}{\conf \I{P_1\newpar Q_1}}{\mu}{k} | N_1$.
Hence $M\equiv \res{c'}\res {\tilde g}  \nodep{n}{\conf \I{\timeout{\OUT c v P_1}P_1'\newpar Q_1}}{\mu}{k} | N_1$
and $
M'\equiv \res{c'}\res {\tilde g} \nodep{n}{\conf \I{P_1\newpar Q_1}}{\mu}{k} | N_1 $.
Thus, since $c \not \in  \res{c'} \res {\tilde g} $ the thesis holds. 
\end{itemize}
The remaining case (\ref{struc_out}) is analogous by applying Lem.~\ref{struc_p}(\ref{struc_p_out}).
\end{proof}

%%\begin{lemma}
%%\label{struc_p_trans}
%%Let $P$ be a network.
%%If $P\trans{\lambda}P'$ and $P\equiv Q$ then there exists $Q'$ such that $Q\trans{\lambda} Q'$ and $P'\equiv Q'$.
%%\end{lemma}

%%\begin{proof}
%%The thesis follows by rule induction on the length of the proof of any inference $P\trans{\lambda}P'$ and a case analysis on $P\equiv Q$.
%%\end{proof}

\begin{lemma}
\label{struc_trans}
Let $M$ be a network.
If $M\trans{\alpha}M'$ and $M\equiv N$ then there exists $N'$ such that 
$N\trans{\alpha} N'$ and $M'\equiv N'$.
\end{lemma}

%%\begin{proof}
%%This result is well known in literature and the thesis follows by rule induction on the length of the proof of any inference $M\trans{\nu}M'$ and a case analysis on $M\equiv N$ taken together with previous Lem.~\ref{struc_p_trans}.
%%\end{proof}

\paragraph{\textbf{Proof of Thm.~\ref{thm:harmony}.}}
We have to prove the following sub-results:
\begin{enumerate}
\item \label{tau_red} If $M\trans{\tau}M'$ then $M\redtau M'$.
\item \label{red_tau} If $M\redtau M'$ then $M\trans{\tau}\equiv M'$.
\item \label{a_red} If $M\trans{a}M'$ then $M\red_{a} M'$.
\item \label{red_a} If $M\red_{a} M'$ then $M\trans{a}{\equiv}M'$.
\item \label{sigma_red} If $M\trans{\sigma}M'$ then $M\redtime M'$.
\item \label{red_sigma} If $M\redtime M'$ then $M\trans{\sigma}{\equiv}M'$.
\end{enumerate}
Let us start with the sub-result 
(\ref{tau_red}). The proof is by rule induction on why $M \trans{\tau} M'$. 
\begin{itemize}
\item Let $M\trans{\tau}M'$  by an application of  rule \rulename{SensRead}:
\[
\Txiombis{\I(s)=v \q P\trans{\rsensa v s}P'}{\nodep{n}{\conf \I P}{\mu}{h}\trans{\tau}\nodep{n}{\conf \I P'}{\mu}{h}}
\] 
with $M=\nodep{n}{\conf \I P}{\mu}{h}$ and $M'=\nodep{n}{\conf \I P'}{\mu}{h}$.
By Lem.~\ref{struc_p}(\ref{struc_p_read})
there exist $P_1,Q$ such that  $P\equiv \rsens x s .P_1\newpar Q$ and $P'\equiv P_1\subst v x \newpar Q$.
Then we can apply the reduction rules \rulename{sensread} and \rulename{parp}  inferring $M\redtau M'$ as required.

\item Let $M\trans{\tau}M'$  by an application of  rule \rulename{Pos}. 
This case follows by an application of Lem.~\ref{struc_p}(\ref{struc_p_pos}) 
together with  reduction rules \rulename{pos} and \rulename{parp}.

\item Let $M\trans{\tau}M'$  by an application of rule \rulename{LocCom}:
\[
\Txiombis{P  \trans{\tau} P' }{\nodep{n}{\conf \I P }{\mu}{h}\trans{\tau}\nodep{n}{\conf \I P' }{\mu}{h}}
\]
with $M=\nodep{n}{\conf \I P}{\mu}{k}$ and $M'=\nodep{n}{\conf \I P' }{\mu}{k}$.
By Lem.~\ref{struc_p}(\ref{struc_p_tau}), $P  \trans{\tau} P'$ ensures that there exist $P_1,P_2,Q_1,Q_2,R,c$ with $\rng c=-1$  such that 
$P\equiv \timeout{\LIN c x .P_1}Q_1\newpar \timeout{\OUT c v .P_2}Q_2 \newpar R$ 
and 
$P'\equiv P_1\subst v x \newpar P_2 \newpar R$.

By structural congruence, 
$
M\equiv \nodep{n}{\conf \I { \timeout{\LIN c x .P_1}Q_1\newpar \timeout{\OUT c v .P_2}Q_2 \newpar R }}{\mu}{k}
$
and analogously 
$
M'\equiv \nodep{n}{\conf \I {P_1\subst v x \newpar P_2 \newpar R}}{\mu}{k}.
$
Hence, by an application of  rules \rulename{struct} and \rulename{loccom}   we get $M \redtau M'$.

\item Let $M\trans{\tau}M'$  by an application of  rule \rulename{ActUnChg}. 
This case follows by an application of  Lem.~\ref{struc_p}(\ref{struc_p_write}) together with an application of reduction rules \rulename{actunchg} and \rulename{parp}.

%\item $M\trans{\tau}M'$ is derived by applying  rules \rulename{Move} and \rulename{Stop}. 
%Then the thesis immediately follows through application of reduction rules \rulename{move} and \rulename{stop}, respectively. 

\item Let $M\trans{\tau}M'$  by an application of  rule \rulename{ParN}:
\[
\Txiombis{M_1 \trans{\tau} M_1'}{M_1 | M_2 \trans{\tau} M_1' | M_2}
\]
with $M=M_1 | M_2$ and $M'=M'_1 | M_2$. 
By inductive hypothesis  $M_1\redtau M_1'$.
Therefore, by an application of rule \rulename{parn} we get $M \redtau M'$.  

\item Let $M\trans{\tau}M'$  by an application of  rule \rulename{Res}:
\[
\Txiombis{M_1 \trans{\tau} M_1'}{\res   {\tilde g}  M_1 \trans{\tau} \res   {\tilde g} M_1'}
\]
with $M=\res   {\tilde g} M_1$ and $M'=\res   {\tilde g} M'_1 $. 
By inductive hypothesis  $M_1\redtau M_1'$.
Therefore, by an application of the reduction rule \rulename{res}  we derive
$M \redtau M'$.

\item Let $M\trans{\tau}M'$  by an application of rule \rulename{GlbCom}:
\[
\Txiombis{M_1 \trans{\send{c}{v}{h}} M_1' \q M_2 \trans{\rec{c}{v}{k}}M_2' \q \dist h k \leq\rng c}{M_1 | M_2 \trans{\tau} M_1' | M_2'}
\]
with $M=M_1 | M_2$ and $M'=M'_1 | M_2'$. 
Lem.~\ref{lem:struc-in-out}(\ref{struc_in}) guarantees that since $M_1\trans{\send{c}{v}{h}}M_1'$ then
 $M_1\equiv \res  {\tilde g} \nodep{n}{\conf \I{\timeout{\OUT c v P}P'\newpar R}}{\mu}{h} | N$ 
and $M_1'\equiv \res {\tilde g} \nodep{n}{\conf \I{P\newpar R}}{\mu}{h} | N$, for some $n,P,P',R,\mu,h,N,{\tilde g}$.
At the same time, by Lem.~\ref{lem:struc-in-out}(\ref{struc_out}) there exist $m,Q,Q',$ $R',\mu,k,N',{\tilde g'} $ such that 
 $M_2\equiv \res {\tilde g'} \nodep{m}{\conf \I{\timeout{\LIN c x Q}Q'\newpar R'}}{\mu}{k} | N'$ 
 and 
$M_2'\equiv \res {\tilde g'} \nodep{m}{\conf \I{Q\subst v x \newpar R'}}{\mu}{k} | N'.$
Therefore, by applying reduction rule  \rulename{struct},  \rulename{res}, \rulename{glbcom},  \rulename{parp}  and  \rulename{parn} we can infer 
$M \redtau M'$.
\end{itemize}
Let us prove the sub-result 
(\ref{red_tau})  by rule induction on why  $M\redtau M'$. 
\begin{itemize}
\item Let $M\redtau M'$  by an applicaion of rule \rulename{sensread}:
\[
\Txiombis{\I(s)=v}{\nodep n {\conf \I {\rsens x s.P\newpar Q}}{\mu}{h}  \redtau \nodep n {\conf \I {P\subst v x \newpar Q}}{\mu}{h}}
\]
with $M=\nodep n {\conf \I {\rsens x s.P\newpar Q}}{\mu}{h}$ and  $M'=\nodep n {\conf \I {P\subst v x \newpar Q}}{\mu}{h}$. 
Hence, by rule  \rulename{Sensor} we have $  \rsens x  s .P \trans{\rsensa v s} P {\subst v x} $, by rule \rulename{ParP} we have $ \rsens x  s .P\newpar Q \trans{\rsensa v s} P {\subst v x} \newpar Q$ and finally by rule\rulename{SensRead} we have 
 $\nodep n {\conf \I {\rsens x s.P\newpar Q}}{\mu}{h}  \trans{\tau} \nodep n {\conf \I {P\subst v x \newpar Q}}{\mu}{h}$.

\item Let $M\redtau M'$   by applying rule \rulename{pos}:
\[
\Txiombis{-}{\nodep{n}{\conf \I {@(x).P\newpar Q}}{\mu}{h} \redtau \nodep{n}{\conf \I {P \subst x h \newpar Q}}{\mu}{h}}
\]
with $M=\nodep{n}{\conf \I {@(x).P\newpar Q}}{\mu}{h}$ and $M'=\nodep{n}{\conf \I {P \subst x h \newpar Q}}{\mu}{h}$.
We get $M \trans{\tau} \equiv M'$ by applying rules \rulename{PosP}, \rulename{ParP} and \rulename{Pos}.

\item Let $M\redtau M'$  by an application of rule \rulename{actunchg}. 
This case is similar to the previous one, by an application 
of the  transition rule \rulename{ActUnChg}.

\item Let $M\redtau M'$  by an application of rule \rulename{parp}:
\[ 
\Txiombis
{\prod_{i}\nodep {n_i} {\conf {\I_i} {P_i}}{\mu_i}{h_i} \red_{\tau} 
\prod_{i}\nodep {n_i} {\conf {\I'_i} {P'_i}}{\mu'_i}{h'_i} }
{\prod_{i}\nodep {n_i} {\conf {\I_i} {P_i | Q_i}}{\mu_i}{h_i} \red_{\tau} 
\prod_{i}\nodep {n_i} {\conf {\I'_i} {P'_i| Q_i}}{\mu'_i}{h'_i}}
\]
By inductive hypothesis we have
${\prod_{i}\nodep {n_i} {\conf {\I_i} {P_i}}{\mu_i}{h_i}\trans{\tau}\equiv 
\prod_{i}\nodep {n_i} {\conf {\I'_i} {P'_i}}{\mu'_i}{h'_i} }$.
%\begin{itemize}
%\item 
The $\tau$-transition can be derived using different transition rules. 
Suppose that  
${\prod_{i}\nodep {n_i} {\conf {\I_i} {P_i}}{\mu_i}{h_i}}$ $\trans{\tau}\equiv$
$\prod_{i}\nodep {n_i} {\conf {\I'_i} {P'_i}}{\mu'_i}{h'_i}$ by an application of rule  \rulename{SensRead} to node $n_j$, for some $j \in I$.  Then,  by using rule \rulename{ParP} to derive $P_j \newpar Q_j  \trans{\rsensa v s} P'_j \newpar Q_j$,  rule  
\rulename{SensRead} to derive 
$ \nodep {n_j} {\conf {\I_j} {P_j|Q_j}}{\mu_j}{h_j}\trans{\tau}_{\equiv} 
 \nodep {n_j} {\conf {\I'_j} {P'_j|Q_j}}{\mu'_j}{h'_j} $,  
 and rule \rulename{ParN}  to derive 
 $\prod_{i}\nodep {n_i} {\conf {\I_i} {P_i|Q_i}}{\mu_i}{h_i}$ $\trans{\tau}\equiv$
 $\prod_{i}\nodep {n_i} {\conf {\I'_i} {P'_i|Q_i}}{\mu'_i}{h'_i} $, we get 
$M \trans{\tau} \equiv M'$. 

The cases when the $\tau$-transition is derived by an application of the rules  \rulename{ActUnChg},  \rulename{Com}  and \rulename{Pos}  are similar.
 
%\item $\nodep n {\conf \I P}{\mu}{h}  \trans{\tau}_{\equiv} \nodep n {\conf {\I'} {P'}}{\mu'}{h'}$ is derived by rule  \rulename{Move} then $\mu=\mob$ and $\mu'=\mob$ with $ i<j   \leq \delta$ and $ j =i + \dist h {h'}$.  Therefore the rule  \rulename{Move} can be used  to derive also ${\nodep n {\conf \I {P \newpar Q}}{\mu}{h}  \trans{\tau}_{\equiv} \nodep n {\conf {\I'} {P' \newpar Q}}{\mu'}{h'}}$. 

%The case of rule \rulename{Stop} is analogous.
%\end{itemize}

%\item $M \redtau M'$ is derived by applying rule \rulename{move}.
%Then there exists a proof for closed reduction rule
%\[
%\Txiombis{ i<j   \leq \delta  \Q j =i + \dist h k }{\nodep{n}{\conf \I P}{\mob}{h} \redtau \nodep{n}{\conf \I {P}}{\mob}{k}}
%\]
%with $M=\nodep{n}{\conf \I {P}}{\mob}{k}$ and $M'=\nodep{n}{\conf \I {P}}{\mob}{k}$.
%Hence the thesis immediately follows since transition rule \rulename{Move} is enabled.

%The case of $M\redtau M'$ derived by applying \rulename{stop} is analogous.

\item Let $M\redtau M'$  by an application of \rulename{loccom}:
\[
\Txiombis{\rng c =-1}
{
{\nodep{n}{\conf \I {\timeout{\OUT c v .P}R\newpar \timeout{\LIN c x .Q}S }}{\mu}{h} \redtau \nodep{n}{\conf \I P\newpar Q \subst v x }{\mu}{h}}
}
\]
 with 
$M=\nodep{n}{\conf \I {\timeout{\OUT c v .P}R\newpar \timeout{\LIN c x .Q}S }}{\mu}{h}$ 
and 
$M'=\nodep{n}{\conf \I P\newpar Q \subst v x }{\mu}{h}$.
Therefore the following derivation is enabled for $ \rng c=-1$
\[
\Txiombis
{
\Txiombis
{\timeout{\OUT c v .P}R \trans{\out c v} P   \q \q \q \timeout{\LIN c v .Q}S \trans{\inp c v} Q}
{\timeout{\OUT c v .P}R\newpar \timeout{\LIN c x .Q}S  \trans \tau  P\newpar Q \subst v x }
}
{
{\nodep{n}{\conf \I {\timeout{\OUT c v .P}R\newpar \timeout{\LIN c x .Q}S }}{\mu}{h} \trans \tau \nodep{n}{\conf \I P\newpar Q \subst v x }{\mu}{h}}
}
\]
and $M\trans{\tau}\equiv M'$ is derived as required.

\item Let $M\redtau M'$   by an application of \rulename{glbcom}:
\[
\Txiombis
{\dist h k \leq \rng c}
{\nodep {n}{\conf \I {\timeout{\OUT c v .P}{R}}}{\mu}{h} | 
\nodep {m}{\conf \I {\timeout{\LIN c x Q}{S}}}{\mu'}{k}
\redtau
\nodep{n}{\conf \I P}{\mu}{h} | \nodep{m}{\conf \I Q{\subst v x}}{\mu'}{k}}
\]

with 
$M=\nodep{n}{\conf \I {\timeout{\OUT c v .P}  R}}{\mu}{h} | \nodep{m}{\conf \I {\timeout{\LIN c x .Q}S}}{\mu'}{k}$ 
and 
$M'=\nodep{n}{\conf \I P}{\mu}{h} | \nodep{m}{\conf \I {Q \subst v x }}{\mu'}{k}$.
Therefore the following derivation can be built up for $\dist h k \leq \rng c$
\[
\Txiombis
{
\Txiombis
{\timeout{\OUT c v .P}R \trans{\out c v} P }
{ \nodep{n}{\conf \I {\timeout{\OUT c v .P} R}}{\mu}{h}\trans{\send{c}{v}{h}} \nodep{n}{\conf \I {P}}{\mu}{h}} 
\Q\Q
\Txiombis
{\timeout{\LIN c v .Q}S \trans{\inp c v} Q}
{ \nodep{m}{\conf \I {\timeout{\LIN c v .Q}S}}{\mu'}{k}\trans{\rec{c}{v}{k}}\nodep{m}{\conf \I Q\subst v x }{\mu'}{k}}
}
{
\nodep{n}{\conf \I {\timeout{\OUT c v .P}  R}}{\mu}{h} | \nodep{m}{\conf \I {\timeout{\LIN c x .Q}S}}{\mu'}{k} \trans{\tau}\nodep{n}{\conf \I P }{\mu}{h} | \nodep{m}{\conf \I {Q \subst v x }}{\mu'}{k}
}
\]
and we get $M\trans{\tau} \equiv M'$.

\item Let $M\redtau M'$   by an application of rule \rulename{res}:
\[
\Txiombis{M_1 \redtau M_1'}{\res{\tilde g}  M_1 \redtau \res{\tilde g} M_1'}
\]
with $M=\res{\tilde g} M_1$ and $M'=\res{\tilde g} M'_1 $. 
By inductive hypothesis we have $M_1\trans{\tau}\equiv M_1'$.
Hence, by applying transition rules \rulename{Res}, we can derive
$M \trans{\tau} \equiv M'$.  

\item Let $M\redtau M'$  by an application of rule \rulename{struct}:
\[
\Txiombis{M \equiv N \q N\redtau N' \q N'\equiv M'}{M \redtau M'}
\]
By inductive hypothesis we have $N\trans{\tau}\equiv N'$. 
Since moreover, $M\equiv N$ and $M'\equiv N'$,  by an application of Lem.~\ref{struc_trans} 
we obtain  $M\trans{\tau}\equiv M'$.

\item Let $M\redtau M'$ by  an application of rule \rulename{parn}:
\[
\Txiombis{M_1 \redtau M_1'}{M_1 | N \redtau M_1' | N}
\]
with $M=M_1 | N$ and  $M'=M_1' | N$.
By inductive hypothesis, $M_1\redtau M_1'$ implies that $M_1\trans{\tau}\equiv M_1'$.
Hence, the thesis follows by applying transition rule \rulename{ParN}.
\end{itemize}
Let us prove the sub-result (\ref{a_red}). The proof is  by rule induction on
why  $M\trans{a}M'$.
\begin{itemize}
\item Let $M\trans{a}M'$  by an application of rule \rulename{ActChg}:
\[
\Txiombis{\I(a)=w\neq  v \Q P \trans{\wact v {a} } P' \Q{\I'}:=\I[a \mapsto v]}
{\nodep n {\conf \I P}{\mu}{h}  \trans{a} \nodep n {\conf {\I'} {P'}}{\mu}{h}}
\]
with $M=\nodep n {\conf \I P}{\mu}{h}$ and $M'=\nodep n {\conf {\I'} {P'}}{\mu}{h}$.
By Lem.~\ref{struc_p}(\ref{struc_p_write})
there exist $P_1,Q$ such that  $P\equiv \wact v {a} .P_1\newpar Q$ and $P'\equiv P_1 \newpar Q$.
Then we can apply reduction rules \rulename{actchg} and \rulename{parp} 
to infer $M\red_{a}M'$.

\item The cases when $M\trans{a}M'$ is derived by an application of either 
rule \rulename{ParN} or  rule \rulename{Res} are analogous to the corresponding cases when $M\trans{\tau}M'$.
\end{itemize}
Let us prove the sub-result 
\noindent (\ref{red_a}). The proof is  by rule induction on why
 $M\red_{a} M'$.
\begin{itemize}
\item Let $M\red_{a} M'$  by an  application of rule \rulename{actchg}:
\[
\Txiombis{\I(a)=w\neq  v \Q{\I'}:=\I[a \mapsto v]}
{\nodep n {\conf \I \wact v {a}.P}{\mu}{h}  \red_{a} \nodep n {\conf {{\I'}} {P}}{\mu}{h}}
\]
with $M=\nodep n {\conf \I \wact v {a}.P}{\mu}{h}$ and $M'=\nodep n {\conf {{\I'}} {P}}{\mu}{h}$.
By an application of rule 
\rulename{Actuator} we derive 
$\wact v a .P \trans{\wact v  a} P$. The thesis follows by  an 
application of rule \rulename{ActChg}.

\item The cases when $M\red_{a} M'$  is derived by an application of one
of the rules  among  \rulename{parp}, \rulename{parn}, \rulename{res}  or \rulename{struct} are analogous to the corresponding cases written for $M\redtau M'$.
\end{itemize}
Let us prove the sub-result  (\ref{sigma_red}). The proof is by rule induction
on why  $M\trans{\sigma}M'$.
\begin{itemize}
\item Let $M \trans{\sigma}M'$ by an application of rule \rulename{TimeZero}. This case  is immediate.

\item Let $M\trans{\sigma}M'$ by an application of rule \rulename{TimeStat}:
\[
\Txiombis{P \trans{\sigma} P'  \Q \nodep n {\conf \I P}{\stat}{h}  \ntrans{\tau} }
{\nodep n {\conf \I P}{\stat}{h}  \trans{\sigma} \nodep n {\conf \I {P'}}{\stat}{h}}
\]
with $M=\nodep n {\conf \I P}{\stat}{h}$ and $M'=\nodep n {\conf \I P'}{\stat}{h}$.
Since $P\trans{\sigma}P'$,  by Lem.~\ref{struc_p}(\ref{struc_p_time}) we derive 
$P\equiv  \prod_{ i\in I} \timeout{\pi_i.P_i}Q_i\newpar \prod_{ j\in J} \sigma.P_j$
and
$P'\equiv \prod_{ i\in I } Q_i\newpar \prod_{ j\in J} P_j$ for some $I,J,\pi_i,P_i,Q_i,P_j$. %%, for all $i\in I$ and $j\in J$.
By an application of the sub-result %%(\ref{tau_red}) and 
 (\ref{red_tau}) above, from $\nodep n {\conf \I P}{\stat}{h}  \ntrans{\tau}$  we derive $\nodep n {\conf \I P}{\stat}{h}  \not\redtau$.
Then the thesis follows by applying the reduction rule \rulename{timestat}.
 
\item  Let $M\trans{\sigma}M'$ by an application of rule \rulename{TimeStat}. 
This case is similar to the previous one by applying the  reduction rule \rulename{timemob} in place of \rulename{timestat}. 

\item Let $M\trans{\sigma} M'$ by an application of rule \rulename{TimePar}:
\[
\Txiombis{M_1 \trans{\sigma} M_1' \Q M_2 \trans{\sigma} M_2' \Q M_1 | M_2 \ntrans{\tau} }{M_1 | M_2 \trans{\sigma} M_1' | M_2'}
\]
with $M=M_1 | M_2$ and $M'=M_1' | M_2'$.
By inductive hypothesis we  $M_1\redtime M_1'$ and  $M_2\trans{\sigma}M_2'$.
Moreover, By an application of the sub-result %%(\ref{tau_red}) and 
 (\ref{red_tau}) above $M_1 | M_2 \ntrans{\tau}$ implies $M_1 | M_2 \not\!\redtau$.
Therefore we can apply the reduction rule \rulename{timepar} to 
get $M \redtime M'$. 

\item Let $M\trans{\sigma}M'$ by an application of rule \rulename{Res}:
\[
\Txiombis{M_1 \trans{\sigma} M_1'}{\res   {\tilde g}  M_1 \trans{\sigma} \res   {\tilde g} M_1'}
\]
with $M=\res   {\tilde g} M_1$ and $M'=\res   {\tilde g} M'_1 $. 
By inductive hypothesis, $M_1\trans{\sigma}M_1'$ implies that $M_1\redtime M_1'$.
Therefore, by applying the reduction rule \rulename{res}  we get our result.  
\end{itemize}
Let us prove the sub-result (\ref{red_sigma}). The proof is  by 
rule induction on why $M\redtime M'$.
\begin{itemize}
\item Let $M \redtime M'$ by an application of the reduction rule \rulename{timezero}. This case is immediate.

\item Let $M\redtime M'$ by an application of rule \rulename{timestat}:
\[
\Txiombis
{\nodep{n}{\conf {\I}{\prod_{ i\in I } \timeout{\pi_i.P_i}Q_i\newpar \prod_{j\in J} \sigma.P_j}}{\stat}{h} \not\redtau}
{\nodep{n}{\conf {\I}{\prod_{ i\in I } \timeout{\pi_i.P_i}Q_i\newpar \prod_{ j\in J  }\sigma.P_j}}{\stat}{h} 
\redtime 
\nodep{n}{\conf \I{ \prod_{ i\in I } Q_i\newpar  \prod_{ j\in J } P_j}}{\stat}{h}}
\]
with 
$M=\nodep{n}{\conf {\I}{ \prod_{ i\in I } \timeout{\pi_i.P_i}Q_i\newpar   \prod_{ j\in J } \sigma.P_j}}{\stat}{h}$ 
and 
$M'=\nodep{n}{\conf \I{ \prod_{ i\in I } Q_i\newpar  \prod_{ j\in J } P_j}}{\stat}{h}$.

By rule \rulename{Timeout} we derive $\timeout{\pi.P_i}{Q_i} \trans{\sigma} P_i$ and by rule \rulename{Delay}  we derive $\sigma.P_j \trans{\sigma} P_j$.
Now, we can repeatedly apply rule \rulename{TimeParP} to derive 
$\prod_{ i\in I } \timeout{\pi_i.P_i}Q_i\newpar \prod_{ j\in J  }\sigma.P_j  
\trans{\sigma} 
  \prod_{ i\in I } Q_i\newpar  \prod_{ j\in J } P_j$. Indeed, by contradiction, if \rulename{TimeParP}  would not be enabled, then rule \rulename{Com} would be enabled, and  by applying rule \rulename{ParP},  there would exist  $R$ such that 
$\prod_{ i\in I } \timeout{\pi_i.P_i}Q_i\newpar \prod_{ j\in J  }\sigma.P_j  
\trans{\tau} R$. Then, by applying rule \rulename{LocCom} and  the sub-result  (\ref{tau_red}) above, we would contradict the hypothesis 
 $\nodep{n}{\conf {\I}{\prod_{ i\in I } \timeout{\pi_i.P_i}Q_i\newpar \prod_{j\in J} \sigma.P_j}}{\stat}{h} \not\redtau$.
Therefore, the thesis follows by applying the transition rule \rulename{TimeStat}.
% if $\mu=\stat$ or transition rule \rulename{TimeMob} if $\mu=\mob$.

\item Let $M\redtime M'$ by an application of rule \rulename{timemob}. 
This case is analogous to the previous one by applying the transition rule \rulename{TimeMob} in place of rule \rulename{TimeStat}.

\item Let $M\redtime M'$  by an application of rule \rulename{res}:
\[
\Txiombis{M_1 \redtime M_1'}{\res   {\tilde g}  M_1 \redtime \res   {\tilde g} M_1'}
\]
with $M=\res   {\tilde g} M_1$ and $M'=\res   {\tilde g} M'_1 $. 
By inductive hypothesis, $M_1 \redtime M_1'$ implies that $M_1\trans{\sigma}\equiv M_1'$.
Therefore, by applying the transition rule \rulename{Res}, we derive 
$M \trans{\sigma}\equiv M'$.  

\item Let $M\redtime M'$  by an application of rule  \rulename{struct}: 
\[
\Txiombis{M \equiv N \q N\redtime N' \q N'\equiv M'}{M \redtime M'}
\]
By inductive hypothesis, $N\redtime N'$ implies that $N\trans{\sigma}\equiv N'$.

Moreover, since $M\equiv N$ and $M'\equiv N'$, by an application of 
 Lem.~\ref{struc_trans} we can derive $M\trans{\sigma}\equiv M'$.

\item Let $M\redtime M'$ by an application of rule \rulename{timepar}:
\[
\Txiombis{M_1 \redtime M_1' \Q M_2 \redtime M_2' \Q M_1 | M_2 \not\!\redtau}{M_1 | M_2 \redtime M_1' | M_2'}
\]
with $M=M_1 | M_2$  and $M'=M_1' | M_2'$.
By inductive hypothesis, $M_1\redtime M_1'$ implies $M_1\trans{\sigma}_{\equiv} M_1'$ and $M_2\redtime M_2'$ implies $M_2\trans{\sigma}\equiv  M_2'$.
Finally, by  an an application of the sub-result (\ref{tau_red}) above % and (\ref{red_tau}),
 $M_1 | M_2 \not\redtau$ implies $M_1 | M_2 \ntrans{\tau}$.
Therefore we can derive $M \trans{\sigma} \equiv M'$ by an application of the 
 transition rule \rulename{TimePar}.
\end{itemize}
\hfill\qed

%%%%%%%%%%%%%%%%%%%%%%%%%%%%%%%%%%%%%%%%%%%%%%%%%%%%%%%%%%%%%%%%%%%%%%%%%

\subsection{Proofs of Sec.~\ref{full-abstraction}}

\paragraph*{\textbf{Proof of Thm.~\ref{thm:algebraic-laws}. }}
For each law we exhibit the proper bisimulation. 
It is easy to see that for the first four laws the left-hand-side system
evolves into the right-hand-side by performing a $tau$-actions. So, 
in order to prove those laws it is enough to show that the two terms
under considerations are bisimilar. 
%% In the following, given a binary relation $\rel$ over 
%%network we denote with   $Clos(\rel)$ the reflexive,  symmetric and  sensor substitution closure of  $\rel$. Formally, $Clos(\rel)$  is defined as follows:
%%\begin{itemize}
%%\item $\rel \subseteq Clos(\rel)$
%%\item $(M,M) \in Clos(\rel)$, for any $M$
%%\item if $(M,N) \in Clos(\rel)$ then $(N,M) \in Clos(\rel)$
%%\item if $(M,N) \in Clos(\rel)$ then $(M[s@h \mapsto v], N[s@h \mapsto v])  \in Clos(\rel)$ for all sensor $s$ and value $v$ in the domain of $s$.
%%\end{itemize}
Let us proceed case by case. 
\begin{enumerate}
\item 
Let use define the relation 
\[
\rel = \left \{\left(\nodep{n}{\conf \I {{\wact v a}.P} | R}{\mu}{h}, \nodep{n}{\conf \I {P | R}}{\mu}{h} \right)\,|\, \I(a)=v \mbox{ and } a \mbox{ does not occur in } R \right \} \, \cup \, Id 
\]
where $Id$ is the identity relation. 
We prove that the symmetric closure of $\rel$ is a bisimulation.
Let $M=\nodep{n}{\conf \I {{\wact v a}.P} | R}{\mu}{h}$ and $N = 
\nodep{n}{\conf \I {P | R}}{\mu}{h}$, with $(M,N)\in \rel$. 
\begin{itemize}
\item If $M \trans{\rsensa w {s@h}} \nodep{n}{\conf {\I[s \mapsto w]} {{\wact v a}.P} | R}{\mu}{h}=M'$ then there is $N'$ such that $N  \trans{\rsensa w {s@h}}\nodep{n}{\conf {\I[s \mapsto w]} {P | R}}{\mu}{h}=N' $, with $(M',N')\in \rel$.
\item If $M \trans{\rsensa w {s@k}} M$, $h\neq k$, then  $N  \trans{\rsensa w {s@k}} N $ and  $(M,N)\in \rel$.
\item If $M \trans{\wact v{a@h}} M$ then $N \trans{\wact v{a@h}} N$ and $(M,N) \in \rel$.
\item Let $M \trans{\wact w{b@k}} M$. As $M$ and $N$ have the same physical interface it follows that $N \trans{\wact w{b@k}} N$ and $(M,N) \in \rel$.
\item If $M= 
\nodep{n}{\conf \I {{\wact v a}.P} | R}{\mu}{h} \trans{\alpha}\nodep{n}{\conf {\I'} {{\wact v a}.P} | R'}{\mu}{h}=M'$ then there is $N'$ such that
$N =
\nodep{n}{\conf \I {P} | R}{\mu}{h} \trans{\alpha}\nodep{n}{\conf {\I'} {P} | R'}{\mu}{h}=N'$. 
Since $a$ does not occur in $R$, then $ {\I'}(a)=v$ and therefore $(M',N') \in 
\rel$.

\item If $M =
\nodep{n}{\conf \I {{\wact v a}.P} | R}{\mu}{h} \trans \tau   \nodep{n}{\conf \I {P} | R}{\mu}{h}=M'$  then 
$N \Trans{} \nodep{n}{\conf \I P | R}{\mu}{h} = N'$ and $(M', N') \in Id \subset \rel$. 
\end{itemize}

As $\rel$ is not symmetric, let us show how $M$ can simulate the transitions of $N$. 
\begin{itemize}
\item If $N =\nodep{n}{\conf \I {P} | R}{\mu}{h} \trans{\alpha}   \nodep{n}{\conf {\I'} {P'} | R'}{\mu}{h'}=N'$ then there is $M'$ such that 
$M = 
\nodep{n}{\conf \I {{\wact v a}.P} | R}{\mu}{h} \trans \tau  \trans \alpha   \nodep{n}{\conf {\I'} {P'} | R'}{\mu}{h'}=M'$, with $(M', N') \in Id \subseteq \rel$. 
\end{itemize}

\item 
Let us define the relation 
\[
\rel = \left \{  \left(   \nodep{n}{\conf \I {@(x).P | R}}{\mu}{h} , \nodep{n}{\conf \I P{\subst h x} | R}{\mu}{h} \right)  \right \} \, \cup \, Id
\]
where $Id$ is the identity relation. We show that the symmetric closure
of $\rel$ is a bisimulation. The proof is similar to that of case  \ref{law1} where 
$\nodep{n}{\conf \I @(x).P | R}{\mu}{h}\trans{\tau}\nodep{n}{\conf \I P\subst{h}{x} | R}{\mu}{h}$. 

\item 
Let us define the relation 
$$
\rel = \left \{ \big( \nodep{n}{\conf \I {\timeout{\OUT c v.P}S} | \timeout{\LIN c x. Q}T |  R}{\mu}{h} \: , \: \nodep{n}{\conf \I {P | Q{\subst v x}| R}}{\mu}{h}
\big) \right \} \, \cup \, Id $$
such that $c$  is not in $ R$ and $ \rng{c}=-1$. We show that the 
symmetric closure of $\rel$ is a bisimulation. Let $M=\nodep{n}{\conf \I {\timeout{\OUT c v.P}S} | \timeout{\LIN c x. Q}T |  R}{\mu}{h} $ and 
$N = \nodep{n}{\conf \I {P | Q{\subst v x}| R}}{\mu}{h}$, with $(M,N) \in
 \rel$. 

\begin{itemize}
\item Let $M \trans{\alpha} M'$, with $\alpha \in \{\rsensa v {s@h}, \rsensa v {s@k},\wact v{a@h}  \}$. These cases are easy and very similar to the 
corresponding cases of law \ref{law1}
%%, networks in $\rel$ have equal interfaces, and moreover $\rel$ is closed under sensor value substitution.
%%Thus the thesis holds for steps $\trans{\rsensa v {s@h}} $ and $ \trans{\wact v{a@h}}$.
%%Let us consider the following cases:

\item Let $M\trans{\alpha}\nodep{n}{\conf {\I'} {\timeout{\OUT c v.P}{S}} | \timeout{\LIN c x. Q}{T} |  R'}{\mu}{h} = M'$. As $M \trans{\tau}$ it follows that
$\alpha\neq \sigma$ and the node cannot change location.  Then 
%% As $c$ does not occur in $R$,  the action $\trans{\alpha}$ does not involve any communication at
%% $c$. 
%% Massimo: Non e' questo il punto. Il punto cruciale e' altrove. 
there is $N'$ such that 
$N= \nodep{n}{\conf \I {P | Q{\subst v x}| R}}{\mu}{h}\trans{\alpha}\nodep{n}{\conf {\I'} {P | Q{\subst v x}| R'}}{\mu}{h}= N'$, with $(M',N') \in \rel$. 

\item If $M = 
\nodep{n}{\conf \I {\timeout{\OUT c v.P}{S}} | \timeout{\LIN c x. Q}{T} |  R}{\mu}{h} \trans{\tau}\nodep{n}{\conf \I P | Q\subst{v}{x} |  R}{\mu}{h} =M'$ then 
there is $N'$ such that 
$N = 
\nodep{n}{\conf \I P | Q\subst{v}{x} |  R}{\mu}{h} \Trans{} 
\nodep{n}{\conf \I P | Q\subst{v}{x} |  R}{\mu}{h} = N'$
and $(M',N') \in Id \subseteq \rel$.
%\item The case where a step is derived from $\timeout{\LIN c x.Q}Q'$ is analogous.
\end{itemize}

As $c$ cannot occur in $R$, this process cannot interfere with the 
communication at $M$ along channel $c$. 
As $\rel$ is not symmetric, let us show how $M$ can simulate the transitions of $N$. 
\begin{itemize}
\item If $N = 
\nodep{n}{\conf \I {P | Q{\subst v x}| R}}{\mu}{h}\trans{\alpha}\nodep{n}{\conf {\I'} {P' | Q'| R'}}{\mu}{h'} = N'$ then there is $M'$ such that 
$M = \nodep{n}{\conf \I {\timeout{\OUT c v.P}{S}} | \timeout{\LIN c x. Q}{T} |  R}{\mu}{h} \trans{\tau}\trans{\alpha}\nodep{n}{\conf {\I'} {P' | Q' | R'}}{\mu}{h'} = M'$ 
and $(M', N') \in Id \subseteq \rel$.
\end{itemize}

\item The proof of Law~\ref{law4} is similar to that of Law~\ref{law3}.
\item 
Let us define the relation 
$$\rel = \left \{ \left(\nodep{n}{\conf \I P}{\mu}{h},\nodep{n}{\conf \I \nil}{\mu}{h}\right) \right \}$$
where $ P$ does not contains terms of the form $ \timeout{\pi.P_1}{P_2} $ or
$ {\wact v a}.P_1 $, for any $a$. We prove that the symmetric closure of $\rel$ is a bisimulation. Let $M =\nodep{n}{\conf \I P}{\mu}{h}$ and $N =\nodep{n}{\conf \I \nil}{\mu}{h}$, with $(M, N) \in \rel$. 
\begin{itemize}
\item Let $M \trans{\alpha} M'$, with $\alpha \in \{\rsensa v {s@h}, \rsensa v {s@k},\wact v{a@h}  \}$. These cases are easy and very similar to the 
corresponding cases of law \ref{law1}. 
\item Let $M \trans{b} M'$.  This case is not admissible as $P$ does 
not contain terms of the form  ${\wact v b}.P_1$. 
\item Let $M \trans{\alpha} M'$, with $\alpha \in \{ \sendobs{c}{v}{k} , 
\recobs{c}{v}{k} \}$ This case is not admissible as $P$ does 
not contain terms of the form  $\timeout{\pi.P_1}{P_2}$. 
\item 
If $M = \nodep{n}{\conf \I P}{\mu}{h}\trans{\tau}\nodep{n}{\conf \I P'}{\mu}{h}=M'$ 
then there is $N'$ such that $N = \nodep{n}{\conf \I \nil}{\mu}{h}
\Trans{} \nodep{n}{\conf \I \nil}{\mu}{h} = N'$, with $(M'N') \in \rel$. 
Notice that the physical interface $\I$ cannot change, via an application of 
rule \rulename{ActChg},  as $P$ cannot 
write on actuators.  

\item 
If $M = \nodep{n}{\conf \I P}{\mu}{h}\trans{\sigma}\nodep{n}{\conf \I P'}
{\mu}{h'}=M'$ 
there there is $N'$ such that $N = \nodep{n}{\conf \I \nil}{\mu}{h}
\Trans{} \nodep{n}{\conf \I \nil}{\mu}{h'} = N'$, with $(M'N') \in \rel$.
\end{itemize}
The case when $N$ moves first is easier as $N$ can only perform $\sigma$-actions. 
\item 

Let us consider the relation 
\[
\rel=\left\{\left(\nodep{n}{\conf \I \nil}{\mu}{h},\zero\right)\mid \I(a) \text{ is undefined for any actuator } a \right\} \enspace . 
\]
We prove that the symmetric closure of $\rel$ is a bisimulation. 
Let $M =\nodep{n}{\conf \I \nil}{\mu}{h}$ and $N=\zero $, with $(M,N) \in \rel$. 
\begin{itemize}
\item 
If 
$
M= \nodep{n}{\conf \I \nil}{\mu}{h}\trans{\rsensa v {s@h}}\nodep{n}{\conf {\I[s \mapsto v]} \nil}{\mu}{h}=M'$
then there is $N'$ such that 
 $N=\zero\trans{\rsensa v {s@h}}\zero=N'$, with $(M', N') \in \rel$. 

\item 
If 
$
M= \nodep{n}{\conf \I \nil}{\mu}{h}\trans{\rsensa v {s@k}}\nodep{n}{\conf {\I} \nil}{\mu}{h}=M'$, $h \neq k$, 
then there is $N'$ such that 
 $N=\zero\trans{\rsensa v {s@h}}\zero=N'$, with $(M', N') \in \rel$. 

\item Let $M \trans{\wact w {b@h}} M'$.  This case is not admissible as $\I$ is 
undefined for any actuator $b$. 

\item If $M = 
\nodep{n}{\conf \I \nil}{\mu}{h}\trans{\sigma}\nodep{n}{\conf \I \nil}{\mu}{h'}
= M'$ then there is $N'$ such that $N = \zero \trans{\sigma} \zero = N'$, 
with $(M', N') \in \rel$. 

\end{itemize}

The case when $N$ moves first is easier as $N$ can only perform $\sigma$-actions.

\item 
Let us define the relation 
$$ \rel =  \left \{ \left(\nodep{n}{\conf \emptyset P}{\mob}{h},\nodep{m}{\conf \emptyset P}{\stat}{k}\right) \right \}$$
such that $ P$  does not contain terms $ @(x).Q $ and for any channel $ c $  in $ P$ either $ \rng c =  \infty $ or $ \rng c = -1$. We prove that the symmetric closure of $\rel$ is a bisimulation. Let $M = \nodep{n}{\conf \emptyset P}{\mob}{h}$ and $N =\nodep{m}{\conf \emptyset P}{\stat}{k}$, with $(M, N) \in \rel$. 

\begin{itemize}
 \item Let $M \trans{\alpha} M'$, with $\alpha \in \{\rsensa v {s@h}, \rsensa v {s@k},\wact v{a@h}  \}$. These cases are trivial or not admissible 
as the physical interfaces of both nodes is empty.

\item Let $M \trans{b} M'$. This case is not admissible because we deal 
with well-formed networks and an actuator must be defined in its 
physical interface before being used. 

\item If $M = \nodep{n}{\conf \emptyset P}{\mob}{h}\trans{\sigma}
\nodep{n}{\conf \emptyset P'}{\mob}{h'}= M'$ then there is $N'$ 
such that
$N = \nodep{m}{\conf \emptyset P}{\stat}{k}\trans{\sigma}\nodep{m}{\conf \emptyset P'}{\stat}{k} = N'$, with $(M', N') \in \rel$.

\item If $M = \nodep{n}{\conf \emptyset P}{\mob}{h}\trans{\tau}
\nodep{n}{\conf \emptyset P'}{\mob}{h}= M'$ then there is $N'$ 
such that
$N = \nodep{m}{\conf \emptyset P}{\stat}{k}\trans{\tau}\nodep{m}{\conf \emptyset P'}{\stat}{k} = N'$, with $(M', N') \in \rel$. Note that as 
 $ P$  does not contain terms of the form $ @(x).Q $ we can be 
sure that in node $n$ there will be the same process contained in  $m$,
i.e.\ $P'$.

\item If $M = \nodep{n}{\conf \emptyset P}{\mob}{h}\trans{\alpha}
\nodep{n}{\conf \emptyset P'}{\mob}{h}= M'$, $\alpha \in 
\{ \sendobs{c}{v}{h}, \recobs{c}{v}{h} \}$ with $\rng c = \infty$,  then there is $N'$ 
such that
$N = \nodep{m}{\conf \emptyset P}{\stat}{k}\trans{\alpha}\nodep{m}{\conf \emptyset P'}{\stat}{k} = N'$, with $(M', N') \in \rel$.
\end{itemize}

The cases when $N$ moves first are similar. 

\end{enumerate}
\hfill\qed

%%%%%%%%%%%%%%%%%%%%%%%%%%%%%%%%%%%%%%%%%%%%%%%%%%%%%

The next goal is the proof of Prop.~\ref{prop:SYS-bis}. For that 
we need a technical lemma:
\begin{lemma}
\label{lem:adding}
If $\res { \tilde{c} }  \left( \nodep{n}{\conf \I  {P_1}  }{\mu}{h}| O_1 \right)\approx 
\res { \tilde{d} } \left(\nodep{n}{\conf \I  {P_2}  }{\mu}{k}|O_2 \right)$ then
$\res { \tilde{c} }\left( \nodep{n}{\conf \I  {P_1 | R}  }{\mu}{h}|O_1 \right) \approx 
\res { \tilde{d} } \left(  \nodep{n}{\conf \I  {P_2 | R}  }{\mu}{k}|O_2\right)$
for any process $R$ which can only read sensors, transmit along some fresh Internet channel, and let time passes. 
\end{lemma}
\begin{proof}
Let use define below the relation $\rel$:
\begin{center}
{\small $
\left \{ \big(\res { \tilde{c} } (\nodep{n}{\conf \I  {P_1 {|} R}  }{\mu}{h}|O_1 ) ,  
\res { \tilde{d} } ( \nodep{n}{\conf \I  {P_2 {|} R}  }{\mu}{k}|O_2 )\big) :  \res { \tilde{c} } (\nodep{n}{\conf \I  {P_1}  }{\mu}{h}|O_1)\approx 
\res { \tilde{d} } (\nodep{n}{\conf \I  {P_2}  }{\mu}{k}|O_2)  \right \} 
$
}
\end{center}
where $R$ respects the hypotheses above. That means that $R$ can only
(i)  read 
the sensors  of $\I$; (ii) transmit along some fresh Internet channel; 
(iii) let time passes. 
We prove that the symmetric closure of $\rel$ is a bisimulation.
Let $(M,N) \in \rel$, we proceed by case analysis  on why 
$M \trans{\alpha} M'$. 

\begin{itemize}
\item 
Let $M = \res {\tilde{c}} ( \nodep{n}{\conf \I  {P_1 | R}  }{\mu}{h}|O_1) 
\trans{\alpha}\res { \tilde{c}} ( \nodep{n}{\conf {\I'}  {P'_1 | R} }{\mu}{h'}
|O'_1 ) = M'$, with $\alpha \neq \sigma$, be a transitions which does not involve 
(and affect) $R$ at all. This means that 
$\res {\tilde{c}} ( \nodep{n}{\conf \I  {P_1}  }{\mu}{h}|O_1) 
\trans{\alpha}\res { \tilde{c}} ( \nodep{n}{\conf {\I'}  {P'_1} }{\mu}{h'}
|O'_1 )$. 
%%As $ \res { \tilde{c} } (\nodep{n}{\conf \I  {P_1}  }{\mu}{h}|O_1)$$\approx$ 
%%$\res { \tilde{d} } (\nodep{n}{\conf \I  {P_2}  }{\mu}{h}|O_2)$
By hypothesis there are  $\I''$, $P'_2$, $O'_2$ and $k'$ such that 
$\res {\tilde{c}} ( \nodep{n}{\conf \I  {P_2}  }{\mu}{k}|O_2) 
\Trans{\alpha}\res { \tilde{c}} ( \nodep{n}{\conf {\I''}  {P'_2} }{\mu}{k'}
|O'_2 )$ and  
 \[ \res { \tilde{c} } (\nodep{n}{\conf {\I'}  {P'_1}  }{\mu}{h'}|O'_1)\approx 
\res { \tilde{d} } (\nodep{n}{\conf {\I''}  {P'_2}  }{\mu}{k'}|O'_2) 
\enspace . 
\]
By Prop.~\ref{prop:strong-obs} it follows that $\I' = \I''$. Furthermore as $\alpha \neq 
\sigma$ we have $h=h'$ and $k=k'$. 
Thus, $N = \res {\tilde{c}} ( \nodep{n}{\conf \I  {P_2 | R}  }{\mu}{k}|O_2) 
\Trans{\hat{\alpha}}\res { \tilde{c}} ( \nodep{n}{\conf {\I'}  {P'_2 |R} }{\mu}{k'}
|O'_2 )= N'$, with $(M' , N') \in \rel$. 

\item 
Let $M = \res {\tilde{c}} ( \nodep{n}{\conf \I  {P_1 | R}  }{\mu}{h}|O_1) 
\trans{\sigma}\res { \tilde{c}} ( \nodep{n}{\conf {\I}  {P'_1 | R'} }{\mu}{h'}
|O'_1 ) = M'$. We know that timed actions do not change the physical interface $\I$.  This implies that: 
\begin{enumerate}
\item $R \trans{\sigma} R'$
\item  $\res {\tilde{c}} ( \nodep{n}{\conf \I  {P_1 | R}  }{\mu}{h}|O_1) 
\ntrans{\tau}$
\item $\res {\tilde{c}} ( \nodep{n}{\conf \I  {P_1}  }{\mu}{h}|O_1) 
\trans{\sigma}\res { \tilde{c}} ( \nodep{n}{\conf {\I}  {P'_1} }{\mu}{h'}
|O'_1 )$. 
\end{enumerate}
In particular, the second item means that $R$ does not have any interaction with the network. It even does not read some sensor of $\I$.
By hypothesis we have that $\res {\tilde{c}} ( \nodep{n}{\conf \I  {P_2}  }{\mu}{k}|O_2) 
\Trans{}\trans{\sigma}\Trans{}\res { \tilde{c}} ( \nodep{n}{\conf {\I'}  {P'_2} }{\mu}{k'}
|O'_2 )$ with 
\[ \res { \tilde{c} } (\nodep{n}{\conf {\I}  {P'_1}  }{\mu}{h'}|O'_1)\approx 
\res { \tilde{d} } (\nodep{n}{\conf {\I'}  {P'_2}  }{\mu}{k'}|O'_2) 
\enspace . 
\]
By Prop.~\ref{prop:strong-obs} we know that it must be  $\I = \I'$. As $R$ 
cannot have any interaction with the rest of the network apart from 
time synchronisation it follows that $N = \res {\tilde{c}} ( \nodep{n}{\conf \I  {P_2 | R}  }{\mu}{k}|O_2) 
\Trans{} \trans{\sigma}\Trans{}
\res { \tilde{c}} ( \nodep{n}{\conf {\I}  {P'_2 |R'} }{\mu}{k'}
|O'_2 )= N'$, with $(M' , N') \in \rel$. 

\item 
Let $M = \res {\tilde{c}} ( \nodep{n}{\conf \I  {P_1 | R}  }{\mu}{h}|O_1) 
\trans{\alpha}\res { \tilde{c}} ( \nodep{n}{\conf {\I}  {P_1 | R'} }{\mu}{h}
|O'_1 ) = M'$, with $\alpha \neq \sigma$, be a transitions which is due to $R$. 
This can be a sensor reading or a transmission along some channel $b$, with $\rng b = \infty$. In that case, it is easy to see that, as $\rng b = \infty$,  then 
$N = \res {\tilde{c}} ( \nodep{n}{\conf \I  {P_2 | R}  }{\mu}{k}|O_2) 
\trans{\alpha}\res { \tilde{c}} ( \nodep{n}{\conf {\I}  {P_2 | R'} }{\mu}{k}
|O_2 ) = N'$, with $(M', N') \in \rel$.
\end{itemize}

\end{proof}

%%%%%%%%%%%%%%%%%%%%%%%%%%%%%%%%%%%%%%%%%%%%%%%%%%%%%%%%%%%%%%%%%%%%%%%%%%%%%%

\paragraph{\textbf{Proof of Prop.~\ref{prop:SYS-bis}.}}
First of all,  we notice that we can focus on smaller systems. This is because: 
\[
\begin{array}{rcl}
\res {\tilde{c}} Sys & = & \res {\tilde{c}} \big(
 Phone | LightMng1 | LightMng2 | BoilerMng \big)\\[3pt]
 & \equiv &
 \res {\tilde{c}} \big(
 Phone | LightMng1 | LightMng2 \big)  | BoilerMng 
\end{array}
\]
and 
\[
\begin{array}{rcl}
\res{\tilde{c},g}%\res{g}
 \overline{Sys}  & = &
\res  { \tilde{c},g} %\res{g} 
 \big( \overline{Phone}  | LightMng1 | LightMng2  |  \overline{{C}LightMng} | BoilerMng \big)\\[3pt]
 & \equiv &
 \res  { \tilde{c},g} %\res{g} 
 \big( \overline{Phone}  | LightMng1 | LightMng2  |  \overline{{C}LightMng} \big) | BoilerMng . 
\end{array}
\]
By Thm.~\ref{thm:congruence} the relation $\approx$ is preserved by parallel composition. Thus, in order to prove our result it is enough to show that: 
 \[
\begin{array}{c}
\res {\tilde{c}} \big(  Phone | LightMng1 | LightMng2  \big)  \approx \\[3pt]
\res  { \tilde{c}}\res{g} \big( \overline{Phone}  | LightMng1 | LightMng2 |
 \overline{{C}LightMng}  \big) \enspace .
\end{array} \] 
Actually we  can consider even smaller systems. As  $Phone = \nodep{n_P}{\conf {\I_P}{BoilerCtrl \newpar LightCtrl}}{\mob}{out}$ and $\overline{Phone} = \nodep{n_P}{\conf {\I_P}{BoilerCtrl \newpar 
\overline{LightCtrl}}}{\mob}{out}$  by Lem.~\ref{lem:adding} 
it is enough to show that 
\[
\begin{array}{c}
\res {\tilde{c}} \big( \nodep{n_P}{\conf {\I_P}{LightCtrl}}{\mob}{out}  | LightMng1 | LightMng2  \big)  \; \approx \; \\[3pt]
\res  { \tilde{c}}\res{g} \big(  \nodep{n_P}{\conf {\I_P}{\overline{LightCtrl}}}{\mob}{out} | LightMng1 | LightMng2  |  \overline{{C}LightMng}  \big) \enspace . \end{array}
\]
Let us call $S_{\mathrm{L}}$ the system on the left side, and $S_{\mathrm{R}}$ the 
system on the right side. 
Let us define the relation
\[
\rel \deff \bigcup_{i=1}^{17}
 \Big( \res {\tilde{c}} {M_i} \, , \, \res{\tilde{c}}\res g {N_i} \Big)
\]
 where the pairs $(M_i, N_i)$ are enumerated below. We will prove that the symmetric closure of $\rel$ is a bisimulation up to expansion~\cite{Sangiorgi-book}. Then we will show
that $S_{\mathrm{L}} = \res{\tilde{c}}M_1 \RRr \res{\tilde{c}}\res{g} N_1  \isexpan S_{\mathrm{R}}$. As
the up to expansion technique is sound, and the bisimilarity is a transitive relation,  it
follows that $S_{\mathrm{L}} \approx S_{\mathrm{R}}$. 

Let us provide the list  of pairs $(M_i, N_i)$, for $1 \leq i \leq 17$: 
{\small
\begin{itemize}
\item $M_1 = \nodep{n_P}{\conf {\I_P}{LightCtrl}}{\mob}{k}  \; \big| 
\;  \nodep{n_{1}}{\conf {\I_{1}}{L_1}}{\stat}{loc1} \; \big| \; \nodep{n_{2}}{\conf {\I_{2}}{L_2}}{\stat}{loc4} $  \\[2pt]
$N_1 = \nodep{n_P}{\conf {\I_P}{\sigma.\overline{LightCtrl}}}{\mob}{k} 
 \; \big| 
\;  \nodep{n_{1}}{\conf {\I_{1}}{L_1}}{\stat}{loc1} \; \big| \; \nodep{n_{2}}{\conf {\I_{2}}{L_2}}{\stat}{loc4} \; \big| \; 
 \nodep{n_{LM}}{\conf {\emptyset}{\sigma.\overline{CLM}}}{\stat}{loc3}$,\\[2pt]
with $k  \not\in \{ loc1, loc2, loc3, loc4 \}$, $\I_1(light_1)=\mathsf{off}$
and $\I_2(light_2)=\mathsf{off}$
\item $M_2= \nodep{n_P}{\conf {\I_P}{\sigma.LightCtrl}}{\mob}{loc1}   \; \big| \;
\nodep{n_{1}}{\conf {\I_{1}}{\wact{\mathsf{on}}{light_1}.\sigma.L_1}}{\stat}{loc1} \; \big| \; \nodep{n_{2}}{\conf {\I_{2}}{L_2}}{\stat}{loc4} $  \\[2pt]
$N_2 = \nodep{n_P}{\conf {\I_P}{\sigma.\overline{LightCtrl}}}{\mob}{loc1} \q \big| \q 
\nodep{n_1}{\conf {\I_1}{\wact{\mathsf{on}}{light_1}.\sigma.L_1}}{\stat}{loc1} \q \big| 
\q \nodep{n_2}{\conf {\I_2}{L_2}}{\stat}{loc4} \q \big| \q
\nodep{n_{LM}}{\conf {\emptyset}{\sigma.\overline{CLM}}}{\stat}{loc3}$
\item  $M_3= \nodep{n_P}{\conf {\I_P}{\sigma.LightCtrl}}{\mob}{loc1}   \; \big| \;
\nodep{n_{1}}{\conf {\I'_{1}}{\sigma.L_1}}{\stat}{loc1} \; \big| \; \nodep{n_{2}}{\conf {\I_{2}}{L_2}}{\stat}{loc4} $  \\[2pt]
$N_3 = \nodep{n_P}{\conf {\I_P}{\sigma.\overline{LightCtrl}}}{\mob}{loc1} \q \big| \q 
\nodep{n_1}{\conf {\I'_1}{\sigma.L_1}}{\stat}{loc1} \q \big| 
\q \nodep{n_2}{\conf {\I_2}{L_2}}{\stat}{loc4} \q \big| \q
\nodep{n_{LM}}{\conf {\emptyset}{\sigma.\overline{CLM}}}{\stat}{loc3}$, \\[2pt]
with $\I'_1(light_1) = \mathsf{on}$.

\item 
 $M_4 = \nodep{n_P}{\conf {\I_P}{LightCtrl}}{\mob}{k}  \; \big| 
\;  \nodep{n_{1}}{\conf {\I'_{1}}{L_1}}{\stat}{loc1} \; \big| \; \nodep{n_{2}}{\conf {\I_{2}}{L_2}}{\stat}{loc4} $  \\[2pt]
$N_4 = \nodep{n_P}{\conf {\I_P}{\sigma.\overline{LightCtrl}}}{\mob}{k} 
 \; \big| 
\;  \nodep{n_{1}}{\conf {\I'_{1}}{L_1}}{\stat}{loc1} \; \big| \; \nodep{n_{2}}{\conf {\I_{2}}{L_2}}{\stat}{loc4} \; \big| \; 
 \nodep{n_{LM}}{\conf {\emptyset}{\sigma.\overline{CLM}}}{\stat}{loc3}$,\\[2pt]
with $k  \not\in \{ loc1, loc2, loc3, loc4 \}$

\item 
$M_5 = \nodep{n_P}{\conf {\I_P}{LightCtrl}}{\mob}{k} 
  \; \big| \;
\nodep{n_{1}}{\conf {\I'_{1}}{\wact{\mathsf{off}}{light_1}.L_1}}{\stat}{loc1} \; \big| \; \nodep{n_{2}}{\conf {\I_{2}}{L_2}}{\stat}{loc4} $\\[2pt]
$N_5 = \nodep{n_P}{\conf {\I_P}{\sigma.\overline{LightCtrl}}}{\mob}{k} \q \big| \q 
\nodep{n_1}{\conf {\I'_1}{\wact{\mathsf{off}}{light_1}.L_1}}{\stat}{loc1} \q \big| 
\q \nodep{n_2}{\conf {\I_2}{L_2}}{\stat}{loc4} \q \big| \q
\nodep{n_{LM}}{\conf {\emptyset}{\sigma.\overline{CLM}}}{\stat}{loc3}$\\[2pt]
with $k  \not\in \{ loc1, loc2, loc3, loc4 \}$

\item 
$ M_6 = \nodep{n_P}{\conf {\I_P}{LightCtrl}}{\mob}{loc1} 
  \; \big| \;
\nodep{n_{1}}{\conf {\I'_{1}}{\wact{\mathsf{off}}{light_1}.L_1}}{\stat}{loc1} \; \big| \; \nodep{n_{2}}{\conf {\I_{2}}{L_2}}{\stat}{loc4}$\\[2pt]
\noindent
\hspace*{-2mm}$
\begin{array}{rcl}
N_6  & =  & \nodep{n_P}{\conf {\I_P}{\sigma.\overline{LightCtrl}}}{\mob}{loc1} \q \big| \q 
\nodep{n_1}{\conf {\I'_1}{\wact{\mathsf{off}}{light_1}.L_1}}{\stat}{loc1} \q \big| \\[2pt]
&& \nodep{n_2}{\conf {\I_2}{L_2}}{\stat}{loc4} \q \big| \q
\nodep{n_{LM}}{\conf {\emptyset}{\timeout{\OUT{c_1}{}.\sigma.\overline{CLM}}}{CLM}}{\stat}{loc3}
\end{array}
$

\item 
$M_7 = \nodep{n_P}{\conf {\I_P}{LightCtrl}}{\mob}{loc2}   \; \big| \;
\nodep{n_{1}}{\conf {\I'_{1}}{L_1}}{\stat}{loc1} \; \big| \; \nodep{n_{2}}{\conf {\I_{2}}{L_2}}{\stat}{loc4}$\\[2pt]
$N_7 = \nodep{n_P}{\conf {\I_P}{\sigma.\overline{LightCtrl}}}{\mob}{loc2} 
 \; \big| 
\;  \nodep{n_{1}}{\conf {\I'_{1}}{L_1}}{\stat}{loc1} \; \big| \; \nodep{n_{2}}{\conf {\I_{2}}{L_2}}{\stat}{loc4} \; \big| \; 
 \nodep{n_{LM}}{\conf {\emptyset}{\sigma.\overline{CLM}}}{\stat}{loc3}$

\item 
$M_8 = \nodep{n_P}{\conf {\I_P}{LightCtrl}}{\mob}{loc3} 
  \; \big| \;
\nodep{n_{1}}{\conf {\I'_{1}}{\wact{\mathsf{off}}{light_1}.L_1}}{\stat}{loc1} \; \big| \; \nodep{n_{2}}{\conf {\I_{2}}{L_2}}{\stat}{loc4}$\\[2pt] 
$N_8 = \nodep{n_P}{\conf {\I_P}{\sigma.\overline{LightCtrl}}}{\mob}{loc3} 
 \; \big| 
\;  \nodep{n_{1}}{\conf {\I'_{1}}{\wact{\mathsf{off}}{light_1}.L_1}}{\stat}{loc1} \; \big| \; \nodep{n_{2}}{\conf {\I_{2}}{L_2}}{\stat}{loc4} \; \big| \; 
 \nodep{n_{LM}}{\conf {\emptyset}{\sigma.\overline{CLM}}}{\stat}{loc3}$

\item 

$M_9 =  \nodep{n_P}{\conf {\I_P}{LightCtrl}}{\mob}{loc2} 
  \; \big| \;
\nodep{n_{1}}{\conf {\I'_{1}}{\wact{\mathsf{off}}{light_1}.L_1}}{\stat}{loc1} \; \big| \; \nodep{n_{2}}{\conf {\I_{2}}{L_2}}{\stat}{loc4}$\\[2pt]
$N_9 = \nodep{n_P}{\conf {\I_P}{\sigma.\overline{LightCtrl}}}{\mob}{loc2} 
 \; \big| 
\;  \nodep{n_{1}}{\conf {\I'_{1}}{\wact{\mathsf{off}}{light_1}.L_1}}{\stat}{loc1} \; \big| \; \nodep{n_{2}}{\conf {\I_{2}}{L_2}}{\stat}{loc4} \; \big| \; 
 \nodep{n_{LM}}{\conf {\emptyset}{\sigma.\overline{CLM}}}{\stat}{loc3}$

\item
$M_{10} =  \nodep{n_P}{\conf {\I_P}{LightCtrl}}{\mob}{loc3} 
  \; \big| \;
\nodep{n_{1}}{\conf {\I_{1}}{L_1}}{\stat}{loc1} \; \big| \; \nodep{n_{2}}{\conf {\I_{2}}{L_2}}{\stat}{loc4}$\\[2pt]
$N_{10} = \nodep{n_P}{\conf {\I_P}{\sigma.\overline{LightCtrl}}}{\mob}{loc3} 
 \; \big| 
\;  \nodep{n_{1}}{\conf {\I_{1}}{L_1}}{\stat}{loc1} \; \big| \; \nodep{n_{2}}{\conf {\I_{2}}{L_2}}{\stat}{loc4} \; \big| \; 
 \nodep{n_{LM}}{\conf {\emptyset}{\sigma.\overline{CLM}}}{\stat}{loc3}$

\item

$M_{11} =  \nodep{n_P}{\conf {\I_P}{LightCtrl}}{\mob}{loc2} 
  \; \big| \;
\nodep{n_{1}}{\conf {\I_{1}}{L_1}}{\stat}{loc1} \; \big| \; \nodep{n_{2}}{\conf {\I_{2}}{L_2}}{\stat}{loc4}$\\[2pt]
$N_{11} = \nodep{n_P}{\conf {\I_P}{\sigma.\overline{LightCtrl}}}{\mob}{loc2} 
 \; \big| 
\;  \nodep{n_{1}}{\conf {\I_{1}}{L_1}}{\stat}{loc1} \; \big| \; \nodep{n_{2}}{\conf {\I_{2}}{L_2}}{\stat}{loc4} \; \big| \; 
 \nodep{n_{LM}}{\conf {\emptyset}{\sigma.\overline{CLM}}}{\stat}{loc3}$

\item 
$M_{12}= \nodep{n_P}{\conf {\I_P}{\sigma.LightCtrl}}{\mob}{loc4} 
  \; \big| \;
\nodep{n_{1}}{\conf {\I_{1}}{L_1}}{\stat}{loc1} \; \big| \; 
\nodep{n_{2}}{\conf {\I_{2}}{\wact{\mathsf{on}}{light_2}.\sigma.L_2}}{\stat}{loc4} $\\[2pt]
$
N_{12} = \nodep{n_P}{\conf {\I_P}{\sigma.\overline{LightCtrl}}}{\mob}{loc4} 
 \; \big| 
\;  \nodep{n_{1}}{\conf {\I_{1}}{L_1}}{\stat}{loc1} \; \big| \; \nodep{n_{2}}{\conf {\I_{2}}{\wact{\mathsf{on}}{light_2}.\sigma.L_2}}{\stat}{loc4} \; \big| \; 
 \nodep{n_{LM}}{\conf {\emptyset}{\sigma.\overline{CLM}}}{\stat}{loc3}$

\item 
$M_{13}= \nodep{n_P}{\conf {\I_P}{\sigma.LightCtrl}}{\mob}{loc4} 
  \; \big| \;
\nodep{n_{1}}{\conf {\I_{1}}{L_1}}{\stat}{loc1} \; \big| \; 
\nodep{n_{2}}{\conf {\I'_{2}}{\sigma.L_2}}{\stat}{loc4} $\\[2pt]
$
N_{13} = \nodep{n_P}{\conf {\I_P}{\sigma.\overline{LightCtrl}}}{\mob}{loc4} 
 \; \big| 
\;  \nodep{n_{1}}{\conf {\I_{1}}{L_1}}{\stat}{loc1} \; \big| \; \nodep{n_{2}}{\conf {\I'_{2}}{\sigma.L_2}}{\stat}{loc4} \; \big| \; 
 \nodep{n_{LM}}{\conf {\emptyset}{\sigma.\overline{CLM}}}{\stat}{loc3}$\\[2pt]
where $\I'_2(light_2)=\mathsf{on}$

\item 

$M_{14} = \nodep{n_P}{\conf {\I_P}{LightCtrl}}{\mob}{loc3} 
  \; \big| \;
\nodep{n_{1}}{\conf {\I_{1}}{L_1}}{\stat}{loc1} \; \big| \; \nodep{n_{2}}{\conf {\I'_{2}}{L_2}}{\stat}{loc4}$\\[2pt]
$N_{14} = \nodep{n_P}{\conf {\I_P}{\sigma.\overline{LightCtrl}}}{\mob}{loc3} 
 \; \big| 
\;  \nodep{n_{1}}{\conf {\I_{1}}{L_1}}{\stat}{loc1} \; \big| \; \nodep{n_{2}}{\conf {\I'_{2}}{L_2}}{\stat}{loc4} \; \big| \; 
 \nodep{n_{LM}}{\conf {\emptyset}{\sigma.\overline{CLM}}}{\stat}{loc3}
$

\item

$M_{15} =\nodep{n_P}{\conf {\I_P}{LightCtrl}}{\mob}{loc4} 
  \; \big| \;
\nodep{n_{1}}{\conf {\I_{1}}{L_1}}{\stat}{loc1} \; \big| \; 
\nodep{n_{2}}{\conf {\I'_{2}}{\wact{\mathsf{off}}{light_2}.L_2}}{\stat}{loc4}$\\[2pt]
\noindent \hspace*{-3mm}
$
\begin{array}{rcl}
N_{15} &  = &  \nodep{n_P}{\conf {\I_P}{\sigma.\overline{LightCtrl}}}{\mob}{loc4} 
 \; \big| 
\;  \nodep{n_{1}}{\conf {\I_{1}}{L_1}}{\stat}{loc1} \; \big| \; \nodep{n_{2}}{\conf {\I'_{2}}{\wact{\mathsf{off}}{light_2}.L_2}}{\stat}{loc4} \; \big| \; \\[2pt]
&& \nodep{n_{LM}}{\conf {\emptyset}{\timeout{\OUT{c_2}{}.\sigma.\overline{CLM}}{\overline{CLM}}}}{\stat}{loc3}
\end{array}$

\item
$M_{16} = \nodep{n_P}{\conf {\I_P}{LightCtrl}}{\mob}{loc2} 
  \; \big| \;
\nodep{n_{1}}{\conf {\I_{1}}{L_1}}{\stat}{loc1} \; \big| \; \nodep{n_{2}}{\conf {\I'_{2}}{\wact{\mathsf{off}}{light_2}.L_2}}{\stat}{loc4}$\\[2pt]
$N_{16}  =  \nodep{n_P}{\conf {\I_P}{\sigma.\overline{LightCtrl}}}{\mob}{loc2} 
 \; \big| 
\;  \nodep{n_{1}}{\conf {\I_{1}}{L_1}}{\stat}{loc1} \; \big| \; \nodep{n_{2}}{\conf {\I'_{2}}{\wact{\mathsf{off}}{light_2}.L_2}}{\stat}{loc4} \; \big| \; 
 \nodep{n_{LM}}{\conf {\emptyset}{\sigma.\overline{CLM}}}{\stat}{loc3}$

\item 
$M_{17} = \nodep{n_P}{\conf {\I_P}{LightCtrl}}{\mob}{loc3} 
  \; \big| \;
\nodep{n_{1}}{\conf {\I_{1}}{L_1}}{\stat}{loc1} \; \big| \; \nodep{n_{2}}{\conf {\I'_{2}}{\wact{\mathsf{off}}{light_2}.L_2}}{\stat}{loc4}$\\[2pt]
$N_{17}  =  \nodep{n_P}{\conf {\I_P}{\sigma.\overline{LightCtrl}}}{\mob}{loc3} 
 \; \big| 
\;  \nodep{n_{1}}{\conf {\I_{1}}{L_1}}{\stat}{loc1} \; \big| \; \nodep{n_{2}}{\conf {\I'_{2}}{\wact{\mathsf{off}}{light_2}.L_2}}{\stat}{loc4} \; \big| \; 
 \nodep{n_{LM}}{\conf {\emptyset}{\sigma.\overline{CLM}}}{\stat}{loc3}$
\end{itemize}
}

We show that the symmetric closure of $\rel$ is a bisimulation up to expansion.
For each pair $\big( \res{{\tilde{c}}}M_i \, , \, \res{{\tilde{c}},g}N_i \big) \in \rel$ we proceed by case analysis on why $ \res{{\tilde{c}}}M_i 
\trans{\alpha} \hat{M}$. Then, we do the same for $ \res{{\tilde{c}},g}N_i
\trans{\alpha} \hat{N}$. Before starting the case analysis 
we notice that in all pairs 
of $\rel$ the physical interfaces of the corresponding nodes are
the same. For that reason  we can safely omit 
the extensional actions of the form   $\wact v {a@h}$. 
Moreover, our processes never read sensors (we removed from the initial system both $BoilerCtrl$
and $BoilerMng$), thus we can safely omit
actions of the form $\rsensa v {s@h}$ as well. 
\\

- Let us consider the pair $\big( \res{\tilde{c}} M_1 \, , \, \res{\tilde{c}}\res {g} N_1 \big)$.
\begin{itemize}
\item Let $\res {\tilde{c}} M_1 \trans{\alpha} \hat{M}$, for $\alpha \neq \sigma$. This case is not admissible as the phone is too far to interact with the 
some light manager. 
\item Let  $\res {\tilde{c}} M_1 \trans{\sigma} \res {\tilde{c}} M'_1$, with 
\[
M'_1 =  \nodep{n_P}{\conf {\I_P}{LightCtrl}}{\mob}{k}  | 
\nodep{n_{1}}{\conf {\I_{1}}{\wact{\mathsf{off}}{light_1}.L_1}}{\stat}{loc1} | 
 \nodep{n_{2}}{\conf {\I_{2}}{\wact{\mathsf{off}}{light_2}.L_2}}{\stat}{loc4}
\]
with $k \not \in \{ loc1, loc2, loc3, loc4 \}$, that is  the phone didn't
get inside the smarthome. By the sake of simplicity we will call $k$ all locations outside the smarthome. 
By two applications of Law~\ref{law1}  of Thm.~\ref{thm:algebraic-laws}
we have:
\[ 
\begin{array}{rcl}
\res {\tilde{c}}M'_1  & \gtrsim  & \res {\tilde{c}} \big( 
\nodep{n_P}{\conf {\I_P}{LightCtrl}}{\mob}{k}   \; \big| \;
\nodep{n_{1}}{\conf {\I_{1}}{L_1}}{\stat}{loc1} \; \big| \; \nodep{n_{2}}{\conf {\I_{2}}{L_2}}{\stat}{loc4}
\big)\\[3pt]
& = & \res {\tilde{c}} M_1 \enspace . 
\end{array}
\]
Then,  $\res {\tilde{c},g} N_1 
\trans{\sigma} \Trans{} \res {\tilde{c},g} N_1$, 
and $\big(  \res {\tilde{c}} M_1 \, , \,   \res {\tilde{c},g} N_1 \big) \in \rel$. 

\item Let  $\res {\tilde{c}} M_1 \trans{\sigma} \res {\tilde{c}} M'_1$ with 
\[ M'_1 = \nodep{n_P}{\conf {\I_P}{LightCtrl}}{\mob}{loc1}  | 
\nodep{n_{1}}{\conf {\I_{1}}{\wact{\mathsf{off}}{light_1}.L_1}}{\stat}{loc1} | 
 \nodep{n_{2}}{\conf {\I_{2}}{\wact{\mathsf{off}}{light_2}.L_2}}{\stat}{loc4}
\enspace . \]
In this case the smartphone just entered the smarthome from its entrance 
room (Room1) located at $loc1$. 
By two applications of Law~\ref{law1} and one application of 
Law~\ref{law4} of Thm.~\ref{thm:algebraic-laws}
we have:
\[ 
\begin{array}{rcl}
\res {\tilde{c}}M'_1  & \gtrsim  & \res {\tilde{c}} \big( 
\nodep{n_P}{\conf {\I_P}{\sigma.LightCtrl}}{\mob}{loc1}   \; \big| \;
\nodep{n_{1}}{\conf {\I_{1}}{\wact{\mathsf{on}}{light_1}.\sigma.L_1}}{\stat}{loc1} \; \big| \; \nodep{n_{2}}{\conf {\I_{2}}{L_2}}{\stat}{loc4}
\big)\\[3pt]
& = & \res {\tilde{c}} M_2 \enspace . 
\end{array}
\]
Then there is $N_2$ such that  $\res {\tilde{c},g} N_1 
\trans{\sigma} \Trans{} \res {\tilde{c},g} N_2$ with 
\[
\begin{array}{rcl}
 \res {\tilde{c},g} N_2 & = &\res {\tilde{c},g} 
\big( 
\nodep{n_P}{\conf {\I_P}{\sigma.\overline{LightCtrl}}}{\mob}{loc1} \q \big| \q 
\nodep{n_1}{\conf {\I_1}{\wact{\mathsf{on}}{light_1}.\sigma.L_1}}{\stat}{loc1} \q \big| \\[3pt]
&&\nodep{n_2}{\conf {\I_2}{L_2}}{\stat}{loc4} \q \big| \q
\nodep{n_{LM}}{\conf {\emptyset}{\sigma.\overline{CLM}}}{\stat}{loc3}
\big) 
\end{array}
\]
and $\big(  \res {\tilde{c}} M_2 \, , \,   \res {\tilde{c},g} N_2 \big) \in \rel$. 
\end{itemize}
Now, we proceed by case analysis on why $\res {\tilde{c},g}N_1 \trans{\alpha}
\hat{N}$. 
\begin{itemize}
\item Let $\res {\tilde{c},g}N_1 \trans{\alpha}
\hat{N}$, with $\alpha \neq \sigma$. This case is not admissible. 
\item  Let $\res {\tilde{c},g}N_1 \trans{\sigma} \res {\tilde{c},g}N'_1$, 
 where 
the phone didn't enter the house, as its location is different from $loc1$. 
This case is similar to the previous one. 
\item $\res {\tilde{c},g}N_1 \trans{\sigma} \res {\tilde{c},g}N'_1$, with 
\[
\begin{array}{rcl}
N'_1 & = & 
\nodep{n_P}{\conf {\I_P}{\overline{LightCtrl}}}{\mob}{loc1}  \q \big| \q 
\nodep{n_{1}}{\conf {\I_{1}}{\wact{\mathsf{off}}{light_1}.L_1}}{\stat}{loc1} 
\q \big| \\[3pt] 
&& \nodep{n_{2}}{\conf {\I_{2}}{\wact{\mathsf{off}}{light_2}.L_2}}{\stat}{loc4}
\q \big| \q  \nodep{n_{LM}}{\conf {\emptyset}{\overline{CLM}}}{\stat}{loc3}
\enspace .
\end{array}
\]
Because the phone just moved to location $loc1$.
By two applications of Law~\ref{law1}, one applications of Law~\ref{law2}, 
and two applications of Law~\ref{law4} of Thm.~\ref{thm:algebraic-laws}
we have:
\[
\begin{array}{rcl}
 \res {\tilde{c},g} N'_1 & \gtrsim &\res {\tilde{c},g} 
\big( 
\nodep{n_P}{\conf {\I_P}{\sigma.\overline{LightCtrl}}}{\mob}{loc1} \q \big| \q 
\nodep{n_1}{\conf {\I_1}{\wact{\mathsf{on}}{light_1}.\sigma.L_1}}{\stat}{loc1} \q \big| \\[3pt]
&&\nodep{n_2}{\conf {\I_2}{L_2}}{\stat}{loc4} \q \big| \q
\nodep{n_{LM}}{\conf {\emptyset}{\sigma.\overline{CLM}}}{\stat}{loc3}
\big) \\[3pt]
& = &  \res {\tilde{c},g} N_2 \enspace . 
\end{array}
\]
Then there is $M_2$ such that $\res {\tilde{c}}M_1 \trans{\sigma} 
\Trans{} \res{\tilde{c}} M_2$, 
with 
\[
M_2 = \nodep{n_P}{\conf {\I_P}{\sigma.LightCtrl}}{\mob}{loc1}   \; \big| \;
\nodep{n_{1}}{\conf {\I_{1}}{\wact{\mathsf{on}}{light_1}.\sigma.L_1}}{\stat}{loc1} \; \big| \; \nodep{n_{2}}{\conf {\I_{2}}{L_2}}{\stat}{loc4}
\]
and $\big( \res{\tilde{c}}M_2 \, , \,   \res {{\tilde{c}}, g}N_2 \big) \in \rel$. 
\end{itemize}

- Let us consider the pair $\big( \res {\tilde{c}} M_2 \, , \, \res{\tilde{c}}\res g N_2 \big)$. The only possible transition in both networks is a strong 
transition $\trans{light_1}$ which leads to the pair $\big( \res {\tilde{c}} 
M_3 \, , \, \res{\tilde{c}}\res g N_3 \big) \in \rel$.\\

- Let us consider the pair $\big( \res {\tilde{c}} M_3 \, , \, \res{\tilde{c}}\res g N_3 \big)$. 
\begin{itemize}
\item Let $\res {\tilde{c}} M_3 \trans{\alpha} \hat{M}$, for $\alpha \neq \sigma$. This case is not admissible. 
\item Let  $\res {\tilde{c}} M_3 \trans{\sigma} \res {\tilde{c}} M'_3$, where 
 \[
M'_3 = \nodep{n_P}{\conf {\I_P}{LightCtrl}}{\mob}{loc1}   \; \big| \;
\nodep{n_{1}}{\conf {\I'_{1}}{L_1}}{\stat}{loc1} \; \big| \; \nodep{n_{2}}{\conf {\I_{2}}{\wact{\mathsf{off}}{light_2}.L_2}}{\stat}{loc4} 
\]
because the phone remained in Room1, at location $loc1$. By two applications 
of Law~\ref{law1} and one application of Law~\ref{law4}  of Thm.~\ref{thm:algebraic-laws} 
we get:
\[
\begin{array}{rcl}
\res {\tilde{c}} M'_3 & \gtrsim & 
\res{\tilde{c}} \big( \nodep{n_P}{\conf {\I_P}{\sigma.LightCtrl}}{\mob}{loc1}   \; \big| \;
\nodep{n_{1}}{\conf {\I'_{1}}{\sigma.L_1}}{\stat}{loc1} \; \big| \; \nodep{n_{2}}{\conf {\I_{2}}{L_2}}{\stat}{loc4} \big)
\\[3pt]
& = & \res{\tilde{c}} M_3 \enspace . 
\end{array}
\]
Then, 
$\res{\tilde{c},g} N_3 \trans{\sigma} \Trans{}\res{\tilde{c},g} N_3 $, 
and obviously $\big( \res{\tilde{c}} M_3 \, , \, \res{\tilde{c},g} N_3\big) \in \rel$.  
\item Let  $\res {\tilde{c}} M_3 \trans{\sigma} \res {\tilde{c}} M'_3$, where 
 \[
M'_3 = \nodep{n_P}{\conf {\I_P}{LightCtrl}}{\mob}{k}   \; \big| \;
\nodep{n_{1}}{\conf {\I'_{1}}{L_1}}{\stat}{loc1} \; \big| \; \nodep{n_{2}}{\conf {\I_{2}}{\wact{\mathsf{off}}{light_2}.L_2}}{\stat}{loc4}
\]
with $k \not \in \{ loc1, loc2, loc3, loc4 \}$, i.e.\ the phone 
moved out of the house.  
By applying Law~\ref{law1} of Thm.~\ref{thm:algebraic-laws} we get
\[
\begin{array}{rcl}
\res{{\tilde{c}}}M'_3 & \gtrsim & \res {\tilde{c}}\nodep{n_P}{\conf {\I_P}{LightCtrl}}{\mob}{k}   \; \big| \;
\nodep{n_{1}}{\conf {\I'_{1}}{L_1}}{\stat}{loc1} \; \big| \; \nodep{n_{2}}{\conf {\I_{2}}{L_2}}{\stat}{loc4}
\\[3pt]
& = & \res{{\tilde{c}}}M_4 \enspace . 
\end{array}
\]
Then there is $N_4$ such that 
$\res{{\tilde{c}},g} N_3 \trans{\sigma} \Trans{} \res{{\tilde{c}},g} N_4$ with: 
\[
N_4 = \nodep{n_P}{\conf {\I_P}{\sigma.\overline{LightCtrl}}}{\mob}{k} \q \big| \q 
\nodep{n_1}{\conf {\I'_1}{L_1}}{\stat}{loc1} \q \big| 
\q \nodep{n_2}{\conf {\I_2}{L_2}}{\stat}{loc4} \q \big| \q
\nodep{n_{LM}}{\conf {\emptyset}{\sigma.\overline{CLM}}}{\stat}{loc3} 
\]
and  $\big( \res{{\tilde{c}}}M_4 \, , \, \res{{\tilde{c}},g}N_4 \big) \in \rel$.

\item  Let  $\res {\tilde{c}} M_3 \trans{\sigma} \res {\tilde{c}} M'_3$, where 
 \[
M'_3 = \nodep{n_P}{\conf {\I_P}{LightCtrl}}{\mob}{loc2}   \; \big| \;
\nodep{n_{1}}{\conf {\I'_{1}}{L_1}}{\stat}{loc1} \; \big| \; \nodep{n_{2}}{\conf {\I_{2}}{\wact{\mathsf{off}}{light_2}.L_2}}{\stat}{loc4}
\]
because the phone moved from $loc1$ to $loc2$. In this case, 
by applying Law~\ref{law1} of Thm.~\ref{thm:algebraic-laws} 
we have: 
\[
\begin{array}{rcl}
\res {\tilde{c}} M'_3 & \gtrsim & 
\res{\tilde{c}}
\big( \nodep{n_P}{\conf {\I_P}{LightCtrl}}{\mob}{loc2}   \; \big| \;
\nodep{n_{1}}{\conf {\I'_{1}}{L_1}}{\stat}{loc1} \; \big| \; \nodep{n_{2}}{\conf {\I_{2}}{L_2}}{\stat}{loc4} \big)
\\[3pt]
& = & \res{\tilde{c}} M_7 \enspace . 
\end{array}
\]
Then, we have that 
$\res{\tilde{c},g} N_3 \trans{\sigma} \Trans{}\res{\tilde{c},g} N_7 $, 
where 
\[
N_7 = \nodep{n_P}{\conf {\I_P}{\sigma.\overline{LightCtrl}}}{\mob}{loc2} 
 \; \big| 
\;  \nodep{n_{1}}{\conf {\I'_{1}}{L_1}}{\stat}{loc1} \; \big| \; \nodep{n_{2}}{\conf {\I_{2}}{L_2}}{\stat}{loc4} \; \big| \; 
 \nodep{n_{LM}}{\conf {\emptyset}{\sigma.\overline{CLM}}}{\stat}{loc3}
\]
and  $\big( \res{\tilde{c}} M_7 \, , \, \res{\tilde{c},g} N_7\big) \in \rel$.

\end{itemize}
The case analysis  when $\res{{\tilde{c}},g} N_3 \trans{\alpha} \hat{N}$ is similar. \\

- Let us consider the pair $\big( \res {\tilde{c}} M_4 \, , \, \res{\tilde{c}}\res g N_4 \big)$. We proceed by case analysis on why $\res {\tilde{c}} M_4
\trans{\alpha} \hat{M}$. 

\begin{itemize}
\item Let $\res {\tilde{c}} M_4
\trans{\alpha} \hat{M}$, with $\alpha \neq \sigma$. This case is not 
admissible. 
\item Let  $\res {\tilde{c}} M_4 \trans{\sigma} \res {\tilde{c}} M'_4$, where 
 \[
M'_4 = \nodep{n_P}{\conf {\I_P}{LightCtrl}}{\mob}{k}   \; \big| \;
\nodep{n_{1}}{\conf {\I'_{1}}{\wact{\mathsf{off}}{light_1}.L_1}}{\stat}{loc1} \; \big| \; \nodep{n_{2}}{\conf {\I_{2}}{\wact{\mathsf{off}}{light_2}.L_2}}{\stat}{loc4}
\]
with $k \not \in \{ loc1, loc2, loc3, loc4 \}$, because the phone remains 
outside.  Then by an application 
of Law~\ref{law1} of Thm.~\ref{thm:algebraic-laws} we get 
\[
\begin{array}{rcl}
\res {\tilde{c}} M'_4 & \gtrsim & 
\res {\tilde{c}} \big( \nodep{n_P}{\conf {\I_P}{LightCtrl}}{\mob}{k} 
  \; \big| \;
\nodep{n_{1}}{\conf {\I'_{1}}{\wact{\mathsf{off}}{light_1}.L_1}}{\stat}{loc1} \; \big| \; \nodep{n_{2}}{\conf {\I_{2}}{L_2}}{\stat}{loc4} \big) \\[3pt]
& = & \res {\tilde{c}}M_5 \enspace . 
\end{array}
\]
Then there is $N_5$ such that $\res {{\tilde{c}},g}N_4 \trans{\sigma} \Trans{}
\res {{\tilde{c}},g}N_5$, where 
\[
N_5 {=} \nodep{n_P}{\conf {\I_P}{\sigma.\overline{LightCtrl}}}{\mob}{k} \; 
\big| \; 
\nodep{n_1}{\conf {\I'_1}{\wact{\mathsf{off}}{light_1}.L_1}}{\stat}{loc1} \; \big| 
\; \nodep{n_2}{\conf {\I_2}{L_2}}{\stat}{loc4} \; \big| \;
\nodep{n_{LM}}{\conf {\emptyset}{\sigma.\overline{CLM}}}{\stat}{loc3}
\]
with $\big( \res {{\tilde{c}}}M_5 \, , \, \res {{\tilde{c}},g}N_5\big) \in \rel $. 

\item  Let  $\res {\tilde{c}} M_4 \trans{\sigma} \res {\tilde{c}} M'_4$, where 
 \[
M'_4 = \nodep{n_P}{\conf {\I_P}{LightCtrl}}{\mob}{loc1}   \; \big| \;
\nodep{n_{1}}{\conf {\I'_{1}}{\wact{\mathsf{off}}{light_1}.L_1}}{\stat}{loc1} \; \big| \; \nodep{n_{2}}{\conf {\I_{2}}{\wact{\mathsf{off}}{light_2}.L_2}}{\stat}{loc4} 
\]
because the phone re-enter the smarthome at location $loc1$.
Then, by an application
of Law~\ref{law1} of Thm.~\ref{thm:algebraic-laws} we get 
\[
\begin{array}{rcl}
\res {\tilde{c}} M'_4 & \gtrsim & 
\res {\tilde{c}} \big( \nodep{n_P}{\conf {\I_P}{LightCtrl}}{\mob}{loc1} 
  \; \big| \;
\nodep{n_{1}}{\conf {\I'_{1}}{\wact{\mathsf{off}}{light_1}.L_1}}{\stat}{loc1} \; \big| \; \nodep{n_{2}}{\conf {\I_{2}}{L_2}}{\stat}{loc4} \big) \\[3pt]
& = & \res {\tilde{c}}M_6 \enspace . 
\end{array}
\]
Then there is $N_6$ such that $\res {{\tilde{c}},g}N_4 \trans{\sigma} \Trans{}
\res {{\tilde{c}},g}N_6$ where 
\[
\begin{array}{rcl}
N_6  & =  & \nodep{n_P}{\conf {\I_P}{\sigma.\overline{LightCtrl}}}{\mob}{loc1} \q \big| \q 
\nodep{n_1}{\conf {\I'_1}{\wact{\mathsf{off}}{light_1}.L_1}}{\stat}{loc1} \q \big| \\[3pt]
&& \nodep{n_2}{\conf {\I_2}{L_2}}{\stat}{loc4} \q \big| \q
\nodep{n_{LM}}{\conf {\emptyset}{\timeout{\OUT{c_1}{}.\sigma.\overline{CLM}}}{CLM}}{\stat}{loc3}
\end{array}
\]
with $\big( \res {{\tilde{c}}}M_6 \, , \, \res {{\tilde{c}},g}N_6\big) \in \rel $. 
\end{itemize}
The case analysis when $\res{{\tilde{c}},g} N_4 \trans{\alpha} \hat{N}$ is similar. \\

- Let us consider the pair $\big( \res {\tilde{c}} M_5 \, , \, \res{\tilde{c}}\res g N_5 \big)$. The only possible transition in both networks is a strong 
transition $\trans{light_1}$ which leads, up to expansion,  to the pair $\big( \res {\tilde{c}} 
M_1 \, , \, \res{\tilde{c}}\res g N_1 \big) \in \rel$.\\

- Let us consider the pair $\big( \res {\tilde{c}} M_6 \, , \, \res{\tilde{c}}\res g N_6 \big)$.
\begin{itemize}
\item Let  $\res {\tilde{c}} M_6
\trans{\alpha} \hat{M}$, with $\alpha \neq light_1$. This case is not
admissible. 
\item Let  $\res {\tilde{c}} M_6
\trans{light_1} \res{\tilde{c}}M'_6$, where
\[
M'_6 = \nodep{n_P}{\conf {\I_P}{LightCtrl}}{\mob}{loc1} 
  \; \big| \;
\nodep{n_{1}}{\conf {\I_{1}}{L_1}}{\stat}{loc1} \; \big| \; \nodep{n_{2}}{\conf {\I_{2}}{L_2}}{\stat}{loc4}
\]
Then, by one application of Law~\ref{law4} of Thm.~\ref{thm:algebraic-laws} we get:
\[
\begin{array}{rcl}
\res{\tilde{c}}M'_6 & \gtrsim & 
\nodep{n_P}{\conf {\I_P}{\sigma.LightCtrl}}{\mob}{loc1} 
  \; \big| \;
\nodep{n_{1}}{\conf {\I_{1}}{\wact{\mathsf{on}}{light_1}.\sigma.L_1}}{\stat}{loc1} \; \big| \; \nodep{n_{2}}{\conf {\I_{2}}{L_2}}{\stat}{loc4}\\[3pt]
& = & \res{\tilde{c}} M_2 \enspace . 
\end{array}
\]
Then it holds that $\res{{\tilde{c}},g}N_6 \trans{light_1} \Trans{} 
\res{{\tilde{c}},g}N_2$, with  $\big( \res {{\tilde{c}}}M_2 \, , \, \res {{\tilde{c}},g}N_2\big) \in \rel $. 
\end{itemize}
As the pair $(M_6, N_6)$ is a bit different from the others let us do 
a case analysis on  when $\res{{\tilde{c}},g} N_6 \trans{\alpha} \hat{N}$. 
\begin{itemize}
\item Let  $\res {{\tilde{c}},g} N_6
\trans{\alpha} \hat{N}$, with $\alpha \neq light_1$. This case is not
admissible. 

\item Let $\res {{\tilde{c}},g} N_6 \trans{light_1} \res {{\tilde{c}},g} N'_6$, 
where
\[
\begin{array}{rcl}
N'_6 & = & \nodep{n_P}{\conf {\I_P}{\sigma.\overline{LightCtrl}}}{\mob}{loc1} 
\; \big| \; 
\nodep{n_1}{\conf {\I_1}{L_1}}{\stat}{loc1} \; \big| \; 
 \nodep{n_2}{\conf {\I_2}{L_2}}{\stat}{loc4} \; \big| \;\\[3pt]
&&\nodep{n_{LM}}{\conf {\emptyset}{\timeout{\OUT{c_1}{}.\sigma.\overline{CLM}}}{CLM}}{\stat}{loc3}
\end{array}
\]
Then, by an application of Law~\ref{law4} of Thm.~\ref{thm:algebraic-laws}
we have:
\[
\begin{array}{rcl}
\res {{\tilde{c}},g} N'_6 & = & 
\res {{\tilde{c}},g} \big( \nodep{n_P}{\conf {\I_P}{\sigma.\overline{LightCtrl}}}{\mob}{loc1} 
\q \big| \q 
\nodep{n_1}{\conf {\I_1}{L_1}}{\stat}{loc1} \q \big| \q
 \nodep{n_2}{\conf {\I_2}{L_2}}{\stat}{loc4} \q \big| \q\\[3pt]
&&\nodep{n_{LM}}{\conf {\emptyset}{\timeout{\OUT{c_1}{}.\sigma.\overline{CLM}}}{CLM}}{\stat}{loc3} \big)\\[3pt]
& \gtrsim & \res {{\tilde{c}},g} \big( \nodep{n_P}{\conf {\I_P}{\sigma.\overline{LightCtrl}}}{\mob}{loc1} 
\q \big| \q
\nodep{n_1}{\conf {\I_1}{\wact{\mathsf{on}}{light_1}.\sigma.L_1}}{\stat}{loc1} \q \big| \q \\[3pt]
&& \nodep{n_2}{\conf {\I_2}{L_2}}{\stat}{loc4} \q \big| \q
\nodep{n_{LM}}{\conf {\emptyset}{\sigma.\overline{CLM}}}{\stat}{loc3} \big)\\[3pt]
& = & \res {{\tilde{c}},g} N_2. 
\end{array}
\]
Then it is easy to see that $\res{\tilde{c}}M_6 \trans{light_1}\Trans{} \res{\tilde{c}}M_2$, with $(M_2, N_2) \in \rel$.
\end{itemize}

- Let us consider the pair $\big( \res {\tilde{c}} M_7 \, , \, \res{\tilde{c}}\res g N_7 \big)$. 
\begin{itemize}
\item Let  $\res {\tilde{c}} M_7
\trans{\sigma} \hat{M}$, with $\alpha \neq \sigma$. This case is not
admissible.

\item Let $\res{\tilde{c}}M_7 \trans{\sigma} \res{\tilde{c}} M'_7$ where
\[
M'_7 = \nodep{n_P}{\conf {\I_P}{LightCtrl}}{\mob}{loc3} 
  \; \big| \;
\nodep{n_{1}}{\conf {\I'_{1}}{\wact{\mathsf{off}}{light_1}.L_1}}{\stat}{loc1} \; \big| \; \nodep{n_{2}}{\conf {\I_{2}}{\wact{\mathsf{off}}{light_1}.L_2}}{\stat}{loc4}
\]
because the smartphone moved from $loc2$ to $loc3$. Then, by an application 
of Law~\ref{law1} of Thm.~\ref{thm:algebraic-laws} we have 
\[
\begin{array}{rcl}
\res {\tilde{c}} M'_7 &\gtrsim & 
\res{\tilde{c}}\big( \nodep{n_P}{\conf {\I_P}{LightCtrl}}{\mob}{loc3} 
  \; \big| \;
\nodep{n_{1}}{\conf {\I'_{1}}{\wact{\mathsf{off}}{light_1}.L_1}}{\stat}{loc1} \; \big| \; \nodep{n_{2}}{\conf {\I_{2}}{L_2}}{\stat}{loc4} \big)
\\[3pt]
& = & \res{\tilde{c}} M_8 \enspace . 
\end{array}
\]
Then, we can derive $N_8$ such that $\res{{\tilde{c}},g}N_7 \trans{\sigma}
\Trans{}\res{{\tilde{c}},g}N_8 $ where 
\[
\begin{array}{rcl}
N_8 & = & \nodep{n_P}{\conf {\I_P}{\sigma.\overline{LightCtrl}}}{\mob}{loc3} 
 \; \big| 
\;  \nodep{n_{1}}{\conf {\I'_{1}}{\wact{\mathsf{off}}{light_1}.L_1}}{\stat}{loc1} \; \big| \; \nodep{n_{2}}{\conf {\I_{2}}{L_2}}{\stat}{loc4} \; \big| \; \\[3pt]
&& 
 \nodep{n_{LM}}{\conf {\emptyset}{\sigma.\overline{CLM}}}{\stat}{loc3}
\end{array}
\]
and $\big( \res {\tilde{c}} M_8 \, , \, \res{\tilde{c}}\res g N_8 \big) \in \rel$.

\item Let $\res{\tilde{c}}M_7 \trans{\sigma} \res{\tilde{c}} M'_7$ where
\[
M'_7 = \nodep{n_P}{\conf {\I_P}{LightCtrl}}{\mob}{loc1} 
  \; \big| \;
\nodep{n_{1}}{\conf {\I'_{1}}{\wact{\mathsf{off}}{light_1}.L_1}}{\stat}{loc1} \; \big| \; \nodep{n_{2}}{\conf {\I_{2}}{\wact{\mathsf{off}}{light_1}.L_2}}{\stat}{loc4}
\]
because the smartphone moved back from $loc2$ to  $loc1$. Then, by an  application  
of Law~\ref{law1} of Thm.~\ref{thm:algebraic-laws} we have 
\[
\begin{array}{rcl}
\res {\tilde{c}} M'_7 &\gtrsim &
\res{\tilde{c}} \big(  \nodep{n_P}{\conf {\I_P}{LightCtrl}}{\mob}{loc1} 
  \; \big| \;
\nodep{n_{1}}{\conf {\I'_{1}}{\wact{\mathsf{off}}{light_1}.L_1}}{\stat}{loc1} \; \big| \; \nodep{n_{2}}{\conf {\I_{2}}{L_2}}{\stat}{loc4} \big)\\[3pt]
& = & \res{\tilde{c}} M_6 \enspace . 
\end{array}
\]
Then, it holds that  $\res{{\tilde{c}},g}N_7 \trans{\sigma}
\Trans{}\res{{\tilde{c}},g}N_6 $, 
and $\big( \res {\tilde{c}} M_6 \, , \, \res{\tilde{c}}\res g N_6 \big) \in \rel$.

\item Let $\res{\tilde{c}}M_7 \trans{\sigma} \res{\tilde{c}} M'_7$ where
\[
M'_7 = \nodep{n_P}{\conf {\I_P}{LightCtrl}}{\mob}{loc2} 
  \; \big| \;
\nodep{n_{1}}{\conf {\I'_{1}}{\wact{\mathsf{off}}{light_1}.L_1}}{\stat}{loc2} \; \big| \; \nodep{n_{2}}{\conf {\I_{2}}{\wact{\mathsf{off}}{light_1}.L_2}}{\stat}{loc4}
\]
because the smartphone remained at location $loc2$. Then, by an application 
of Law~\ref{law1} of Thm.~\ref{thm:algebraic-laws} we have 
\[
\begin{array}{rcl}
\res {\tilde{c}} M'_7 &\gtrsim & 
\res{\tilde{c}} \big( \nodep{n_P}{\conf {\I_P}{LightCtrl}}{\mob}{loc2} 
  \; \big| \;
\nodep{n_{1}}{\conf {\I'_{1}}{\wact{\mathsf{off}}{light_1}.L_1}}{\stat}{loc1} \; \big| \; \nodep{n_{2}}{\conf {\I_{2}}{L_2}}{\stat}{loc4} \big)\\[3pt]
& = & \res{\tilde{c}} M_9 \enspace . 
\end{array}
\]
Then, we can derive $N_9$ such that $\res{{\tilde{c}},g}N_7 \trans{\sigma}
\Trans{}\res{{\tilde{c}},g}N_9 $ where 
\[
\begin{array}{rcl}
N_9 & = & \nodep{n_P}{\conf {\I_P}{\sigma.\overline{LightCtrl}}}{\mob}{loc2} 
 \; \big| 
\;  \nodep{n_{1}}{\conf {\I'_{1}}{\wact{\mathsf{off}}{light_1}.L_1}}{\stat}{loc1} \; \big| \; \nodep{n_{2}}{\conf {\I_{2}}{L_2}}{\stat}{loc4} \; \big| \; \\[3pt]
&& \nodep{n_{LM}}{\conf {\emptyset}{\sigma.\overline{CLM}}}{\stat}{loc3}
\end{array}
\]
and $\big( \res {\tilde{c}} M_9 \, , \, \res{\tilde{c}}\res g N_9 \big) \in \rel$.
\end{itemize}
The case analysis when $\res{{\tilde{c}},g} N_7 \trans{\alpha} \hat{N}$ is similar. \\

- Let us consider the pair $\big( \res {\tilde{c}} M_8 \, , \, \res{\tilde{c}}\res g N_8 \big)$. 
\begin{itemize}
\item Let $\res {\tilde{c}} M_8
\trans{\alpha} \hat{M}$, $\alpha \neq {light_1}$. This case is not admissible. 

\item $\res {\tilde{c}} M_8 \trans{light_1} \res{\tilde{c}}M_{10}$, where
\[
M_{10} =  \nodep{n_P}{\conf {\I_P}{LightCtrl}}{\mob}{loc3} 
  \; \big| \;
\nodep{n_{1}}{\conf {\I_{1}}{L_1}}{\stat}{loc1} \; \big| \; \nodep{n_{2}}{\conf {\I_{2}}{L_2}}{\stat}{loc4} \enspace .
\]
Then there is $N_{10}$ such that $\res{{\tilde{c}},g} N_8 \trans{light_1}
\res{{\tilde{c}},g} N_{10}$, where 
\[
N_{10} = \nodep{n_P}{\conf {\I_P}{\sigma.\overline{LightCtrl}}}{\mob}{loc3} 
 \; \big| 
\;  \nodep{n_{1}}{\conf {\I_{1}}{L_1}}{\stat}{loc1} \; \big| \; \nodep{n_{2}}{\conf {\I_{2}}{L_2}}{\stat}{loc4} \; \big| \; 
 \nodep{n_{LM}}{\conf {\emptyset}{\sigma.\overline{CLM}}}{\stat}{loc3}
\]
and $\big( \res {\tilde{c}} M_{10} \, , \, \res{{\tilde{c}},g} N_{10} \big) \in \rel$.
\end{itemize}
The case analysis when $\res{{\tilde{c}},g} N_8 \trans{\alpha} \hat{N}$ is similar. \\

- Let us consider the pair $\big( \res {\tilde{c}} M_9 \, , \, \res{\tilde{c}}\res g N_9 \big)$. 
\begin{itemize}
\item Let $\res {\tilde{c}} M_9
\trans{\alpha} \hat{M}$, $\alpha \neq {light_1}$. This case is not admissible. 

\item $\res {\tilde{c}} M_9 \trans{light_1} \res{\tilde{c}}M_{11}$, where
\[
M_{11} =  \nodep{n_P}{\conf {\I_P}{LightCtrl}}{\mob}{loc2} 
  \; \big| \;
\nodep{n_{1}}{\conf {\I_{1}}{L_1}}{\stat}{loc1} \; \big| \; \nodep{n_{2}}{\conf {\I_{2}}{L_2}}{\stat}{loc4} \enspace .
\]
Then there is $N_{11}$ such that $\res{{\tilde{c}},g} N_9 \trans{light_1}
\res{{\tilde{c}},g} N_{11}$, where 
\[
N_{11} = \nodep{n_P}{\conf {\I_P}{\sigma.\overline{LightCtrl}}}{\mob}{loc2} 
 \; \big| 
\;  \nodep{n_{1}}{\conf {\I_{1}}{L_1}}{\stat}{loc1} \; \big| \; \nodep{n_{2}}{\conf {\I_{2}}{L_2}}{\stat}{loc4} \; \big| \; 
 \nodep{n_{LM}}{\conf {\emptyset}{\sigma.\overline{CLM}}}{\stat}{loc3}
\]
and $\big( \res {\tilde{c}} M_{11} \, , \, \res{{\tilde{c}},g} N_{11} \big) \in \rel$.
\end{itemize}
The case analysis when $\res{{\tilde{c}},g} N_9 \trans{\alpha} \hat{N}$ is similar. \\

- Let us consider the pair $\big( \res {\tilde{c}} M_{10} \, , \, \res{\tilde{c}}\res g N_{10} \big)$. 
\begin{itemize}
\item Let  $\res {\tilde{c}} M_{10}
\trans{\sigma} \hat{M}$, with $\alpha \neq \sigma$. This case is not
admissible.

\item Let $\res{\tilde{c}}M_{10} \trans{\sigma} \res{\tilde{c}} M'_{10}$ where
\[
M'_{10} = \nodep{n_P}{\conf {\I_P}{LightCtrl}}{\mob}{loc4} 
  \; \big| \;
\nodep{n_{1}}{\conf {\I_{1}}{\wact{\mathsf{off}}{light_1}.L_1}}{\stat}{loc1} \; \big| \; \nodep{n_{2}}{\conf {\I_{2}}{\wact{\mathsf{off}}{light_2}.L_2}}{\stat}{loc4}
\]
because the smartphone moved from $loc3$ to $loc4$. Then, by two applications 
of Law~\ref{law1} and one application of Law~\ref{law4} of  
  Thm.~\ref{thm:algebraic-laws} we have:
\[
\begin{array}{rcl}
\res {\tilde{c}} M'_{10} & {\gtrsim} & 
\res{\tilde{c}} ( \nodep{n_P}{\conf {\I_P}{\sigma.LightCtrl}}{\mob}{loc4} 
  \; \big| \;
\nodep{n_{1}}{\conf {\I_{1}}{L_1}}{\stat}{loc1} \; \big| \; 
\nodep{n_{2}}{\conf {\I_{2}}{\wact{\mathsf{on}}{light_2}.\sigma.L_2}}{\stat}{loc4} )
\\[3pt]
& {=} & \res{\tilde{c}} M_{12} \enspace . 
\end{array}
\]
Then, we can derive $N_{12}$ such that $\res{{\tilde{c}},g}N_{10} \trans{\sigma}
\Trans{}\res{{\tilde{c}},g}N_{12} $ where 
\[
\begin{array}{rcl}
N_{12} & = & \nodep{n_P}{\conf {\I_P}{\sigma.\overline{LightCtrl}}}{\mob}{loc4} 
 \; \big| 
\;  \nodep{n_{1}}{\conf {\I_{1}}{L_1}}{\stat}{loc1} \; \big| \; \nodep{n_{2}}{\conf {\I_{2}}{\wact{\mathsf{on}}{light_2}.\sigma.L_2}}{\stat}{loc4} \; \big| \; 
\\[3pt]
&& \nodep{n_{LM}}{\conf {\emptyset}{\sigma.\overline{CLM}}}{\stat}{loc3}
\end{array}
\]
and $\big( \res {\tilde{c}} M_{12} \, , \, \res{\tilde{c}}\res g N_{12} \big) \in \rel$.

\item Let $\res{\tilde{c}}M_{10} \trans{\sigma} \res{\tilde{c}} M'_{10}$ where
\[
M'_{10} = \nodep{n_P}{\conf {\I_P}{LightCtrl}}{\mob}{loc2} 
  \; \big| \;
\nodep{n_{1}}{\conf {\I_{1}}{\wact{\mathsf{off}}{light_1}.L_1}}{\stat}{loc1} \; \big| \; \nodep{n_{2}}{\conf {\I_{2}}{\wact{\mathsf{off}}{light_2}.L_2}}{\stat}{loc4}
\]
because the smartphone moved back from $loc3$ to $loc2$. Then, by two 
 applications  
of Law~\ref{law1} of Thm.~\ref{thm:algebraic-laws} we have 
\[
\begin{array}{rcl}
\res {\tilde{c}} M'_{10} &\gtrsim &
\res{\tilde{c}} \big(  \nodep{n_P}{\conf {\I_P}{LightCtrl}}{\mob}{loc2} 
  \; \big| \;
\nodep{n_{1}}{\conf {\I_{1}}{L_1}}{\stat}{loc1} \; \big| \; \nodep{n_{2}}{\conf {\I_{2}}{L_2}}{\stat}{loc4} \big)\\[3pt]
& = & \res{\tilde{c}} M_{11} \enspace . 
\end{array}
\]
Then,  $\res{{\tilde{c}},g}N_{10} \trans{\sigma}
\Trans{}\res{{\tilde{c}},g}N_{11}$, 
and $\big( \res {\tilde{c}} M_{11} \, , \, \res{\tilde{c}}\res g N_{11} \big) \in \rel$.

\item Let $\res{\tilde{c}}M_{10} \trans{\sigma} \res{\tilde{c}} M'_{10}$ where
\[
M'_{10} = \nodep{n_P}{\conf {\I_P}{LightCtrl}}{\mob}{loc3} 
  \; \big| \;
\nodep{n_{1}}{\conf {\I_{1}}{\wact{\mathsf{off}}{light_1}.L_1}}{\stat}{loc1} \; \big| \; \nodep{n_{2}}{\conf {\I_{2}}{\wact{\mathsf{off}}{light_2}.L_2}}{\stat}{loc4}
\]
because the smartphone remained at location $loc3$. Then, by two 
 applications  
of Law~\ref{law1} of Thm.~\ref{thm:algebraic-laws} we have 
\[
\begin{array}{rcl}
\res {\tilde{c}} M'_{10} &\gtrsim &
\res{\tilde{c}} \big(  \nodep{n_P}{\conf {\I_P}{LightCtrl}}{\mob}{loc3} 
  \; \big| \;
\nodep{n_{1}}{\conf {\I_{1}}{L_1}}{\stat}{loc1} \; \big| \; \nodep{n_{2}}{\conf {\I_{2}}{L_2}}{\stat}{loc4} \big)\\[3pt]
& = & \res{\tilde{c}} M_{10} \enspace . 
\end{array}
\]
Then,  $\res{{\tilde{c}},g}N_{10} \trans{\sigma}
\Trans{}\res{{\tilde{c}},g}N_{10}$, 
and $\big( \res {\tilde{c}} M_{10} \, , \, \res{\tilde{c}}\res g N_{10} \big) \in \rel$.

\end{itemize}
The case analysis when $\res{{\tilde{c}},g} N_{10} \trans{\alpha} \hat{N}$ is similar. \\

- Let us consider the pair $\big( \res {\tilde{c}} M_{11} \, , \, \res{\tilde{c}}\res g N_{11} \big)$. 
\begin{itemize}
\item Let  $\res {\tilde{c}} M_{11}
\trans{\sigma} \hat{M}$, with $\alpha \neq \sigma$. This case is not
admissible. 

\item Let $\res{\tilde{c}}M_{11} \trans{\sigma} \res{\tilde{c}} M'_{11}$ where
\[
M'_{11} = \nodep{n_P}{\conf {\I_P}{LightCtrl}}{\mob}{loc3} 
  \; \big| \;
\nodep{n_{1}}{\conf {\I_{1}}{\wact{\mathsf{off}}{light_1}.L_1}}{\stat}{loc1} \; \big| \; \nodep{n_{2}}{\conf {\I_{2}}{\wact{\mathsf{off}}{light_2}.L_2}}{\stat}{loc4}
\]
because the smartphone moved from $loc2$ to $loc3$. Then, by two applications 
of Law~\ref{law1}  of  
  Thm.~\ref{thm:algebraic-laws} we have:
\[
\begin{array}{rcl}
\res {\tilde{c}} M'_{11} &\gtrsim & 
\res{\tilde{c}} ( \nodep{n_P}{\conf {\I_P}{LightCtrl}}{\mob}{loc3} 
  \; \big| \;
\nodep{n_{1}}{\conf {\I_{1}}{L_1}}{\stat}{loc1} \; \big| \; 
\nodep{n_{2}}{\conf {\I_{2}}{L_2}}{\stat}{loc4} )
\\[3pt]
& = & \res{\tilde{c}} M_{10} \enspace . 
\end{array}
\]
Then,  $\res{{\tilde{c}},g}N_{11} \trans{\sigma}
\Trans{}\res{{\tilde{c}},g}N_{10} $, 
and $\big( \res {\tilde{c}} M_{10} \, , \, \res{\tilde{c}}\res g N_{10} \big) \in \rel$.

\item Let $\res{\tilde{c}}M_{11} \trans{\sigma} \res{\tilde{c}} M'_{11}$ where
\[
M'_{11} = \nodep{n_P}{\conf {\I_P}{LightCtrl}}{\mob}{loc1} 
  \; \big| \;
\nodep{n_{1}}{\conf {\I_{1}}{\wact{\mathsf{off}}{light_1}.L_1}}{\stat}{loc1} \; \big| \; \nodep{n_{2}}{\conf {\I_{2}}{\wact{\mathsf{off}}{light_2}.L_2}}{\stat}{loc4}
\]
because the smartphone moved back from $loc2$ to $loc1$. Then, by two 
 applications  
of Law~\ref{law1}  and one application of Law~\ref{law4} of Thm.~\ref{thm:algebraic-laws} we have 
\[
\begin{array}{rcl}
\res {\tilde{c}} M'_{11} &\gtrsim &
\res{\tilde{c}} \big(  \nodep{n_P}{\conf {\I_P}{\sigma.LightCtrl}}{\mob}{loc1} 
  \; \big| \;
\nodep{n_{1}}{\conf {\I_{1}}{\wact{\mathsf{on}}{light_1}.\sigma.L_1}}{\stat}{loc1} \; \big| \; \nodep{n_{2}}{\conf {\I_{2}}{L_2}}{\stat}{loc4} \big)\\[3pt]
& = & \res{\tilde{c}} M_{2} \enspace . 
\end{array}
\]
Then,  $\res{{\tilde{c}},g}N_{11} \trans{\sigma}
\Trans{}\res{{\tilde{c}},g}N_{2}$, 
and $\big( \res {\tilde{c}} M_{2} \, , \, \res{\tilde{c}}\res g N_{2} \big) \in \rel$.

\item Let $\res{\tilde{c}}M_{11} \trans{\sigma} \res{\tilde{c}} M'_{11}$ where
\[
M'_{11} = \nodep{n_P}{\conf {\I_P}{LightCtrl}}{\mob}{loc2} 
  \; \big| \;
\nodep{n_{1}}{\conf {\I_{1}}{\wact{\mathsf{off}}{light_1}.L_1}}{\stat}{loc1} \; \big| \; \nodep{n_{2}}{\conf {\I_{2}}{\wact{\mathsf{off}}{light_1}.L_2}}{\stat}{loc4}
\]
because the smartphone remained at location $loc2$. Then, by two 
 applications  
of Law~\ref{law1} of Thm.~\ref{thm:algebraic-laws} we have 
\[
\begin{array}{rcl}
\res {\tilde{c}} M'_{11} &\gtrsim &
\res{\tilde{c}} \big(  \nodep{n_P}{\conf {\I_P}{LightCtrl}}{\mob}{loc2} 
  \; \big| \;
\nodep{n_{1}}{\conf {\I_{1}}{L_1}}{\stat}{loc1} \; \big| \; \nodep{n_{2}}{\conf {\I_{2}}{L_2}}{\stat}{loc4} \big)\\[3pt]
& = & \res{\tilde{c}} M_{11} \enspace . 
\end{array}
\]
Then,  $\res{{\tilde{c}},g}N_{11} \trans{\sigma}
\Trans{}\res{{\tilde{c}},g}N_{11}$, 
and $\big( \res {\tilde{c}} M_{11} \, , \, \res{\tilde{c}}\res g N_{11} \big) \in \rel$.

\end{itemize}
The case analysis when $\res{{\tilde{c}},g} N_{11} \trans{\alpha} \hat{N}$ is similar. \\

- Let us consider the pair $\big( \res {\tilde{c}} M_{12} \, , \, \res{\tilde{c}}\res g N_{12} \big)$. The only possible transition in both networks is a strong 
transition $\trans{light_2}$ which leads to the pair $\big( \res {\tilde{c}} 
M_{13} \, , \, \res{\tilde{c}}\res g N_{13} \big) \in \rel$.\\

- Let us consider the pair $\big( \res {\tilde{c}} M_{13} \, , \, \res{\tilde{c}}\res g N_{13} \big)$. 
\begin{itemize}
\item Let  $\res {\tilde{c}} M_{13}
\trans{\sigma} \hat{M}$, with $\alpha \neq \sigma$. This case is not
admissible.

\item Let $\res{\tilde{c}}M_{13} \trans{\sigma} \res{\tilde{c}} M'_{13}$ where
\[
M'_{13} = \nodep{n_P}{\conf {\I_P}{LightCtrl}}{\mob}{loc3} 
  \; \big| \;
\nodep{n_{1}}{\conf {\I_{1}}{\wact{\mathsf{off}}{light_1}.L_1}}{\stat}{loc1} \; \big| \; \nodep{n_{2}}{\conf {\I'_{2}}{L_2}}{\stat}{loc4}
\]
because the smartphone moved back from $loc4$ to $loc3$. Then, by an 
 application  
of Law~\ref{law1} of Thm.~\ref{thm:algebraic-laws} we have 
\[
\begin{array}{rcl}
\res {\tilde{c}} M'_{13} &\gtrsim &
\res{\tilde{c}} \big(  \nodep{n_P}{\conf {\I_P}{LightCtrl}}{\mob}{loc3} 
  \; \big| \;
\nodep{n_{1}}{\conf {\I_{1}}{L_1}}{\stat}{loc1} \; \big| \; \nodep{n_{2}}{\conf {\I'_{2}}{L_2}}{\stat}{loc4} \big)\\[3pt]
& = & \res{\tilde{c}} M_{14} \enspace . 
\end{array}
\]
Then, there is $N_{14}$ such that $\res{{\tilde{c}},g}N_{13} \trans{\sigma}
\Trans{}\res{{\tilde{c}},g}N_{14}$, 
where
\[
N_{14} = \nodep{n_P}{\conf {\I_P}{\sigma.\overline{LightCtrl}}}{\mob}{loc3} 
 \; \big| 
\;  \nodep{n_{1}}{\conf {\I_{1}}{L_1}}{\stat}{loc1} \; \big| \; \nodep{n_{2}}{\conf {\I'_{2}}{L_2}}{\stat}{loc4} \; \big| \; 
 \nodep{n_{LM}}{\conf {\emptyset}{\sigma.\overline{CLM}}}{\stat}{loc3}
\]
and $\big( \res {\tilde{c}} M_{14} \, , \, \res{\tilde{c}}\res g N_{14} \big) \in \rel$.

\item Let $\res{\tilde{c}}M_{13} \trans{\sigma} \res{\tilde{c}} M'_{13}$ where
\[
M'_{13} = \nodep{n_P}{\conf {\I_P}{LightCtrl}}{\mob}{loc4} 
  \; \big| \;
\nodep{n_{1}}{\conf {\I_{1}}{\wact{\mathsf{off}}{light_1}.L_1}}{\stat}{loc1} \; \big| \; \nodep{n_{2}}{\conf {\I'_{2}}{L_2}}{\stat}{loc4}
\]
because the smartphone remained at location $loc4$. Then, by two
 applications  
of Law~\ref{law1} and one application of Law~\ref{law4} of Thm.~\ref{thm:algebraic-laws} we have:
\[
\begin{array}{rcl}
\res {\tilde{c}} M'_{13} &\gtrsim &
\res{\tilde{c}} \big(  \nodep{n_P}{\conf {\I_P}{\sigma.LightCtrl}}{\mob}{loc4} 
  \; \big| \;
\nodep{n_{1}}{\conf {\I_{1}}{L_1}}{\stat}{loc1} \; \big| \; \nodep{n_{2}}{\conf {\I'_{2}}{\sigma.L_2}}{\stat}{loc4} \big)\\[3pt]
& = & \res{\tilde{c}} M_{13} \enspace . 
\end{array}
\]
Then,  $\res{{\tilde{c}},g}N_{13} \trans{\sigma}
\Trans{}\res{{\tilde{c}},g}N_{13}$, 
and $\big( \res {\tilde{c}} M_{13} \, , \, \res{\tilde{c}}\res g N_{13} \big) \in \rel$.

\end{itemize}
The case analysis when $\res{{\tilde{c}},g} N_{13} \trans{\alpha} \hat{N}$ is similar. \\

- Let us consider the pair $\big( \res {\tilde{c}} M_{14} \, , \, \res{\tilde{c}}\res g N_{14} \big)$.
\begin{itemize}
\item Let  $\res {\tilde{c}} M_{14}
\trans{\sigma} \hat{M}$, with $\alpha \neq \sigma$. This case is not
admissible.

\item Let $\res{\tilde{c}}M_{14} \trans{\sigma} \res{\tilde{c}} M'_{14}$ where
\[
M'_{14} = \nodep{n_P}{\conf {\I_P}{LightCtrl}}{\mob}{loc4} 
  \; \big| \;
\nodep{n_{1}}{\conf {\I_{1}}{\wact{\mathsf{off}}{light_1}.L_1}}{\stat}{loc1} \; \big| \; \nodep{n_{2}}{\conf {\I'_{2}}{\wact{\mathsf{off}}{light_2}.L_2}}{\stat}{loc4}
\]
because the smartphone moved from $loc3$ to $loc4$. Then, by an application 
of Law~\ref{law1}  of  
  Thm.~\ref{thm:algebraic-laws} we have:
\[
\begin{array}{rcl}
\res {\tilde{c}} M'_{14} &\gtrsim & 
\res{\tilde{c}}\big( \nodep{n_P}{\conf {\I_P}{LightCtrl}}{\mob}{loc4} 
  \; \big| \;
\nodep{n_{1}}{\conf {\I_{1}}{L_1}}{\stat}{loc1} \; \big| \; 
\nodep{n_{2}}{\conf {\I'_{2}}{\wact{\mathsf{off}}{light_2}.L_2}}{\stat}{loc4} \big)
\\[3pt]
& = & \res{\tilde{c}} M_{15} \enspace . 
\end{array}
\]
Then, we can derive $N_{15}$ such that $\res{{\tilde{c}},g}N_{14} \trans{\sigma}
\Trans{}\res{{\tilde{c}},g}N_{15} $ where 
\[
\begin{array}{rcl}
N_{15} & = & \nodep{n_P}{\conf {\I_P}{\sigma.\overline{LightCtrl}}}{\mob}{loc4} 
 \; \big| 
\;  \nodep{n_{1}}{\conf {\I_{1}}{L_1}}{\stat}{loc1} \; \big| \; \nodep{n_{2}}{\conf {\I'_{2}}{\wact{\mathsf{off}}{light_2}.L_2}}{\stat}{loc4} \; \big| \; \\[3pt]
&& \nodep{n_{LM}}{\conf {\emptyset}{\timeout{\OUT{c_2}{}.\sigma.\overline{CLM}}{\overline{CLM}}}}{\stat}{loc3}
\end{array}
\]
and $\big( \res {\tilde{c}} M_{15} \, , \, \res{\tilde{c}}\res g N_{15} \big) \in \rel$.

\item Let $\res{\tilde{c}}M_{14} \trans{\sigma} \res{\tilde{c}} M'_{14}$ where
\[
M'_{14} = \nodep{n_P}{\conf {\I_P}{LightCtrl}}{\mob}{loc2} 
  \; \big| \;
\nodep{n_{1}}{\conf {\I_{1}}{\wact{\mathsf{off}}{light_1}.L_1}}{\stat}{loc1} \; \big| \; \nodep{n_{2}}{\conf {\I'_{2}}{\wact{\mathsf{off}}{light_2}.L_2}}{\stat}{loc4}
\]
because the smartphone moved from $loc3$ to $loc2$. Then, by an  
 application 
of Law~\ref{law1} of Thm.~\ref{thm:algebraic-laws} we have 
\[
\begin{array}{rcl}
\res {\tilde{c}} M'_{14} &\gtrsim &
\res{\tilde{c}} \big(  \nodep{n_P}{\conf {\I_P}{LightCtrl}}{\mob}{loc2} 
  \; \big| \;
\nodep{n_{1}}{\conf {\I_{1}}{L_1}}{\stat}{loc1} \; \big| \; \nodep{n_{2}}{\conf {\I'_{2}}{\wact{\mathsf{off}}{light_2}.L_2}}{\stat}{loc4} \big)\\[3pt]
& = & \res{\tilde{c}} M_{16} \enspace . 
\end{array}
\]
Then, there is $N_{16}$ such that $\res{{\tilde{c}},g}N_{14} \trans{\sigma}
\Trans{}\res{{\tilde{c}},g}N_{16}$, where
\[
N_{16}  =  \nodep{n_P}{\conf {\I_P}{\sigma.\overline{LightCtrl}}}{\mob}{loc2} 
 \; \big| 
\;  \nodep{n_{1}}{\conf {\I_{1}}{L_1}}{\stat}{loc1} \; \big| \; \nodep{n_{2}}{\conf {\I'_{2}}{\wact{\mathsf{off}}{light_2}.L_2}}{\stat}{loc4} \; \big| \; 
 \nodep{n_{LM}}{\conf {\emptyset}{\sigma.\overline{CLM}}}{\stat}{loc3}
\]
and $\big( \res {\tilde{c}} M_{16} \, , \, \res{\tilde{c}}\res g N_{16} \big) \in \rel$.

\item Let $\res{\tilde{c}}M_{14} \trans{\sigma} \res{\tilde{c}} M'_{14}$ where
\[
M'_{14} = \nodep{n_P}{\conf {\I_P}{LightCtrl}}{\mob}{loc3} 
  \; \big| \;
\nodep{n_{1}}{\conf {\I_{1}}{\wact{\mathsf{off}}{light_1}.L_1}}{\stat}{loc1} \; \big| \; \nodep{n_{2}}{\conf {\I'_{2}}{\wact{\mathsf{off}}{light_2}.L_2}}{\stat}{loc4}
\]
because the smartphone remained at location $loc3$. Then, by an 
 application
of Law~\ref{law1} of Thm.~\ref{thm:algebraic-laws} we have 
\[
\begin{array}{rcl}
\res {\tilde{c}} M'_{14} &\gtrsim &
\res{\tilde{c}} \big(  \nodep{n_P}{\conf {\I_P}{LightCtrl}}{\mob}{loc3} 
  \; \big| \;
\nodep{n_{1}}{\conf {\I_{1}}{L_1}}{\stat}{loc1} \; \big| \; \nodep{n_{2}}{\conf {\I'_{2}}{\wact{\mathsf{off}}{light_2}.L_2}}{\stat}{loc4} \big)\\[3pt]
& = & \res{\tilde{c}} M_{17} \enspace . 
\end{array}
\]
Then,  
 there is $N_{17}$ such that $\res{{\tilde{c}},g}N_{14} \trans{\sigma}
\Trans{}\res{{\tilde{c}},g}N_{17}$, where
\[
N_{17}  =  \nodep{n_P}{\conf {\I_P}{\sigma.\overline{LightCtrl}}}{\mob}{loc3} 
 \; \big| 
\;  \nodep{n_{1}}{\conf {\I_{1}}{L_1}}{\stat}{loc1} \; \big| \; \nodep{n_{2}}{\conf {\I'_{2}}{\wact{\mathsf{off}}{light_2}.L_2}}{\stat}{loc4} \; \big| \; 
 \nodep{n_{LM}}{\conf {\emptyset}{\sigma.\overline{CLM}}}{\stat}{loc3}
\]
and $\big( \res {\tilde{c}} M_{17} \, , \, \res{\tilde{c}}\res g N_{17} \big) \in \rel$.
\end{itemize}
The case analysis when $\res{{\tilde{c}},g} N_{14} \trans{\alpha} \hat{N}$ is similar. \\

- Let us consider the pair $\big( \res {\tilde{c}} M_{15} \, , \, \res{\tilde{c}}\res g N_{15} \big)$. The only possible transition in both networks is a strong 
transition $\trans{light_2}$ which leads, up to expansion,  to the pair $\big( \res {\tilde{c}} 
M_{12} \, , \, \res{\tilde{c}}\res g N_{12} \big) \in \rel$.\\

- Let us consider the pair $\big( \res {\tilde{c}} M_{16} \, , \, \res{\tilde{c}}\res g N_{16} \big)$. The only possible transition in both networks is a strong 
transition $\trans{light_2}$ which leads, up to expansion,  to the pair $\big( \res {\tilde{c}} 
M_{11} \, , \, \res{\tilde{c}}\res g N_{11} \big) \in \rel$.\\

- Let us consider the pair $\big( \res {\tilde{c}} M_{17} \, , \, \res{\tilde{c}}\res g N_{17} \big)$. The only possible transition in both networks is a strong 
transition $\trans{light_2}$ which leads, up to expansion,  to the pair $\big( \res {\tilde{c}} 
M_{10} \, , \, \res{\tilde{c}}\res g N_{10} \big) \in \rel$.\\
\hfill\qed

%%%%%%%%%%%%%%%%%%%%%%%%%%%%%%%%%%%%%%%%%%%%%%%%%%%%%%%%%%%%%%%

%%\section{Appendix G: Bisimulatiion contextuality}

Next goal is to prove Thm.~\ref{thm:congruence}. For that we
 need a simple technical 
lemma on the operator defined in Def.~\ref{def:redi} of this Appendix. 
\begin{lemma}
\label{lem:update-sens}
For any network $M$,  sensor $s$, location $h$ and value $v$ in the 
domain of $s$, it follows that $\ri{M[{s@h} \mapsto v]} = \ri{M}$.
\end{lemma}
\begin{proof} By straightforward induction on the structure of $M$. 
\end{proof}

\paragraph{\textbf{Proof of Thm.~\ref{thm:congruence}.}}
We prove that the bisimilarity relation, $\approx$,  is preserved by network 
contexts, i.e.\ parallel composition, channel restriction and 
sensor updating.

%The symmetric case of $\res c N  \trans{\alpha} \res c N' $   is analogous.\\

\emph{Let us prove that $\approx$ is preserved by parallel composition.}
We show that the relation 
\[
\rel=\left \{(M | O,\; N | O)  :  \textrm{ s.t.  } M | O \textrm{ and } N | O \textrm{ are well-formed and }  M \approx N \right \}
\]
is a bisimulation.
We proceed by case analysis on why $M | O \trans{\alpha} \hat{M}$.

\begin{itemize}
\item Let  $M | O \trans{\tau} \hat{M}$. 
We can distinguish two cases.
\begin{itemize}
\item The transition is derived by applying rule \rulename{GlbCom}:
\[
\Txiombis{M\trans{\send{c}{v}{h}} M' \q O\trans{\rec{c}{v}{k}}O' \q \dist h k \leq \rng c}{M | O \trans{\tau} M' | O'}
\]
with $\hat{M} = M' |O'$. 
Since $M\trans{\send{c}{v}{h}} M'$  and $\dist h k \leq \rng c$, by an application of rule \rulename{SndObs}   we derive  $M\trans{\sendobs{c}{v}{k}}M'$.
As $M\approx N$ there are $N_1$, $N_2$ and $N'$ such that   $N\ttrans{}N_1\trans{\sendobs{c}{v}{k}}N_2\ttrans{} N'$ with $M'\approx N'$.
Thus, there exists a location $h'$ such that  $\dist {h'} k \leq \rng c$ and $N_1\trans{\send{c}{v}{h'}} N_2$.
Therefore,  by several  applications of rule \rulename{ParN} and one application of rule \rulename{GlbCom} we can derive $N | O \Trans{} \hat{N}= N' | O'$, 
with $(\hat{M} , \hat{N} ) \in \rel$. 
The symmetric case is analogous.

\item  The transition is derived by applying rule \rulename{ParN} to $M$:
\[
\Txiombis{M\trans{\tau} M'}{M | O \trans{\tau} M' | O}
\]
As $M \approx N$ it follows that $N\ttrans{} N'$ with $M'\approx N'$.
By several  applications of rule \rulename{ParN} it follows that 
$N | O \Trans{} \hat{N}= N' | O'$, 
with $(\hat{M} , \hat{N} ) \in \rel$. The symmetric case is easier. 
\end{itemize}

\item Let  $M | O \trans{\sigma} \hat{M}= M' | O'$.
This is only possible by an application of  of rule \rulename{TimePar} where: 
\[
\Txiombis{M \trans{\sigma} M' \Q O\trans{\sigma} O' \Q M | O \ntrans{\tau}}{M | O \trans{\sigma} M' | O'}
\]
Since $M \approx N$ and $M \trans{\sigma} M'$ there are $N_1$, $N_2$ and $N'$ such that 
$N \Trans{} N_1 \trans{\sigma}N_2 \Trans{} N'$,  with $M'\approx N'$. By an appropriate number of applications of rule \rulename{ParN} we have that $N | O \Trans{} N_1 | O$. Next step is to show that we can use rule \rulename{TimePar}
to derive $N_1 | O \trans{\sigma} N_2 | O'$. For that we only need to prove
that $N_1 | O \ntrans{\tau}$. In fact, if $N_1 | O \trans{\tau}$ then we would 
reach a contradiction. This because, $M \approx N$ and $N \Trans{} N_1$ 
implies there is $M_1$ such that $M \Trans{} M_1$ with $M_1 \approx N_1$. 
As $M \ntrans{\tau}$ it follows that $M=M_1 \approx N_1$. By Prop.~\ref{prop:maxprog}, $N_1 \trans{\sigma} N_2$ and $O \trans{\sigma} O'$ imply 
$N_1 \ntrans{\tau}$ and $O \ntrans{\tau}$. Thus  $N_1 | O \trans{\tau}$
could be derived only by an application of rule \rulename{GlobCom} where 
$N_1$ interact with $O$ via some channel $c$, with $\rng c \geq 0$.
However, as $N_1 \approx M$ the network $M$ could mimick the same interaction with $O$ giving rise to a reduction of the form $M | O \Trans{}\trans{\tau}$. 
This is in contradiction 
with initial hypothesis that $M | O \ntrans{\tau}$. Thus, $N_1 | O \ntrans{\tau}$ and by an application of rule \rulename{TimePar} we derive $N_1 | O \trans{\sigma} N_2 | O'$. By an appropriate number of applications of rule 
\rulename{ParN} we get  $N_2 | O' \Trans{} N' | O'$. Thus, 
$N | O \Trans{\sigma} \hat{N}=N' | O'$, with $(\hat{M} , \hat{N}) \in \rel$. 

%%, for some   $N_1,\dots ,N_{q+1}$.

\item Let $M | O \trans{a} \hat{M}$. Then we distinguish two cases. 
\begin{itemize}
\item Either $O \trans{a} O'$ and by an application of rule \rulename{ParN}
we derive $M | O \trans{a} M | O'$. This case is easy. 
\item Or $M \trans{a} M'$ and by an application of rule \rulename{ParN}
we derive $M | O \trans{a} M' | O$. As $M \approx N$ there is $N'$ such that 
$N \Trans{a} N'$ and $M' \approx N'$. Thus, by several applications  
of rule \rulename{ParN} we derive $N | O \Trans{a} \hat{N} = N' | O$, 
with $(\hat{M} , \hat{N}) \in \rel$. 
\end{itemize}

\item Let $M | O \trans{\sendobs{c}{v}{k}} \hat{M}$.
By definition of rule \rulename{SndObs} this is only possible if 
$M | O \trans{\send{c}{v}{h}} \hat{M}$, with $\dist h k \leq \rng c$.  Then we distinguish two cases. 
\begin{itemize}
\item Either $O \trans{\send{c}{v}{h}} O'$ and $\hat{M} = M | O'$ by an application of rule \rulename{ParN}. Then, by an application of the same rule 
we derive $N | O \trans{\send{c}{v}{h}} N | O'$. By an application 
of rule \rulename{SndObs} we get $N | O \trans{\sendobs{c}{v}{k}} \hat{N}= N | O'$, with $(\hat{M}, \hat{N} ) \in \rel$. 
\item Or $M \trans{\send{c}{v}{h}} M'$ and $\hat{M} = M' | O$ by an application of rule \rulename{ParN}. By an application of rule \rulename{SndObs} we 
have $M \trans{\sendobs{c}{v}{k}} M'$. As $M \approx N$ there is $N'$ such 
that $N \Trans{\sendobs{c}{v}{k}} N'$, with $M' \approx N'$. As the transition
$\trans{\sendobs{c}{v}{k}}$ can only be derived by an application of rule 
\rulename{SndObs}, it follows that $N \Trans{\send{c}{v}{h'}} N'$, for some $h'$
such that $\dist {h'} k \leq \rng c$. By several applications of rule 
\rulename{ParN} it follows that $N | O \Trans{\send{c}{v}{h'}} N' | O$. 
By an application of rule \rulename{SndObs} we finally obtain $N | O \Trans{\sendobs{c}{v}{k}} \hat{N} = N' | O$, with $(\hat{M} , \hat{N}) \in \rel$. 
\end{itemize}
\item Let $M | O \trans{\recobs{c}{v}{k}} \hat{M}$. This case is similar to the previous one. 
\end{itemize}

\emph{Let us prove that $\approx$ is preserved by channel restriction.}
We show that the relation 
\[
\rel=\{ \big(\res c M ,\; \res c N \big) :  M \approx N\}
\]
is a bisimulation. We proceed by case analysis on 
why $\res c M \trans{\alpha} \hat{M}$. 

\begin{itemize}
\item 
Let $\res c M \trans{\alpha} \hat{M}$, for $\alpha \in \{ \tau , \sigma , a \}$. 
In this case, this transition has been derived by an application of rule 
\rulename{Res} because $M \trans{\alpha} M'$, with $\hat{M} = \res c {M'}$. 
As $M \approx N$ there is $N'$ such that $N \Trans{\alpha} N'$ and $M' \approx N'$. By several applications of rule \rulename{Res} we can derive 
$\res c N \Trans{\alpha} \hat{N} = \res c {N'}$, with $(\hat{M} , \hat{N}) \in \rel$. 
\item 
Let $\res c M \trans{\alpha} \hat{M}$, for $\alpha \in \{ \sendobs{d}{v}{k}, 
\recobs{d}{v}{k} \}$, with  $d \neq c$. This case is similar
to the previous one  except for the fact that we need to pass through 
the definitions of rules \rulename{SndObs} and \rulename{RcvObs} as the
rule \rulename{Res} is only defined for intensional actions. 

\item Let $\res c M \trans{\alpha} \hat{M}$, 
for $\alpha \in \{ \sendobs{c}{v}{k}, 
\recobs{c}{v}{k} \}$. This case is not admissible as rule \rulename{Res}
block intensional actions of the form $\send{c}{v}{h}$ and $\rec{c}{v}{h}$. 

\end{itemize}

\emph{Let us prove that $\approx$ is preserved by sensor updating.}
We prove that, if $M\approx  N$ then, for all  sensors
$s$, locations $h$, and  values $v$ in the domain of $s$,  
$M[s@h \mapsto v]  \approx  N[s@h\mapsto v]$. 

Let $\succ$   be the  well-founded relation  over pairs of  networks  
such that $(M,N) \succ (M',N')$ if and only  if (i) $M \approx N$; 
(ii) $M' \approx N'$; (iii) 
$\ri{M} + \ri{N} > \ri{M'} + \ri{N'}$ (see Def.~\ref{def:redi} in the Appendix).
Note that $\succ$ is trivially irreflexive. Moreover, by Lem.~\ref{lem:wt2}, 
 $\ri{M}$ always return a  finite and positive integer, for any $M$. Thus, the relation  $\succ$ does not have 
infinite descending chains and it is a well-founded relation. 
The  proof is by well-founded induction on the relation $\succ$.

\textbf{Base case.} 
Let $M$ and $N$ be such that $M\approx N$  and 
$\ri{M} + \ri{ N} =0$. By Def.~\ref{def:redi} and by inspection of 
the reduction semantics in Tab.~\ref{reduction}, from $\ri{N}=0$ we derive 
$N \not \redi $. In particular, $N \not \redtau $.  By Thm.~\ref{thm:harmony}
it follwos that $N \ntrans{\tau}$. 
%%By Prop.~\ref{prop:patience} (Patience) there is $M'$ such 
%%that $M \trans{\sigma} M'$. 
 By an application of rule \rulename{SensEnv}, we have   
$M\trans{\rsensa{v}{s@h}}M[{s@h}\mapsto v]$. As $M\approx N$  there are $N_1$, $N_2$ and
$N'$ 
such that  $N \Trans{} N_1 \trans{\rsensa{v}{s@h}} N_1[{s@h}\mapsto v] =N_2 \Trans{} N'$ with $M[{s@h}\mapsto v] \approx N'$. However, 
as $N \ntrans{\tau}$ it follows that $N = N_1$. By Lem.~\ref{lem:update-sens} 
if $N_1 \ntrans{\tau}$ then $N_1[{s@h}\mapsto v] \ntrans{\tau}$. This is 
enough to derive that $N' = N [{s@h}\mapsto v]$. And hence
$M [{s@h}\mapsto v] \approx N [{s@h}\mapsto v]$.

\textbf{Inductive Case.} 
Let $M \approx N$. Without loss of generality 
we can assume $\ri{ M} \geq \ri{N}$.
Let $M\trans{\rsensa{v}{s@h}}M[{s@h}\mapsto v]$.
Since $M\approx N$ there is $N'$ such that $N\Trans{\rsensa v {s@h}}N'$ and  $M[{s@h}\mapsto v]\approx N'$. 

If the number of $\tau$-actions inside the weak transition  $N \Trans{\rsensa v {s@h}} N'$ is $0$, then  $N\trans{\rsensa v {s@h}}N[{s@h}\mapsto v]$ with $M[{s@h}\mapsto v] \approx N[{s@h}\mapsto v]$, and there is nothing else to prove.

Otherwise, we have  $ \ri{N'} < \ri{N}$. 
By Lem.~\ref{lem:update-sens} we have $\ri{M}=\ri{M[{s@h}\mapsto v] }$. Thus, it follows that 
$\ri{M[{s@h}\mapsto v] } + \ri{N'} < \ri{M} + \ri{N}$. 
Hence $ (M[{s@h}\mapsto v], N') \prec (M, N) $. 
By inductive hypothesis we know that $M[{s@h}\mapsto v][{r@k} \mapsto w] \approx  N'[{r@k} \mapsto w]$ for any sensor $r$, location $k$ and value $w$ in the domain of $r$. Thus, if we choose $r=s$, $k=h$ and  $w$ the value such that
$M[{s@h}\mapsto v][{s@h} \mapsto w]= M [{s@h} \mapsto w] =M$\footnote{By Def.~\ref{def:sens-change} the value $w$ must be the value of the sensor $s$ located at $h$ in $M$ if  defined. Otherwise it can be any admissible value for $s$.}, then  
 we get  $M \approx 
 N'[{s@h} \mapsto w]$. 

From this we derive the two following facts:
\begin{itemize}
\item  

By Lem.~\ref{lem:update-sens} we have $\ri{N'[{s@h}\mapsto w] } = \ri{N'}$. 
Since $\ri{N'} < \ri{N}$ it follows that  $( M, N'[{s@h}\mapsto w])
\prec (M, N)$. By inductive hypothesis we can close under the operator $[{s@h} \mapsto v]$, getting 
 $M[{s@h}\mapsto v] \approx  N'[{s@h}\mapsto w] [{s@h}\mapsto v]$. 

\item
Since $\approx$ is a transitive relation, $M \approx 
 N'[{s@h} \mapsto w]$    and $M\approx N$,   we derive that $ N \approx  N'[{s@h}\mapsto w]$.
Since, $\ri{ N} \leq \ri{M}$ (this was an initial assumption) and $\ri{N'[{s@h}\mapsto w] } = \ri{N'} < \ri{N}$ 
it follows that 
$  (N, N'[{s@h}\mapsto w]) \prec (M,N)$. By inductive hypothesis we
can derive 
$  N[{s@h}\mapsto v] \approx N'[{s@h}\mapsto w] [{s@h}\mapsto v] $.
\end{itemize}
From this two facts, by transitivity of $\approx$ we get 
$M[{s@h}\mapsto v] \approx N[{s@h}\mapsto v]$.
\hfill\qed

%%%%%%%%%%%%%%%%%%%%%%%%%%%%%%%%%%%%%%%%%%%%%%%%%%%%%%%%%%%%%%%%%%%%%%

%%\section{Appendix B: Decontextualization Lemmata}

\begin{lemma}
\label{lem:aux-cut}
Let  $O=\nodep{n}{\conf \I \nil}{\stat}{k}$  
for an arbitrary node name $n$, an arbitrary actuator $a$, and an arbitrary 
value $v$ in the domain of $a$,  such that 
$\I$ is only defined for $a$ and $\I(a)=v$. If $M | O \cong N | O$ then 
$M \cong N$. 
\end{lemma}
\begin{proof}
We recall that we always work with well-formed systems. The proofs 
consists in showing tha the relation 
\[ 
\rel=  \{( M,  N)\, : \, M|O \cong N|O  
\textrm{, for some $O$ defined as above} \}.
\]
  is  barb preserving, reduction closed and 
contextual. Since $\cong$ is the largest relation satisfying these properties, then $  \rel \subseteq \, \cong$ and therefore $M\cong N$. The scheme 
of the proof is very similar to that of the following proof.
\end{proof}

\paragraph{\textbf{Proof of Lem.~\ref{lem:cut_cxt}.}}
Let $O=\nodep{n}{\conf{\mathcal I}{\wact v a}.\nil}{\stat}{k}$, 
for an arbitrary node name $n$, an arbitrary actuator $a$, and arbitrary values $v$ and $w$, in the domain of $a$,  such that 
$\I$ is only defined for $a$ and $\I(a)=w\neq v$. 
Let use define the relation 
\[ 
\rel=  \{( M,  N)\, : \, M|O \cong N|O  
\textrm{, for some $O$ defined as above} \} \enspace . 
\]
 We show that the relation $\rel \, \cup  \cong$ is  barb preserving, reduction closed and 
contextual. Since $\cong$ is the largest relation satisfying these properties, then $  \rel \subseteq \, \cong$ and therefore $M\cong N$. 

 We recall that in this paper we only consider  well-formed networks. 
So, in the definition of $\rel$ we assume that all systems of the form 
$M|O$ and $N|O$ are well-formed. In particular, 
in order to decide whether $(M,N) \in \rel$ it is enough 
to find an $O$  of the indicated shape,  which respects the requirements of $\rel$, and which preserves well-formedness.

\emph{Let us prove that $\rel \, \cup \cong$ is barb-preserving\/}.  
We concentrate on the relation $\rel$.
As  $O$ has neither channels or sensors it is  basically 
isolated from the rest of the world, except for signals emitted on the 
 actuator $a$.  So, it is very easy to show
 that $ \rel$ is   barb preserving from $M | O \cong N|O$.

\emph{Let us  prove that $\rel \, \cup \cong$ is reduction closed\/}. We focus on $\rel$. Recall that $\red {\deff} \redtau \cup \redtime$.

Let $(M,N)\in\rel$ and  $M\redtau M'$, for some $M'$. 
We have to show that $N\redmany N'$, for some $N'$ such that 
 $(M',N')\in\rel \, \cup \cong$.  
%Let us consider the case $ M | O \cong    N | O $ and    $M\redtau M'$. 
Let us fix an $O$ which respects the requirements of $\rel$. 
By an application of  rule \rulename{parn} we infer $M | O \redtau M' | O$.
As $M | O \cong N | O$ there is $\overline{N}$ such that  $N | O \redmany \overline{N}$ and $M' | O \cong \overline{N}$.
Since $O$ cannot communicate and since the only enabled reduction for  $O$  is $\red_a$,  none of the reductions in the reduction sequence $N | O \redmany \overline{N}$ involves $O$ and none of these reductions is a timed one. 
Therefore,  $\overline{N} =  N'|O$,  $N \redtau^{\ast}  N'$, and $M' |O \cong N'|O$. 
This  implies $(M',N')\in\rel$.

Let $(M,N)\in\rel$ and  $M\red_b M'$, for some $M'$. 
As both systems $M|O$ and $N|O$ are well-formed,  the actuator $a$ cannot appear neither in 
$M$ or in $N$.  Thus, $a\neq b$. Starting from  $M | O \cong N | O$ 
we reason as in the previous case. 

Let $(M,N)\in\rel$ and  $M\redtime M'$, for some $M'$. 
We have to show that $N\redmany N''$, for some $N''$ such that 
 $(M',N'')\in\rel \, \cup \cong$. By definition of $\rel$ we have $M | O \cong N | O$. 
Let $M | O \red_a M | 
\nodep{n}{\conf {\I[a \mapsto v]} \nil}{\stat}{k}$, by an application 
of rules \rulename{actchg} and \rulename{parn}. As $\cong$ is 
reduction closed it follows that there is $\overline{N}$ such 
that $N | O \redmany \red_a \redmany \overline{N}$, with 
$M | \nodep{n}{\conf {\I[a \mapsto v]} \nil}{\stat}{k} \cong \overline{N}$. 
Due to the structure of $O$ the last reduction sequence can 
be decomposed as follows: $N | O \redmany \red_a \redmany \overline{N} =
N' | \nodep{n}{\conf {\I[a \mapsto v]} \nil}{\stat}{k}$, for some $N'$ such 
that $N \redmany N'$. Thus, for $O'= \nodep{n}{\conf {\I[a \mapsto v]} \nil}{\stat}{k}$, we have $M | O' \cong N' | O'$. 
 Since $M \redtime M'$,  by Prop.~\ref{prop:maxprog},   there is 
no $M''$ such that $M \redtau M''$. More generally, by looking at the definition 
of $O'$ it is 
easy to see that there is no $U$ such that $M | O' \redtau U$. Thus, 
by an application of  rules \rulename{timestat} and  \rulename{timepar} we can infer $M | O'
 \redtime M' | O'$.
As $M | O' \cong N' | O'$, by Prop.~\ref{prop:time-observation} 
there is $\hat{N}$ such that  $N' | O' \redtau^{\ast}\redtime \redtau^{\ast} \hat{N}$ and $M' | O' \cong \hat{N}$. By looking at the definition of 
$O'$ the only possibility is that $\hat{N} = N'' | O'$, with $N' \redtau^{\ast}\redtime \redtau^{\ast} N''$ and $M' |O' \cong N''|O'$. By Lem.\ref{lem:aux-cut}
this  implies $M' \cong N''$. Recapitulating we have that 
for $M \redtime M'$ there is $N''$ such that $N \redmany N''$, with $(M', N'') \in \rel \, \cup \cong$.

\emph{Let us prove that $\rel \, \cup \cong$ is contextual\/}. Again, it 
is enough to focus on $\rel$. 
%%It follows from 
%%the definitions of structural congruence and  sensor updating. 
Let us
consider the three different network contexts:
\begin{itemize}
\item Let $(M,N)\in \rel$. Let $O'$ be  an arbitrary network such that both $M | O'$ and $N| O'$
are well formed. We want to show that $(M|O' \, , \, N|O') \in \rel$. 
As $(M,N)\in \rel$, we can always find an $O=\nodep{n}{\conf{\mathcal I}{\wact v a}.\nil}{\stat}{k}$ which respects the requirements of $\rel$ such that  $M | O \cong 
N | O$ and both systems  $M | O | O' $ and $N| O |O'$ are well-formed. 
As $\cong$ is contextual and structural congruence is a monoid with respect to 
parallel composition, it follows that $(M | O') | O \equiv (M | O) | O' \cong (N | O) | O'
\equiv (N | O') | O$. As $\equiv \, \subset \, \cong$ and $\cong$ is trivially 
transitive,  this is enough to derive that $(M|O' \, , \, N|O') \in \rel$.

\item Let $(M,N)\in \rel$. Let $c$ be  an arbitrary channel name.
Let $O=\nodep{n}{\conf{\mathcal I}{\wact v a}.\nil}{\stat}{k}$ which respects the requirements of $\rel$. As $\cong$ is contextual if follows that 
$\res c (M | O) \cong \res c (N | O)$. 
Since  $O$ does not contain channels it holds that $
(\res c M )|O \equiv \res c (M | O) \cong \res c (N | O) \equiv
(\res c N) | O$. As $\equiv \, \subset \, \cong$ and $\cong$ is trivially 
transitive, this is enough to derive that $(\res c M , \res c N) \in \rel$. 

\item Let $(M,N)\in \rel$. Let $O=\nodep{n}{\conf{\mathcal I}{\wact v a}.\nil}{\stat}{k}$ which respects the requirements of $\rel$.
Since  $O$ does not contain sensors, 
by Def.~\ref{def:sens-change} we have: 
$ M[{s@h} \mapsto v]  |O  =  (M | O) [{s@h} \mapsto v]  \cong 
  (N | O) [{s@h} \mapsto v] =  N[{s@h} \mapsto v]  |O $. This 
is enough to derive that $(M[{s@h} \mapsto v], N[{s@h} \mapsto v]) \in \rel$.
\end{itemize}
\hfill\qed

\end{document}